\documentclass{article}
\usepackage{hhline}
\usepackage[table]{xcolor}
\usepackage{graphicx} 
\usepackage{geometry}
\usepackage[utf8]{inputenc}
\usepackage[T1]{fontenc}
\usepackage{setspace}
\usepackage{amsmath}
\usepackage{graphicx}
\usepackage{float}
\usepackage{microtype}
\usepackage{amssymb}
\usepackage{amsthm}
\usepackage{physics}

\usepackage{pdflscape}
\usepackage{tikz}
\usepackage{hyperref}
\usepackage{pgf-pie}
\usepackage{graphicx}
\usepackage{pdfpages}
\usepackage{color}
\usepackage{subcaption}
\usepackage[]{algorithm2e}
\usepackage{bbold}
\usepackage{mathtools}
\usepackage{enumerate}
\usepackage{tcolorbox}
\usepackage{adjustbox}
\usepackage{floatflt}
\usepackage{comment}
\usepackage{svg}
\usepackage[T1]{fontenc}
\usepackage{hyperref}
\usepackage{authblk}
\usepackage{enumitem}
\usepackage[nottoc, numbib]{tocbibind}
\usepackage{pdfpages}

\usepackage[
    backend=biber,
    style=nature,
    sorting=nyt
]{biblatex}
\newtheorem{mydef}{Definition}[section]
\newtheorem{mytheorem}{Theorem}[section]

\newtheorem{mycor}{Corollary}[section]
\newtheorem{myconj}{Conjecture}[section]

\newcommand{\denop}{\mathcal{D}}
\newcommand{\hs}{\mathcal{H}}
\newcommand{\linop}{\mathcal{L}}
\DeclareMathOperator{\myTr}{Tr} 
\DeclareMathOperator{\pTrA}{Tr_A}
\DeclareMathOperator{\pTrB}{Tr_B}

\newcommand{\M}{\mathcal{M}}
\newcommand{\sumkl}{\sum_{k,l = 0}^{d-1}}

\newcommand{\Pkl}{P_{k,l}}
\newcommand{\Okl}{\Omega_{k,l}}

\newcommand{\E}{\mathcal{E}}
\newcommand{\Ker}{\mathcal{K}}
\newcommand{\ckl}{c_{k,l}}

\newcommand{\w}{\omega}
\newcommand{\id}{\mathbb{1}}
\newcommand{\diff}{\mathrm{d}}

\newcommand{\PPT}{\mathrm{PPT}}
\newcommand{\SEP}{\mathrm{SEP}}
\newcommand{\BE}{\mathrm{BE}}

\newcommand{\ENT}{\mathrm{ENT}}

\title{BIPARTITE BOUND ENTANGLEMENT}

\author[1]{Beatrix C. Hiesmayr}
\author[2]{Christopher Popp}
\author[3]{Tobias C. Sutter}
\affil[ ]{University of Vienna, Faculty of Physics, Währingerstrasse 17, 1090 Vienna.\vspace{3.5mm}}

\affil[1]{Beatrix.Hiesmayr@univie.ac.at}
\affil[2]{Christopher.Popp@univie.ac.at}
\affil[3]{Tobias.Christoph.Sutter@univie.ac.at}
\date{}

\begin{document}

\maketitle

\begin{abstract}
Bound entanglement is a special form of quantum entanglement that cannot be
used for distillation, i.e., the local transformation of copies of
arbitrarily entangled states into a smaller number of approximately
maximally entangled states. Implying an inherent irreversibility of
quantum resources, this phenomenon highlights the gaps in our current
theory of entanglement.
This review provides a comprehensive exploration of the key
findings on bipartite bound entanglement.
We focus on systems of
finite dimensions, an area of high relevance for many quantum
information processing tasks.
We elucidate the properties of bound entanglement and its
interconnections with various facets of quantum information theory and
quantum information processing.
The article illuminates areas where our
understanding of bound entangled states, particularly their detection
and characterization, is yet to be fully developed.
By highlighting the need for further research into this phenomenon and
underscoring relevant open questions, this article invites researchers
to unravel its relevance for our understanding of entanglement in
Nature and how this resource can most effectively be used for
applications in quantum technology.
\end{abstract}

\tableofcontents

\newpage

\section{Introduction}

According to Erwin Schrödinger \cite{schrodinger_discussion_1935}, entanglement is not only ``\textit{one} but rather \textit{the} characteristic trait of quantum mechanics, the one that enforces its entire departure from classical lines of thought.''
Almost 90 years after this statement was made, we still lack a full physical and mathematical understanding of this non-classical phenomenon.
Today, one of the main challenges of modern quantum information theory lies in determining whether a composite system is separable or entangled, a problem known as the \textit{separability problem}~\cite{horodecki_quantum_2009}.
The problem is solved only for qubit-qubit and qubit-qutrit systems since, for these dimensions, there exists a necessary and sufficient criterion to distinguish between separability and entanglement.
For general higher-dimensional bipartite systems, $d>2$, the separability problem is known to be NP-hard~ \cite{gurvits_classical_2003}, i.e., no polynomial-time algorithm is known that is capable of solving it deterministically.
A curious form of entanglement arises in those higher dimensions, so-called \textit{bound entanglement}.
It is extraordinary as it requires maximally entangled quantum states to generate it, yet this process is irreversible.
Hence, once a bound entangled state is produced, one can no longer extract the consumed resources in the form of pure entanglement.
This is in stark contrast to \textit{free entanglement} for which this extraction procedure, in the following named entanglement distillation, is partly possible.

Let us briefly mention how the idea of labeling this new type of entanglement with the word ``bound'' came about by quoting from Ref.~\cite{horodecki_mixed-state_1998}: 
\begin{quote}
\textit{As a matter of fact, we have revealed a kind of entanglement that cannot be used for sending reliably quantum information via teleportation. Using an analogy with thermodynamics, we can consider entanglement as a counterpart of energy, and sending of quantum information as a kind of ``informational work''. Consequently, we can consider ``\textbf{free entanglement}'' ($E_{free}$)  which can be distilled, and ``\textbf{bound entanglement}'' ($E_{bound}$). In particular, the free entanglement is naturally identified with distillable entanglement as the latter asks us how many qubits can we reliably teleport via the mixed state. This kind of entanglement can always be converted via distillation protocol to the ``active'' singlet form.\\
To complete the analogy, one could consider the asymptotic number of singlets which are needed to produce a given mixed state as ``internal entanglement'' $E_{int}$ (the counterpart of internal energy). Then the bound entanglement can be quantitatively defined by the following equation:
\begin{align*}
    E_{int}=E_{free}+E_{bound}\;.
\end{align*}}
\end{quote}
So the initial association was with thermodynamics (cf. also Ref.~\cite{brandao_entanglement_2008, horodecki_are_2002}) and especially the second law of thermodynamics, connecting bound entanglement to the arrow of time.
The latter connection was taken by an identification with the mathematical property PPT (positive under partial transposition, Def.~\ref{def:ppt-property}) of the discovered bound entangled state.
However, this argument has not proved successful since the (partial) transpose has no one-to-one correspondence to the time arrow of thermodynamics.
And yet, the PPT property is central to our understanding of bound entangled states because it remains the only sufficient criterion for undistillability known until today.

Since the discovery of an exemplary state by the Horodecki family at the end of the last millennium~\cite{horodecki_separability_1997, horodecki_mixed-state_1998}, many researchers have put effort into the study of bound entanglement.
However, a full characterization of the set of bound entangled states or merely their construction is still lacking.
On the other hand, it is shown that they can violate a Bell inequality, i.e., exhibit non-classical correlations.
Recent research suggests that, for specific dimensions, the volume of such bound entangled states in Hilbert space is of nonzero measure, and hence it cannot be neglected in general.
These results suggest many similarities, like non-locality, but also distinctive operational differences, like distillability, between bound and free entanglement.
Over the past thirty years, various quantum technologies have been devised that require entangled quantum states for their realization.
It turns out, however, that for several of these applications, bound entangled states do not provide an additional resource compared to separable ones.
However, some applications exist for which bound entanglement can improve the performance when added to available free entangled states.
Current research aims to determine how bound entangled states' non-classical features can be further employed in specialized quantum information applications.
All of this shows that a full understanding of the theory of entanglement is still not achieved, and with that, an answer to the question of how entanglement can be utilized as a resource for advanced quantum technologies such as, e.g., quantum computing, quantum machine learning, quantum cryptography, or quantum communication.

This review provides a comprehensive summary of the current status of knowledge of \textit{finite-dimensional} and \textit{bipartite} bound entanglement, which is of particular interest for a wide range of quantum technologies.
The bipartite scenario also offers the minimal system size for which bound entanglement appears.
The main focus is on contemporary research and the many open questions to be answered.
To this end, this article gives a solid foundation that can pave the way for future developments in the theory of bound entanglement.

The review article is designed so that it can be read in a different order than provided.
However, the order is chosen so readers without prior knowledge can follow the arguments.
In particular, Chapter~\ref{sec:basics} introduces the necessary mathematical concepts and notations to define bound entanglement clearly.
The notion of entanglement distillation, central to the definition of bound entanglement, is explained in detail. 
It further defines Bell-diagonal states, a subset of quantum states with high relevance in many entanglement distillation schemes and a comparatively high volume of bound entanglement.
Equipped with the basic mathematical tools, Chapter~\ref{sec:free_vs_bnd} introduces the PPT criterion and defines bound entanglement as entangled states that cannot be distilled.
We establish the crucial connection between the PPT property and bound entangled states by showing that all distillable states violate the PPT criterion.
Chapter~\ref{sec:detection of BE} reviews the main approaches to detect bound entangled states and differentiate them from free entangled and separable states.
Multiple methods are discussed, and it is emphasized that there is no strict hierarchy between them, i.e., no entanglement criterion is strictly stronger than any other in terms of bound entangled states it detects.
The logical next step following the detection of some bound entangled states is to investigate their characteristics, which is the focus of Chapter~\ref{charOfBE}.
Properties of the set of bound entangled states and individual bound entangled states are discussed.
Furthermore, the limitations and usefulness of bound entanglement for quantum information processing tasks are summarized.
We also introduce various entanglement measures and point out their hierarchical dependence and the resulting consequences.
In this regard, the difficulty of finding a good entanglement measure is connected to the existence of bound entangled states.
Concluding the characteristics of bound entanglement, we argue that bound entanglement can be Bell-nonlocal and is not Lorentz invariant.
In addition to the fundamental difficulties of detecting, characterizing, and classifying bound entanglement, Chapter~\ref{BEexperiment} deals with how this curious phenomenon was uncovered in experiments.
In particular, we present experimental realizations using entanglement witnesses based on SICs and MUBs, which are formally introduced in Sec.~\ref{POVMWitnesses}.
Even though we can detect and characterize bound entangled states to some degree, it is still not fully understood how to construct the set of bound entangled states.
The dual problem is finding a complete set of non-decomposable entanglement witnesses (cf. Sec.~\ref{sec:decomposible}).
Some approaches to constructing bound entangled states are presented in Chapter~\ref{sec:construction}, focusing on the analytic construction via unextendible product bases.
The article concludes in Chapter~\ref{sec:outlook} with a summary of fundamental questions that have been answered and a broad outlook on open questions.
This includes the problem of why Nature provides us with this special kind of entanglement and for what quantum technologies it can be effectively leveraged.
Also, the potential existence of NPT bound entangled states is discussed.
Solving these issues will be challenging as the detection of bound entanglement is intimately linked to the NP-hard separability problem.
At the same time, their characterization can be traced back to our incomplete understanding of entanglement in general.
Nonetheless, closing these open gaps in our knowledge will certainly lead to a more complete comprehension of one of our fundamental physical theories.

\section{Basics of Quantum States and Operations} \label{sec:basics}

This section contains the basic definitions of quantum states and how they can be manipulated by positive operator-valued measures and completely positive trace-preserving maps.
It thus lays the mathematical foundations needed for understanding bound entanglement.
Importantly, in Sec.~\ref{sec:ent_dist}, the concept of entanglement distillation is introduced, which is at the heart of the definition of bound entanglement.

\subsection{Definition of Quantum States}

The mathematical structure of quantum physics is determined by linear algebra on finite- or infinite-dimensional Hilbert spaces.
In the following, we restrict ourselves to Hilbert spaces $\hs$ of finite dimension $d>1$.
For clarity, we will sometimes choose $\hs=\mathbb{C}^d$ when $\dim(\hs)=d$.
We denote the set of linear operators mapping $\hs$ into itself by $\linop(\hs)$, in detail:

\begin{mydef}[Quantum State] \label{def:densitymat}
    A linear operator $\rho \in \linop(\hs)$ acting on a Hilbert space $\hs$ is called quantum state or density operator if it is positive semi-definite \emph{($\rho\geq 0$)} and has unit trace \emph{($\myTr(\rho)=1$)}.\\
    The set of all quantum states on $\hs$ is denoted $\denop(\hs) \subset \linop(\hs)$.
\end{mydef}
Note that the condition $\rho\geq 0$ implies hermiticity, i.e., $\rho=\rho^\dagger$, with $\rho^\dagger$ being the hermitian conjugate of $\rho$.
For general Hilbert space dimension $\dim(\hs)=d$, the quantum system is also called \emph{qudit}.
In the cases $d=2$, $d=3$, and $d=4$, they are known as a qubit, qutrit, and ququart, respectively.
A quantum state $\rho\in\denop(\hs)$ is called \emph{pure} if $\mathrm{Tr}(\rho^2)=1$, and \emph{mixed} if $\mathrm{Tr}(\rho^2)<1$.
Due to the spectral theorem (cf. Ref.~\cite{nielsen_quantum_2012}) we can write any $\rho\in\denop(\hs)$ in its eigenbasis as
\begin{align} \label{eq:rho_spectral_decomp}
    \rho = \sum_{i=0}^{d-1} \lambda_i\; |\psi_i\rangle \langle\psi_i| \;,
\end{align}
where $|\psi_i\rangle$ are eigenvectors of $\rho$ with respective eigenvalues $\lambda_i$, and the set $\{|\psi_i\rangle\} _{i=0}^{d-1}$ of eigenvectors forms an orthonormal basis of $\hs$.
Consequently, any pure state can be written as $\rho=|\psi\rangle\langle\psi|$, and thus directly associated with the corresponding Hilbert space element $|\psi\rangle \in \hs$.
We call $\rho$ a rank-$r$ state if it has $r$ nonzero eigenvalues.

\subsubsection{Bipartite Quantum States, Separability and Entanglement}

Bipartite quantum states are identified with linear operators acting on the tensor product $\hs_A\otimes\hs_B$ of two Hilbert spaces.
Such density matrices can be expanded in the computational product basis as
\begin{align}
    \rho_{AB} = \sum_{i,j,k,l = 0}^{d-1} \rho_{ik,jl} \; |i \rangle \langle j| \otimes | k \rangle \langle l| \;.
\end{align}
In this case, the \textit{partial trace} of only one subsystem is defined by linearity as
\begin{align}
    \pTrA (\rho_{AB}) := (\myTr\otimes \id_B)(\rho_{AB}) = \sum_{i,k,l = 0}^{d-1} \rho_{ik,il} \; | k \rangle \langle l| \;,
\end{align}
where $\id_B$ denotes the identity map on $\linop(\hs_B)$\footnote{
This means $\id_B$ should be seen as a linear map from $\linop(\hs_B) \rightarrow \linop(\hs_B)$.
Below, we also denote the identity matrix by $\id_B \in \linop(\hs_B)$.
The context unambiguously establishes the meaning for each appearance.}.
This operation results in a state on $\denop(\hs_B)$ called the \emph{reduced state} with respect to subsystem $B$.
The operation $\pTrB$ is defined analogously.

For bipartite systems, non-classical features can arise, such as quantum entanglement (an extensive review of this broad topic can be found in Ref.~\cite{horodecki_quantum_2009}).
We only reproduce the facts necessary for understanding bound entanglement, formally introduced in Sec.~\ref{sec:free_vs_bnd}.

\begin{mydef}[Separability and Entanglement]\label{def:sep_and_ent}
A bipartite quantum state $\rho_{AB}\in \denop(\hs_A\otimes\hs_B)$ is called separable if it can be written as a convex combination of product states, i.e., $\rho_{AB} = \sum_i \, p_i \, \rho_A^i \otimes \rho_B^i$ with probabilities $p_i \geq 0$, $\sum_i p_i =1$, $\rho_A^i\in \denop(\hs_A)$, and $\rho_B^i\in \denop(\hs_B)$.
If $\rho_{AB}$ is not separable, it is called inseparable or entangled.
We denote the set of all separable (entangled) states by \emph{SEP (ENT)}.
\end{mydef}

\begin{figure}
    \centering
    \includegraphics[width=0.5\linewidth]{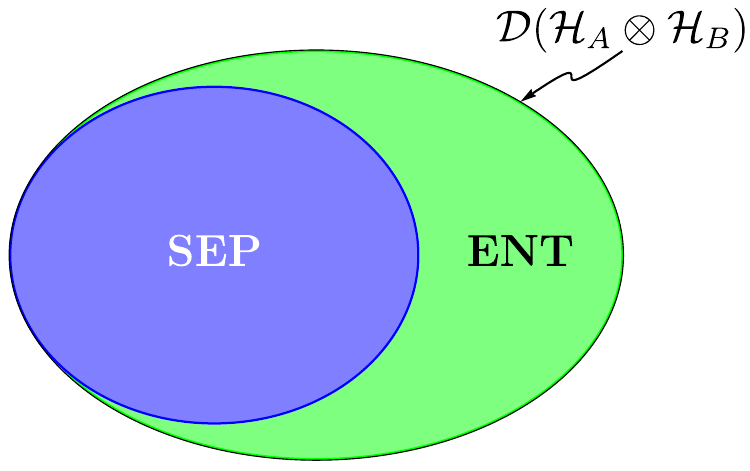}
    \caption{Partition of bipartite state space $\denop(\hs_A\otimes\hs_B)$ into the sets SEP (blue) and ENT (green).
    The diagram does not faithfully represent the relative volumes of SEP and ENT, which is known to depend on the dimensions of $\hs_A$ and $\hs_B$.}
    \label{fig:state_space_partition}
\end{figure}

Note that the representation of a separable state as a mixture of product states is generally not unique.
Furthermore, the above definition induces a partition of the set of density operators $\denop(\hs_A\otimes\hs_B)$.
This means it splits up the set into two proper disjoint subsets $\mathrm{SEP}$ and $\mathrm{ENT}$ satisfying $\mathrm{SEP}\neq\varnothing$, $\mathrm{ENT}\neq\varnothing$, $\mathrm{SEP}\cup \mathrm{ENT} \equiv \denop(\hs_A\otimes\hs_B)$, and $\mathrm{SEP}\cap \mathrm{ENT} \equiv \varnothing$.
Fig.~\ref{fig:state_space_partition} visualizes this partition.

\begin{mydef}[Maximally Entangled States, Locally Maximally Mixed States] \label{def:maximally_entangled}
    A pure bipartite state $|\psi\rangle\in \hs_A \otimes \hs_B$ with $\dim(\hs_A)=\dim(\hs_B) = d$ is said to be maximally entangled if the partial trace on either system yields the maximally mixed state, i.e., $\pTrA(|\psi\rangle\langle\psi|)= d^{-1} \, \id_B$ and $\pTrB(|\psi\rangle\langle\psi|)= d^{-1} \, \id_A$.
    Any bipartite quantum state with the latter property is called locally maximally mixed.
\end{mydef}
This definition is equivalent to $|\psi\rangle$ having full Schmidt rank with Schmidt coefficients $1/\sqrt{d}$.
One important maximally entangled state discussed further in Sec.~\ref{sec:bell_system} is given by
\begin{align}\label{eq:max_ent_state_omega_00}
    |\Omega_{0,0}\rangle = \frac{1}{\sqrt{d}}\sum_{i=0}^{d-1} |i\rangle \otimes | i\rangle \in \mathbb{C}^d \otimes \mathbb{C}^d \;.
\end{align}

Determining whether a given quantum state is separable or entangled is a natural problem of interest that can be formalized as follows:
\begin{mydef}[Separability Problem]\label{separabilityproblem}
    Given $\rho_{AB}\in \denop(\hs_A\otimes\hs_B)$, decide whether or not $\rho_{AB}\in\mathrm{SEP}$.
\end{mydef}
This is a well-posed question because SEP, together with ENT, forms a partition of the state space.
Generally, the separability problem is known to be NP-hard \cite{gurvits_classical_2003}.
This means that there is no known polynomial time algorithm that deterministically decides whether or not $\rho_{AB}\in\mathrm{SEP}$ for all $\rho_{AB}\in\denop(\hs_A\otimes\hs_B)$ and arbitrary Hilbert space dimensions $\dim(\hs_A)$ and $\dim(\hs_B)$ unless the complexity classes P and NP are equal (they are commonly assumed not to be \cite{arora_computational_2016}).

\subsection{Quantum Operations}

Any general operation that modifies a quantum state should naturally result in another quantum state.
This motivates the concepts of completely positive trace-preserving maps and positive operator-valued measures.

\subsubsection{Positive and Completely Positive Maps}\label{sec:PositiveCP}

The following definition captures the central object of interest in this section \cite{watrous_theory_2018}:

\begin{mydef}[Completely Positive (CP) Map] \label{def:P_and_CP}
    A linear map
\begin{align}
    \mathcal{E}:\linop(\hs_A) \rightarrow \linop(\hs_{A'})
\end{align}
is called positive if $\rho_A \geq 0$ implies $\mathcal{E}(\rho_A)\geq 0$ for any $\rho_A\in \linop(\hs_A)$.
Additionally, $\mathcal{E}$ is called completely positive~(CP) if $\mathcal{E}\otimes\mathrm{id}_B :\linop(\hs_A\otimes\hs_B) \rightarrow \linop(\hs_{A'}\otimes\hs_B)$ is a positive map for every Hilbert space $\hs_B$.
Here, the action of $\mathcal{E}\otimes\mathrm{id}_B$ on a composite system is defined by linearity.
If, for a positive map $\mathcal{E}$, there exists a Hilbert space $\hs_B$ such that $\mathcal{E}\otimes\mathrm{id}_B$ is not a positive map, $\mathcal{E}$ is called positive but not completely positive (PNCP).
\end{mydef}

From a physical point of view, applying a positive map to a quantum state ensures that the measurement probabilities remain non-negative.
Furthermore, complete positivity captures the intuition that the probabilities remain non-negative even if the map is only applied to part of a larger system.
It is clear that every completely positive map is also positive.

Any physically implementable map must be completely positive and trace-preserving (CPTP) to map density operators to density operators in accord with Def. \ref{def:densitymat}.
Such maps are also called \emph{quantum channel}.
To check whether a given linear map $\mathcal{E}$ is CP, it is sufficient to check the positive semi-definiteness of its Choi operator
\begin{align} \label{eq:choi_operator}
    J(\mathcal{E}):=
    \sum_{i,j=0}^{d_A-1} \mathcal{E}(|i\rangle\langle j |) \otimes |i\rangle\langle j | \; \in \linop(\hs_{A'}\otimes \hs_A) \;,
\end{align}
One can also show the reverse statement to find \cite{watrous_theory_2018}:
\begin{mytheorem}[Choi's Theorem] \label{thm:Choi_CP}
A linear map $\mathcal{E}$ is CP if and only if $J(\mathcal{E})\geq 0$.
\end{mytheorem}

Although they cannot be realized experimentally, \textit{positive but not completely positive} (PNCP) maps are valuable in theoretical considerations, in particular with regards to the separability problem (see Sec.~\ref{sec:free_vs_bnd}).

\subsubsection{Positive Operator Valued Measures}\label{sec:POVM}

Positive operator-valued measures (POVMs) generalize projection-valued measures (also known as projective or von Neumann measurements).
They describe generalized measurements on quantum states and are defined as follows \cite{bengtsson_geometry_2006}:
\begin{mydef}[Positive Operator Valued Measure]
     A positive operator-valued measure (POVM) on a Hilbert space $\hs$
     is a set of $n$ positive semi-definite operators $E_m\in \linop(\hs)$ that sum up to the identity on $L(\hs)$,
     \begin{equation} \label{eq:def_POVM}
     \sum_{m=1}^{n} E_m\;=\;\id_{\linop(\hs)}\;.
     \end{equation}
\end{mydef}
The individual $E_m$ are called \emph{POVM elements}.
Eq.~\eqref{eq:def_POVM} establishes that a POVM on $\hs$ is a partition of the identity on $\linop(\hs)$.
The probability of obtaining outcome $m$ after conducting a POVM measurement on a quantum state $\rho\in\denop(\hs)$ is given by
\begin{align}
    p_m = \myTr(E_m \,\rho) \;.
\end{align}
That the probabilities sum up to one, i.e., $\sum_{m=1}^n p_m = 1$, is ensured by \eqref{eq:def_POVM}.

Projective measurements are special cases of POVMs for which $E_m ^2 = E_m$ for all $m=1,...,n$.
In this case, each POVM element $E_m$ is a rank-1 projector and can be written as $E_m = |\pi_m\rangle\langle\pi_m|$ for some normalized $|\pi_m\rangle \in \hs$.
The completeness relation \eqref{eq:def_POVM} enforces that the set $\{|\pi_m\rangle\}_m$ constitutes an orthonormal basis of $\hs$, i.e., we need to have exactly $n=\dim(\hs)$ POVM elements.

Alternatively, a POVM can be defined as a set of $n$ linear operators $M_m$ acting on $\linop(\hs)$ and satisfying $\sum_{m=1}^{n} M_m^\dagger M_m = \id$.
These two statements are equivalent as any positive semi-definite operator $E_m$ can be written as
\begin{align}
    E_m = M_m^\dagger M_m \;,
\end{align}
e.g., by writing $M_m = \sqrt{E_m}$ (which is unique up to an isometry).
Here, the probability for outcome $m$ is given by
\begin{align}
    p_m = \myTr(M_m^\dagger M_m\, \rho) = \myTr(M_m\, \rho\, M_m^\dagger) \;.
\end{align}
The advantage of this alternative definition is that the post-measurement state is given by
\begin{align}
    \rho_m = \frac{1}{p_m} M_m \,\rho\, M_m^\dagger \;.
\end{align}

POVMs are crucial for theoretical applications such as the detection of bound entanglement via mutually unbiased bases (MUBs) or symmetric informationally complete POVMs (SICs) (see Sec.~\ref{POVMWitnesses}).

\subsubsection{Entanglement Distillation} \label{sec:ent_dist}
The defining feature of bound entanglement is given by the notion of \textit{entanglement distillation}\footnote{Entanglement distillation is also sometimes called entanglement purification \cite{bennett_purification_1996, alber_efficient_2001, deutsch_quantum_1996}. This should not be confused with the purification of a mixed quantum state $\rho\in \denop(\hs_A)$ by a pure state $|\psi\rangle\in \hs_A\otimes\hs_B$ such that $\rho=\pTrB(|\psi\rangle\langle\psi|)$. We use the term entanglement distillation, or simply distillation, throughout.} \cite{bennett_concentrating_1996, bennett_purification_1996}.
Heuristically speaking, distillation describes the process of two parties, $A$ and $B$, converting multiple ``weakly'' entangled shared quantum states into less but ``more strongly'' entangled states using only local quantum operations and classical communication (LOCC).

More precisely, $A$ and $B$ share $n$ identical copies of some (generally mixed) bipartite resource state $\rho$,
\begin{align}
    \rho^{\otimes n} = \rho_{A_1 B_1}\otimes \rho_{A_2 B_2} \otimes ... \otimes \rho_{A_n B_n} \in \denop(\bigotimes_{i=1}^n \hs_{A_i}\otimes \hs_{B_i})\;,
\end{align}
where $A$ ($B$) is in possession of the subsystems indexed by $A_i$ ($B_i$), $i\in\{ 1,2,...,n\}$.
Here, we only consider $\dim (\hs_{A_i}) = \dim (\hs_{B_i}) = d$ for all $i$.
They are allowed to perform \emph{local} quantum operations, i.e., $A$ can implement any CPTP map $\Phi: \linop(\bigotimes_i \hs_{A_i}) \rightarrow \linop(\bigotimes_i \hs_{A_i})$ or conduct any POVM measurement on $\bigotimes_i \hs_{A_i}$, and similar for $B$.
Furthermore, they can communicate \emph{classically}, e.g., send an obtained measurement outcome to the other party, who may adjust their subsequent quantum operations accordingly.
A general distillation protocol can consist of any number of consecutive LOCC operations.
\textit{Not} allowed are joint quantum operations in the form of CPTP maps or POVMs involving $A$'s and $B$'s qudits together or any exchange of qudits. These operations are considered global and require resources unavailable in a LOCC setting.

Given $n$ copies of a resource state $\rho$, the challenge in any practical approach is to develop an optimal strategy for CPTP maps and POVM measurements.
The goal of distillation is to produce multiple copies of a maximally entangled \emph{target state}, e.g., the maximally entangled qubit state $|\Phi^+\rangle = (1/\sqrt{2})\sum_{i=0}^{1} |i\rangle\otimes |i\rangle \in \mathbb{C}^2 \otimes \mathbb{C}^2$. 
We define entanglement distillation as any LOCC protocol that transforms multiple copies of a resource state into (generally fewer) copies of this target state:
\begin{mydef}[Entanglement Distillation, Distillability, Distillation Rate]\label{def:distillation}
    Entanglement distillation describes any LOCC protocol that transforms $n$ copies of a resource state $\rho \in \denop(\hs_A \otimes \hs_B)$ to $m$ copies of the maximally entangled qubit state $|\Phi^+\rangle\langle\Phi^+|\in \denop(\hs_A \otimes \hs_B)$ with nonzero probability:
    \begin{align}
        \rho^{\otimes n} \underset{LOCC}{\longrightarrow} |\Phi^+\rangle\langle\Phi^+|^{\otimes{m}} \;.
    \end{align}
    A state $\rho$ is distillable if in the asymptotic limit $n\rightarrow\infty$, a distillation rate $r:= \frac{m}{n}$ larger than zero can be achieved via entanglement distillation. If no such operations exist, $\rho$ is called undistillable.
\end{mydef}

Other protocols producing different target states are possible and interesting for practical purposes and theoretical analyses. However, this is not a limitation of the notion of distillability. Analyzing distillability with respect to a different target state is fully equivalent to the given definition if it 
is itself distillable. In this case, multiple copies of this target state can be transformed via LOCC to some number of maximally entangled qubit states.  In this case, the distillation rate generally changes. Still, any state for which a positive distillation rate with respect to such a target state can be achieved (in the asymptotic limit) will also be distillable according to Def.~\ref{def:distillation}. 

In particular, the $d$-dimensional maximally entangled state $|\psi\rangle=|\Omega_{0,0}\rangle = (1/\sqrt{d})\sum_{i=0}^{d-1} |i\rangle\otimes |i\rangle \in \mathbb{C}^d \otimes \mathbb{C}^d$ (cf. Ref.~\cite{horodecki_reduction_1999, alber_efficient_2001, wilde_quantum_2017}) provides straightforward applications in quantum information processing tasks in $d$-dimensional state space (see, e.g., Chap. 6 in \cite{wilde_quantum_2017}) and therefore is a natural target state. Indeed, this target state can be equivalently used as the target state for distillation.
This can be shown by using the concept of majorization on the ordered vectors consisting of the squared Schmidt coefficients of the two states (the interested reader may consult Sec.~12.5 of \cite{nielsen_quantum_2012} for details). Using this method, one finds that $|\Phi^+\rangle$ can be obtained by LOCC from a single copy of $|\Omega_{0,0}\rangle$.
Hence, $|\Omega_{0,0}\rangle$ is equally suited as a target state when we are only interested in the (asymptotic) distillability of the resource state.
The only quantitative difference when considering $|\Omega_{0,0}\rangle$ instead of $|\Phi^+\rangle$ is the change in the distillation rate $m/n \rightarrow m/(n \lceil log_2(d)\rceil)$.
To see this, note that one can distill $|\Omega_{0,0}\rangle$ using $\lceil log_2(d)\rceil$ copies of $|\Phi^+\rangle$ (here $\lceil\cdot\rceil$ denotes the ceiling function).

It is clear from the definition that only entangled states are potentially distillable because LOCC operations cannot convert separable states to entangled ones \cite{horodecki_reduction_1999}.
In light of Sec.~\ref{sec:free_vs_bnd}, this can also be considered a corollary of Thm.~\ref{thm:PPT-crit} together with Thm.~\ref{thm:dist_implies_NPT}.
A formal statement about the quantification of entanglement of a quantum state via distillation can be found in Def.~\ref{def:distillable_entanglement}, introducing distillable entanglement as an entanglement measure.

Depending on whether the resource state $\rho$ is pure or mixed, there are considerable differences in the methods used to determine its distillability.
For pure states, the Schmidt coefficients completely characterize the state's distillation properties, both in the finite and in the asymptotic setting \cite{nielsen_quantum_2012}.
However, for general mixed states, the issue is more complicated and an active area of research.
It is shown in Ref.~\cite{horodecki_reduction_1999} that if the overlap of the resource state $\rho$ with the maximally entangled state $|\Omega_{0,0}\rangle$ is greater than $1/d$, the state is distillable.
In the case of mixed qubit resource states ($d=2$), the following theorem holds \cite{horodecki_inseparable_1997}:
\begin{mytheorem}[Distillability] \label{thm:qubit_distillation}
    A bipartite qubit state $\rho\in\denop(\mathbb{C}^2\otimes \mathbb{C}^2)$ is distillable if and only if $\rho\in \mathrm{ENT}$.
\end{mytheorem}
This no longer holds for qudits of Hilbert space dimension $d>2$, meaning undistillable entangled states exist. These are termed bound entangled \cite{horodecki_mixed-state_1998} and are treated more formally in Sec.~\ref{sec:free_vs_bnd}.
The theorem implies that any entangled qubit state can be used as a target state, as it can be transformed to maximally entangled qubits with some nonzero distillation rate. Using such a target state can significantly simplify analyses of distillability for a given state and motivates the following definition.
\begin{mydef}[$n$-Distillability]\label{def:n-distillability}
    A bipartite state $\rho$ is called $n$-distillable if $n$ copies of $\rho$ can be transformed via LOCC into a bipartite entangled pair of qubits. Otherwise, it is called $n$-undistillable.
\end{mydef}
Note that no restriction is placed upon the amount of entanglement of the pair of qubits. Clearly, any $n$-distillable state is also distillable. However, the following Theorem of Ref.~\cite{watrous_many_2004} shows that establishing the undistillability of a quantum state via $n$-undistillability is difficult. The reason is that one needs to check all possible numbers of copies $n$, as there exist states that are $n_0$-undistillable for some large $n_0$, but are $(n_0+1)$-distillable.
\begin{mytheorem} \label{thm:watrous_n_distillability}
    For any choice of integers $d \geq 3$ and $n\geq 1$, there exists a $d^2 \times d^2$-dimensional bipartite mixed quantum state that is distillable but not $n$-distillable.
\end{mytheorem}
Consequently, to show the undistillability of an entangled bipartite state, it is generally not sufficient to show $n$-undistillability for some finite $n$. This is one manifestation of the fact that bound entangled states are generally hard to detect and characterize (see Sec.~\ref{sec:detection of BE} and Sec.~\ref{charOfBE}).

\subsection{Bell Bases, Bell-diagonal States and the Magic Simplex} \label{sec:bell_system}
This section describes the properties of special quantum states called Bell states that are highly relevant to quantum information processing tasks. 

Consider a bipartite system $\hs_A\otimes\hs_B$ composed of two subsystems $\hs_A$ and $\hs_B$ and corresponding quantum states $\rho_{AB}\in \denop(\hs_A\otimes\hs_B)$. Let $d_A(d_B)$ be the dimension of the subsystem $\hs_A(\hs_B)$ and $d_{AB}=d_A d_B$ the dimension of the product system. If not explicitly specified otherwise, we assume $d_A=d_B = d$.
Pure, maximally entangled states on $\hs_A\otimes\hs_B$ are called \emph{Bell states}.
Their relevance stems from the fact that Bell states contain a maximum amount of entanglement that can be leveraged to process quantum information. Per definition, these states are locally maximally mixed (cf. Def.~\ref{def:maximally_entangled}) as the reduced state for any subsystem is the maximally mixed state. Hence, all information lies in the correlation between the two systems. A \emph{Bell basis} consists of $d_{AB}$ linear independent Bell states that span $\hs_A\otimes\hs_B$. 
These bases are frequently used in applications like quantum teleportation \cite{bennett_teleporting_1993}, dense coding \cite{bennett_communication_1992}, or error correction \cite{bennett_mixed_1996}. 
\emph{Bell-diagonal states} are mixtures of Bell basis states and arise naturally when errors due to imperfections in the experimental setup or other unintentional interactions are considered. The entanglement structure of Bell basis states strongly influences the properties of corresponding Bell-diagonal states. Therefore, understanding the entanglement structure of Bell states is crucial for effectively applying entanglement as a resource.

A frequently used ``standard'' Bell basis generalizing the Pauli-basis of bipartite qubits to higher dimensions shows special algebraic and geometric properties (see Sec.~\ref{sec:magic_simplex}). Due to these properties, the corresponding set of Bell-diagonal states is called ``Magic Simplex'' and its special entanglement structure, including a high relative share of bound entanglement, can be analyzed with unique approaches (see Sec.~\ref{charOfBE}). 
Alternatively, more general Bell-diagonal systems are defined (see Sec.~\ref{sec:gen_bell_system}) and compared with a focus on the frequency of bound entanglement in Sec.~\ref{boundBDS}.

\subsubsection{Standard Bell-diagonal System -- The Magic Simplex} \label{sec:magic_simplex}
A special Bell basis can be generated by the \emph{Weyl-Heisenberg} operators \cite{baumgartner_special_2007}. This generalization of the Pauli operators to general dimension $d$ is defined as follows:
\begin{mydef}[Weyl-Heisenberg Operators]
    \begin{align}
        \label{eq:weyl_ops}
        W_{k,l} := \sum_{j=0}^{d-1}\w^{j k} |j\rangle \langle j+l|,~~k,l = 0,...,d-1 \;,
    \end{align}
\end{mydef}
where here and in the following, $\w:= \exp(\frac{2 \pi i}{d})$, and indices and variables taking the values $0, 1, ..., d-1$ are elements of the ring $\mathbb{Z}_d:= \mathbb{Z}/d\mathbb{Z}$ with addition and multiplication modulo $d$. From $W_{k,l} = Z(k)X(l)$ with $Z(k) = \sum_j \w^{j k} | j \rangle \langle j  |$ and $X(l) = \sum_j |j-l\rangle \langle j | $ it is clear that $W_{k,l}$ acts as phase and shift operations depending on the phase ($k$) and shift ($l$) index. These unitary, but generally not hermitian operators obey the \emph{Weyl-relations} (with $\ast$ and $T$ denoting complex conjugation and transposition with respect to the computational basis, respectively):
\begin{gather}
    \label{eq:weyl_relations}
    \begin{aligned}
           W_{k_1,l_1}W_{k_2,l_2} &= \w^{l_1 k_2}~W_{k_1+k_2, l_1+l_2} \;, \\
           W_{k,l}^\dagger &= \w^{k l}~W_{-k, -l} = W_{k,l}^{-1} \;, \\
           W_{k,l}^\ast &= W_{-k,l} \;, \\
           W_{k,l}^T &= \w^{-k l } ~W_{k, -l} \;.      
    \end{aligned}
\end{gather}
The algebraic relations imply that the elements $\w^m\, W_{k,l}$ ($m,k,l \in \mathbb{Z}_d$) form a finite discrete group under multiplication, named \emph{Weyl-Heisenberg group}. 

A $d^2$-dimensional Bell basis $\lbrace |\Okl\rangle \rbrace $ can be constructed by local application of the Weyl-Heisenberg operators to a specified Bell state, for which we choose $| \Omega_{0,0} \rangle$:
\begin{mydef}[Bell States] \label{def:bell_states}
\begin{gather}
  \label{eq:bell_states}
  \begin{aligned}
     |\Okl\rangle &:= (W_{k,l} \otimes \id_d)|\Omega_{0,0}\rangle,~~k,l \in \mathbb{Z}_d \\
     \Pkl &:= |\Okl\rangle \langle\Okl|
  \end{aligned}
\end{gather}
\end{mydef}
The corresponding density matrices $\Pkl$, called \emph{Bell projectors}, have the reduced states $\pTrA (\Pkl) = \pTrB (\Pkl) = \frac{1}{d}\id_d$, so $| \Okl \rangle$ are indeed maximally entangled. This basis represents a direct generalization of the qubit Bell basis to higher dimensions. In this work, it is referred to as \emph{standard Bell basis} to differentiate this specific but relevant basis from other, generally unitarily inequivalent, Bell bases that exist for $d \geq 3$ (see Sec.~\ref{sec:gen_bell_system}). 
This Bell basis inherits a linear structure from the Weyl-relations \eqref{eq:weyl_relations}, as any element  $\Pkl$ corresponds to one Weyl-Heisenberg operator $W_{k,l}$. The operations $\Pkl \rightarrow (W_{i,j} \otimes \id_d) \Pkl (W_{i,j} \otimes \id_d)^\dagger=P_{k+i, l+j} $ ($i,j \in \mathbb{Z}_d$) map the set of Bell projectors onto itself. This structure allows associating states $\Pkl$ with elements of a discrete phase space $\mathbb{Z}_d^2$ and operations of the above form with translations in that phase space (see Fig.~\ref{fig:discPhaseSpace}).

\begin{figure}
    \centering
    \includegraphics[width=0.75\linewidth]{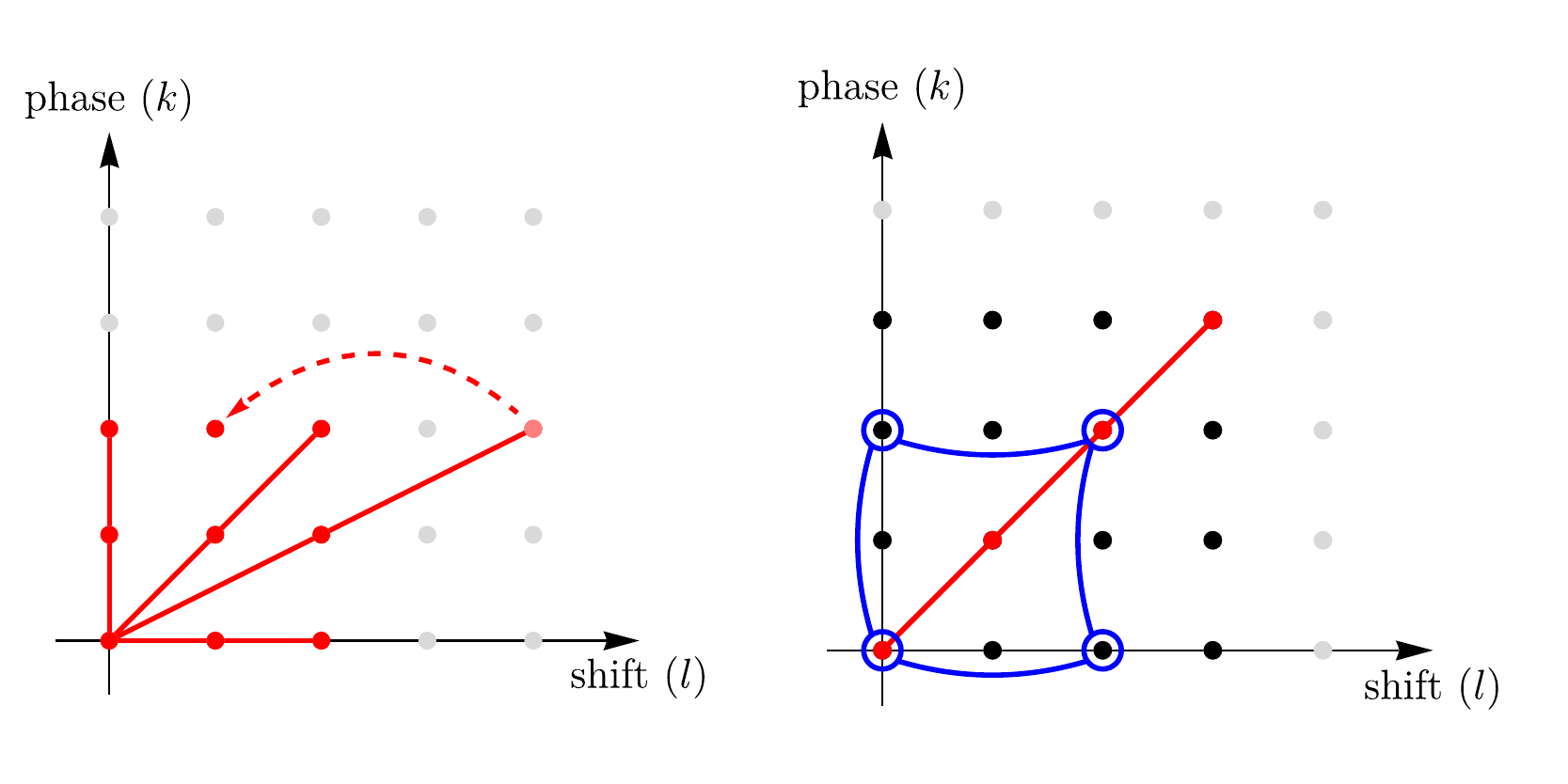}
    \caption{Visualization of the discrete phase space $\mathbb{Z}_d^2$ for $d=3$ (left) and $d=4$ (right). Each vertex represents a Bell state $P_{k,l}$. Connected vertices correspond to Bell states $P_{k,l}$ in some exemplary $d$-element subgroups. For prime $d$, subgroups form so-called ``lines'' (red), while for composite $d$, other subgroups exist (blue).}
    \label{fig:discPhaseSpace}
\end{figure}

Mixtures of standard Bell states with mixing probabilities $\lbrace \ckl \rbrace$ form the set of \emph{standard Bell-diagonal states}, the so-called \emph{Magic Simplex} $\M_d $:
\begin{mydef}[Magic Simplex] \label{def:magic_simplex}
    \begin{gather}
    \label{eq:magic_simplex}
    \M_d := \lbrace \rho = \sumkl \ckl\, \Pkl ~ |~
    \sumkl \ckl = 1, \ckl \geq 0  \rbrace
\end{gather}
\end{mydef}

Following the choice of words for the qubit Bell basis of Wootters and Hill \cite{wootters_optimal_1989}, ``magic'' reflects in the very special geometric and algebraic properties inherited from the Weyl-Heisenberg operators \eqref{eq:weyl_ops} and their relations \eqref{eq:weyl_relations}. 
These properties have been used to characterize the entanglement structure of the Magic Simplex and to detect a significant amount of bound entangled states both analytically and numerically \cite{baumgartner_geometry_2008, baumgartner_special_2007, baumgartner_state_2006, popp_almost_2022, popp_comparing_2023}. Identifying a state in $\M_d$ with a point in $\mathbb{R}^{d^2}$ given by its mixing probabilities $c_{k,l}$, the set of Bell-diagonal states form a $(d^2-1)$-dimensional simplex. The geometry of this simplex strongly reflects its entanglement structure. Two relevant subsets with direct relations to the entanglement properties of contained states are the \emph{enclosure polytope} ($\E_d$) and the \emph{kernel polytope} ($\Ker_d$). $\E_d$ contains all states with positive partial transposition (PPT; see Def.~\ref{def:ppt-property}) \cite{baumgartner_special_2007} and other states, while $\Ker_d$ contains only, but not all, separable states. The former is defined as
\begin{gather}
    \label{eq:enclosure_polytope}
    \E_d := \lbrace \rho = \sumkl \ckl\, \Pkl ~ |~
    \sumkl \ckl = 1, \ckl \in  [0, \frac{1}{d}]  \rbrace.
\end{gather}
The kernel polytope $\Ker_d$ relates to mixed states $\rho_{S_d} \in \M_d$, which are based on $d$-element subgroups of the Weyl-Heisenberg group \cite{baumgartner_special_2007}. More precisely, let $S_d$ be a subgroup of $\mathbb{Z}_d^2$ with $d$ elements. The \emph{subgroup state}
\begin{gather}
    \label{eq:subgroup_state}
    \rho_{S_d}:= \frac{1}{d}\sum_{(k,l) \in S_d} \Pkl
\end{gather}
is a separable state \cite{baumgartner_special_2007}. If $d$ is prime, all such subgroups are generated by one generator each. In this case, the Bell states corresponding to the elements of that subgroup form ``lines'' in discrete phase space (see Fig.~\ref{fig:discPhaseSpace}), and the subgroup state is also called \emph{line state}.
$\Ker_d$ is defined as the set of convex combinations of all $d$-element subgroup states:
\begin{gather}
    \label{eq:kernel}
    \Ker_d := \lbrace \rho = \sum_{S_d} q_{S_d} \rho_{S_d} ~|~
    \sum_{S_d} q_{S_d} = 1, q_{S_d} \geq 0  \rbrace \;,
\end{gather}
where the sum is over all $d$-element subgroups $S_d$.
Since the states $\rho_{S_d}$ are separable, all states in the kernel polytope are separable by definition. Compare Fig.~\ref{Md2} for an illustration of $\Ker_d$ for $d=2$.

The Weyl relations \eqref{eq:weyl_relations} further define symmetry transformations for states in $\M_d$.
These linear transformations form a group and act as permutations of the Bell basis. They are known to preserve entanglement and the PPT/NPT property (see Def.~\ref{def:ppt-property} and Ref.~\cite{baumgartner_special_2007}). Consequently, they conserve the \emph{entanglement class}, i.e., they map the sets of separable, PPT entangled, and NPT entangled states onto themselves~\cite{popp_almost_2022}. The symmetry group is generated by the following generators, which can be defined by their action on the Bell basis projectors:
\begin{itemize}
    \item $t_{p,q}: P_{k,l} \mapsto P_{k+p, l+q}~(p,q \in \mathbb{Z}_d)$ (``translation'')
    \item $m:P_{k,l} \mapsto P_{-k, l}$ (``momentum inversion'') 
    \item $r:P_{k,l} \mapsto P_{l, -k}$ (``quarter rotation'')
    \item $v:P_{k,l} \mapsto P_{k+l, l}$ (``shear'')
\end{itemize}
These symmetries are highly relevant for the entanglement structure and the detection of bound entanglement in $\M_d$ (see Sec.~\ref{sec:detection of BE}).

\begin{figure}[ht]
    \centering
    \vspace{-2mm}
    \includegraphics[scale=0.9]{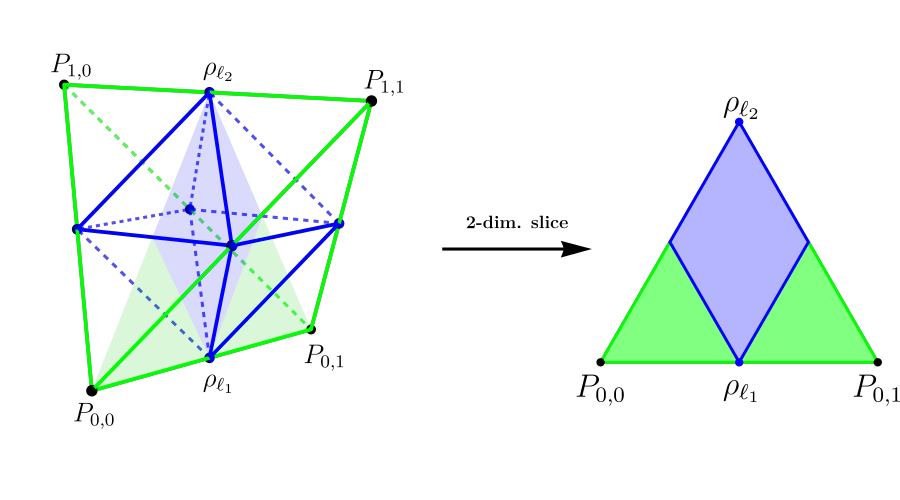}
    \vspace{-7mm}
    \caption{Visualization of the $3$-dimensional Magic Simplex corresponding to $\M_2$ and a $2$-dimensional slice thereof. The convex hull of all four maximally entangled states $P_{k,l}$ (black vertices) is a mathematical simplex in the Euclidean space with dimension $d=3$. The convex hull of the blue vertices representing ``line states'' form the kernel polytope $\Ker_2$ and is equivalent to the set of separable states for $d=2$, i.e., $\SEP \cap \M_2$.  Points outside the blue double pyramid represent entangled states.}
    \label{Md2}
\end{figure}

\subsubsection{General Bell-diagonal Systems}
\label{sec:gen_bell_system}
For $d \geq 3$, there exist other Bell bases of the Hilbert space that are not unitarily equivalent to the basis of standard Bell states \eqref{eq:bell_states}. In Ref.~\cite{popp_special_2024}, it is shown that the entanglement structure of Bell-diagonal states strongly depends on the particular choice of Bell basis. 

A family of \emph{generalized Bell bases} can be constructed as follows: Define the complex matrix $(\alpha_{s,t})_{s,t \in \mathbb{Z}_d}$ with $|\alpha_{s,t}|=1~ \forall s,t \in \mathbb{Z}_d$ and corresponding \emph{generalized Weyl-Heisenberg operators}
\begin{gather}
    \label{eq:gen_weyl_ops}
    V^\alpha_{k,l} := \sum_j w^{j k} \alpha_{j+l,l}\; |j\rangle\langle j+l|,~k,l \in \mathbb{Z}_d.
\end{gather}
Note that these operators are unitary, but generally do not satisfy relations of the form \eqref{eq:weyl_relations} and consequently do not exhibit similar group properties. Nonetheless, by local application of those operators to the maximally entangled state $| \Omega_{0,0} \rangle$, a family of generalized orthonormal Bell states forming a basis can be obtained:
\begin{mydef}[Generalized Bell States]
    \begin{gather}
    \label{eq:gen_bell_states}
    \begin{aligned}
       |\Phi^\alpha_{k,l} \rangle &:= (V^\alpha_{k,l} \otimes \id_d) |\Omega_{0,0} \rangle,~k,l \in \mathbb{Z}_d\\
       \Pkl^\alpha &:= |\Phi^\alpha_{k,l}\rangle \langle\Phi^\alpha_{k,l}|        
    \end{aligned}
\end{gather}  
\end{mydef}
For $\alpha_{s,t} = 1 ~ \forall s,t \in \mathbb{Z}_d$ the generalized Bell states are equal to the standard Bell states \eqref{eq:bell_states}. In general, however, the (generalized) Bell states and the corresponding bases are not unitarily equivalent. The set of \emph{generalized Bell-diagonal states} is defined as mixtures of generalized Bell basis states:
\begin{gather}
    \label{eq:gen_BDS}
    \M^\alpha_d := \lbrace \rho = \sumkl \ckl\, \Pkl^\alpha ~ |~
    \sumkl \ckl = 1, ~\ckl \geq 0  \rbrace
\end{gather}
Due to the lack of a group structure of the generalized Weyl-Heisenberg operators $V^\alpha_{k,l}$, the generalized Bell-diagonal states do not show equally ``magical'' properties as the Magic Simplex (Eq.~\ref{eq:magic_simplex}). The algebraic properties of the generalized Bell states and the geometry of the corresponding simplex are not equally strongly related to the entanglement structure of its Bell-diagonal states \cite{popp_special_2024}.
\section{Free vs. Bound Entanglement} \label{sec:free_vs_bnd}

This chapter discusses the fact that there are two disjoint subsets of entangled states with genuinely different properties.
To derive this fact, we revisit the concepts of Sec.~\ref{sec:PositiveCP}.
There, it is argued that only completely positive (CP) maps are physically realizable.
Hence, for experimental purposes, only the CP maps can be implemented.
However, it can nonetheless be beneficial to consider positive but not completely positive (PNCP) maps for theoretical considerations.
This chapter looks into those maps in Sec.~\ref{sec:PNCP-maps_for_ent_detection} and presents an important example, the transposition map, in Sec.~\ref{sec:PPT-criterion}.
The resulting PPT criterion for separability is closely related to entanglement distillability and is a prime tool for detecting undistillable states.
A surprising twist is the emergence of entangled yet undistillable states for which no examples were known for a long time.
These so-called ``bound'' entangled states are formally introduced in Sec.~\ref{sec:Existence_BE_States}.

\subsection{PNCP Maps for Entanglement Detection} \label{sec:PNCP-maps_for_ent_detection}

From a theoretical point of view, positive but not completely positive (PNCP) maps (cf. Def.~\ref{def:P_and_CP}) are useful as they give an alternative perspective on the separability of quantum states.
The starting point of this consideration is the convexity of the set of separable states SEP.
This is easily shown by taking the convex combination of two arbitrary separable states and observing that the result again satisfies the separability condition of Def. \ref{def:sep_and_ent}.
By Minkowski's theorem, any convex set is fully characterized by the convex hull of its extreme (boundary) points \cite{bengtsson_geometry_2006}.
As any boundary point of a convex set can be viewed as the intersection of the set itself with a tangent hyperplane, one may use the set of all tangent hyperplanes of SEP for its complete characterization.
In light of Sec.~\ref{sec:EWs}, SEP is thus characterized by a complete set of optimal entanglement witnesses.
Due to the Choi-Jamiołkowski isomorphism \cite{choi_completely_1975, jamiolkowski_linear_1972}, there is a correspondence between entanglement witnesses (viewed as linear operators) and positive maps.
This can be utilized to formulate the following characterization of separability \cite{horodecki_separability_1996}:
\begin{mytheorem} \label{thm:alt_sep_ent}
     A state $\rho\in\denop(\hs_A\otimes\hs_B)$ is separable if and only if $(\mathcal{E}\otimes\id_B) (\rho)\geq 0$ for all positive maps $\mathcal{E}:\linop(\hs_A) \rightarrow \linop(\hs_{A'})$.
\end{mytheorem}

Let us shortly outline the practical usage of this theorem for entanglement detection.
The contrapositive of the above statement reads: A state $\rho$ is entangled if and only if there exists a positive map $\mathcal{E}$ such that $(\mathcal{E}\otimes\id_B) (\rho)\ngeq 0$.
By definition, we have $(\mathcal{E}_{CP}\otimes\id_B) (\rho)\geq 0$ for all CP maps $\mathcal{E}_{CP}$, meaning the search for such a map can be restricted to the smaller set of PNCP maps.
By Thm.~\ref{thm:Choi_CP}, the Choi operator of a PNCP map satisfies $J(\mathcal{E}_{PNCP})\ngeq 0$, while for a CP map we have $J(\mathcal{E}_{CP})\geq 0$.
Thus, given a positive map, checking whether it is CP or PNCP is straightforward.
Next, to detect the entanglement of a given quantum state $\rho\in\denop(\hs_A\otimes\hs_B)$ it is sufficient to show that for a PNCP map $\mathcal{E}_{PNCP}:\linop(\hs_A)\rightarrow\linop(\hs_{A'})$ at least one eigenvalue of $(\mathcal{E}_{PNCP}\otimes\id_B) (\rho)$ is negative.
The nontrivial aspect of this procedure lies in finding a suitable PNCP map.
Ideally, this map allows the detection of a large subset of ENT.
This sketches one basic strategy of entanglement detection via PNCP maps.
It is straightforwardly applied in the next section and, with minor adjustments, also for further entanglement detection criteria (see Sec.~\ref{sec:reduction_criterion}).

\subsection{PPT Criterion} \label{sec:PPT-criterion}

The prime example of a PNCP map is the transposition map $T: L(\mathcal{H}) \rightarrow L(\mathcal{H})$ with respect to a specific basis, acting as $\rho \mapsto \rho^T$.
The map is positive since the spectrum of a matrix is invariant under transposition.
Furthermore, it is easy to check that the Choi operator $J(T)\ngeq 0$ (cf. Thm.~\ref{thm:Choi_CP}).
Hence, $T$ is indeed PNCP.
The \textit{partial transposition} of only one part of a larger system is denoted by $\Gamma: L(\mathcal{H}_A \otimes \mathcal{H}_B) \rightarrow L(\mathcal{H}_{A} \otimes \mathcal{H}_B)$ and its action on a bipartite quantum state $\rho\in\denop(\hs_A\otimes\hs_B)$ is defined as
\begin{align}
    \rho^{\Gamma} := (T\otimes\id_B)(\rho) \;.
\end{align}
Here, we arbitrarily choose the transposition map to act on the first subsystem rather than the second.
For the spectrum of $\rho^{\Gamma}$ this choice is indifferent\footnote{To see this, we define the alternative choice as $\rho^{T_B} := (\id_A\otimes T)(\rho)$.
From $(\rho^T)^T =\rho$ we get $(\rho^{T_B})^T = \rho^\Gamma$, and thus $\rho^{T_B}\geq0$ if and only if $\rho^\Gamma\geq0$ because the eigenvalues are invariant under transposition.}.

\begin{mydef}[PPT Property] \label{def:ppt-property}
    A bipartite quantum state $\rho \in \denop(\hs_A \otimes \hs_B)$ satisfying $\rho^{\Gamma} \geq 0$ is said to be positive under partial transposition (denoted $\rho\in\mathrm{PPT}$), else non-positive under partial transposition ($\rho\in\mathrm{NPT}$).
\end{mydef}
By definition, the subsets PPT and NPT form a partition of the space of bipartite quantum states.
A direct consequence of Thm. \ref{thm:alt_sep_ent} is the \textit{PPT criterion} or \textit{Peres-Horodecki criterion} \cite{peres_separability_1996, horodecki_separability_1996}:
\begin{mytheorem}[PPT criterion] \label{thm:PPT-crit}
    For a bipartite quantum state $\rho\in\denop(\hs_A\otimes\hs_B)$ it holds that $\rho\in\mathrm{SEP}\Rightarrow\rho\in\mathrm{PPT}$, or equivalently $\rho\in\mathrm{NPT}\Rightarrow\rho\in\mathrm{ENT}$.
\end{mytheorem}
A quantum state $\rho\in\mathrm{PPT}$ is said to satisfy the PPT criterion.
Hence, a violation thereof indicates entanglement.
Furthermore, for small Hilbert space dimensions, a stronger statement can be derived \cite{horodecki_separability_1996}:
\begin{mytheorem} \label{thm:PPT-crit_qubit_qubit}
    For qubit-qubit ($d_A=d_B=2$) and qubit-qutrit ($d_A=2,\, d_B=3$) systems, the PPT criterion is necessary and sufficient for separability, i.e., $\rho\in\mathrm{PPT} \Leftrightarrow \rho\in \mathrm{SEP}$.
\end{mytheorem}
As elaborated in the upcoming section, this no longer holds for composite systems with greater Hilbert space dimensions.

\subsection{Existence of PPT Entangled and Bound Entangled States} \label{sec:Existence_BE_States}

The violation of the PPT criterion is only sufficient but not necessary for entanglement in higher dimensional systems.
This indicates the possibility of different entanglement classes characterized by their properties regarding certain quantum operations (cf. Sec.~\ref{sec:BE_and_ent_measures}).

Note that the PPT criterion is only one among many criteria to detect, in general, different subsets of entangled states (see Sec.~\ref{sec:detection of BE}).
Nowadays, several PPT entangled states, including qutrit-qutrit ($d_A=d_B=3$) and qubit-quqart ($d_A=2,\, d_B=4$) systems, are known (cf. Ref.~\cite{horodecki_separability_1997}).
Consequently, Thm. \ref{thm:PPT-crit_qubit_qubit} can not be extended to higher dimensions.

Entanglement distillation (cf. Sec.~\ref{sec:ent_dist}) is one quantum information theoretic task that is strongly affected by the PPT property.
Considering general LOCC protocols, one can show the following \cite{horodecki_mixed-state_1998}:
\begin{mytheorem} \label{thm:dist_implies_NPT}
    If a bipartite quantum state $\rho$ is distillable, then $\rho\in\mathrm{NPT}$.
    Equivalently, if $\rho\in\mathrm{PPT}$, then $\rho$ is undistillable.
\end{mytheorem}

\begin{proof}
Consider the $d$-dimensional Hilbert spaces $\hs_A, \hs_B \simeq \mathbb{C}^d$.
A bipartite quantum state $\rho \in \denop(\hs_A \otimes \hs_B)$ is distillable if and only if, for some $n\geq 1$, one can transform by LOCC $n$ copies of $\rho$ into the maximally entangled qubit state $|\Phi^+\rangle = (1/\sqrt{2})\sum_{i=0}^{1} |i\rangle\otimes |i\rangle \in \mathbb{C}^2 \otimes \mathbb{C}^2$ embedded in $\mathbb{C}^d \otimes \mathbb{C}^d$ (cf. Def.~\ref{def:distillation}).
The target state $|\Phi^+\rangle$ is NPT (i.e., $|\Phi^+\rangle\langle \Phi^+|^\Gamma \ngeq 0$, where $\Gamma$ denotes the partial transposition on $\hs_B$).

The crucial observation in the proof is that LOCC transformations preserve the PPT property of a state.
To see this, consider the most general LOCC transformation $\mathcal{E}_{LOCC}$ of $n$ copies of the state $\rho$ given by \cite{vedral_quantifying_1997}
\begin{align} \label{eq:proof_dist_implies_NPT}
    \mathcal{E}_{LOCC}(\rho^{\otimes n}) = \frac{1}{N} \sum_{i} \big(A_i \otimes B_i\big)\, \rho^{\otimes n} \, \big(A_i^\dagger \otimes B_i^\dagger\big) \;,
\end{align}
where $N = \myTr(\sum_{i} (A_i \otimes B_i)\, \rho^{\otimes n} \, (A_i^\dagger \otimes B_i^\dagger))$ is a normalization factor, and $A_i, B_i \in \linop(\mathbb{C}^{d^n})$.
If $\rho$ is distillable, then $\mathcal{E}_{LOCC}(\rho^{\otimes n}) = |\Phi^+\rangle\langle \Phi^+|$ for some LOCC channel.
Next, for any local operators $K, L, M, N \in \linop(\mathbb{C}^d)$, one can verify $((K\otimes L)\, \rho \, (M\otimes N))^\Gamma = (K\otimes L^T)\, \rho^\Gamma \, (M\otimes N^T)$, which, by linearity, extends to the general LOCC transformation given in \eqref{eq:proof_dist_implies_NPT} to yield
\begin{align}
    \mathcal{E}_{LOCC}(\rho^{\otimes n})^\Gamma = \frac{1}{N} \sum_{i} \big(A_i \otimes B_i^T\big)\, \left(\rho^\Gamma\right)^{\otimes n} \, \big(A_i^\dagger \otimes \left(B_i^T\right)^\dagger\big) \;,
\end{align}
defining again a CP map.
Thus, $\rho^\Gamma\geq 0$ implies $\mathcal{E}_{LOCC}(\rho^{\otimes n})^\Gamma\geq 0$.
Consequently, $\mathcal{E}_{LOCC}(\rho^{\otimes n}) \neq |\Phi^+\rangle\langle \Phi^+|$ for any $\rho\in\mathrm{PPT}$ which concludes the proof.
\end{proof}

The theorem implies that also entangled states can be undistillable, particularly if they are PPT entangled.
This leads to the following definition:
\begin{mydef}[Free and Bound Entanglement] \label{def:FE_and_BE}
    A bipartite quantum state $\rho$ is called free entangled (denoted $\rho\in\mathrm{FE}$) if it is distillable.
    If $\rho$ is entangled and undistillable, it is called bound entangled ($\rho\in\mathrm{BE}$).
\end{mydef}

The term ``bound entanglement'' originates from the fact that the entanglement is bound to the state in some sense.
More specifically, it is inaccessible for various quantum information protocols that involve entanglement distillation (cf. Sec.~\ref{sec:ent_dist}).
This can be formalized using entanglement measures where the bound nature is reflected in the difference between two measures with related operational interpretations (cf. Sec.~\ref{sec:BE_and_ent_measures}).

It is important to emphasize that the sets of PPT entangled and bound entangled states are not necessarily equivalent.
Even though PPT entanglement implies bound entanglement by Thm. \ref{thm:dist_implies_NPT}, the converse might not be true.
This becomes more transparent when viewing Thm. \ref{thm:dist_implies_NPT} in set-theoretic language, illustrated in Fig.~\ref{fig:FE_BE_Set_relations}.
Up to today, no NPT bound entangled state has been found.
Therefore, the question whether $\mathrm{PPT}\overset{?}{\equiv}\mathrm{SEP}\cup\mathrm{BE}$ is still open.

\begin{figure}
    \centering
    \includegraphics[width=0.6\linewidth]{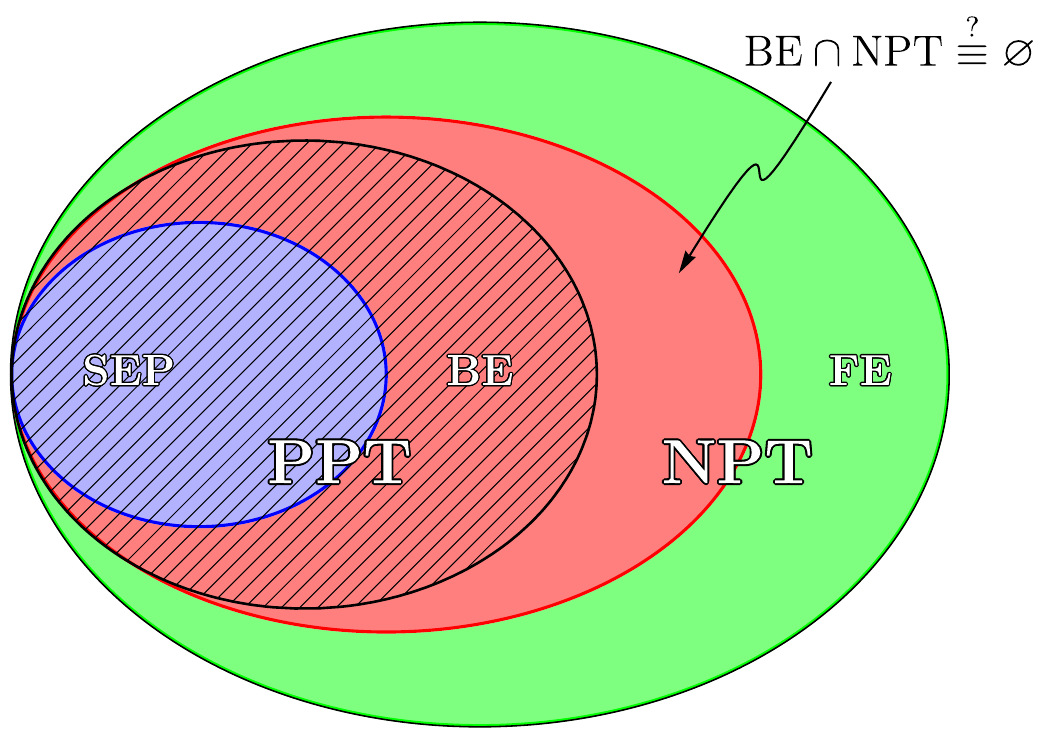}
    \caption{
    Set relations of the sets SEP (blue), ENT (green and red), PPT (hatched), NPT (non-hatched), FE (green), and BE (red).
    We have $\mathrm{FE}\subset\mathrm{NPT}$ or equivalently $\mathrm{PPT}\subset\mathrm{SEP}\cup\mathrm{BE}$.
    The question is whether it holds that $\mathrm{PPT}\overset{?}{\equiv}\mathrm{SEP}\cup\mathrm{BE}$.
    To settle the case, one would need to prove or disprove $\mathrm{BE}\overset{?}{\subset} \mathrm{PPT}$.
    Disproving this could be accomplished by finding a specific NPT bound entangled state, which, up to today, has not been successful. See Sec.~\ref{sec:outlook} for a discussion of this problem.}
    \label{fig:FE_BE_Set_relations}
\end{figure}

\section{Detection of Bound Entanglement} \label{sec:detection of BE}

In this chapter, we describe mathematical methods to verify whether a particular state is separable, bound entangled, or free entangled.
The first two sections introduce the powerful toolbox of entanglement witnesses, whereas the rest of the chapter focuses on further entanglement criteria.
To verify that a quantum state is bound entangled, one needs to prove that the state is entangled and undistillable.
Up to date, the only known strategy to show undistillability is given by the PPT criterion (cf. Thm.~\ref{thm:PPT-crit}).
Therefore, other entanglement criteria must be combined with the PPT criterion to classify a state as PPT bound entangled.
Note that the reduction criterion presented in Sec.~\ref{sec:reduction_criterion} can not be used in this way because it is weaker than the PPT criterion.
However, it may allow the distinction between potential NPT bound entangled and NPT free entangled states.

\subsection{Entanglement Witnesses} \label{sec:EWs}

The concept of entanglement witnesses is of fundamental importance for characterizing the separability problem (Def.~\ref{separabilityproblem}). Optimal entanglement witnesses characterize the convex set of separable states, and decomposable/non-decomposable entanglement witnesses classify the sets of NPT and PPT entangled states, respectively. Moreover, entanglement witnesses are also a powerful tool to verify entanglement in experiments since they are Hermitian operators and, therefore, in principle, measurable observables. In Chapter~\ref{BEexperiment}, we present more details on realizations of those observables for detecting entanglement.

\begin{mydef}[Entanglement Witness]\label{DefEW}
 A Hermitian operator $W$ is called an entanglement witness (EW) if
\begin{gather}
\begin{aligned}
0&\leq\myTr(W\;\sigma)\quad \forall\; \sigma\in \mathrm{SEP}\qquad\textrm{and}\\
0&>\myTr(W\;\rho)\quad\textrm{for some}\;\;   \rho\in \mathrm{ENT},
\end{aligned}
\end{gather}
where $\sigma,\rho \in \denop(\hs)$ and $W\in \linop(\hs)$. In this case, we say that $W$ witnesses/detects the entanglement of $\rho$. Furthermore, an entanglement witness is called optimal if there exists some $\sigma \in \mathrm{SEP}$ such that $\myTr(W\;\sigma)=0$.    
\end{mydef}

Any entangled state is detectable by an entanglement witness, i.e., given a particular entangled state there always exists a Hermitian operator $W$ with the above properties. The knowledge of all optimal entanglement witnesses $\tilde{W}$ fully characterizes the set SEP. This means that the separability problem is dual to the knowledge of the set of all optimal entanglement witnesses. 
The proven NP-hardness of the separability problem is reflected in the equally hard problem of finding a complete set of optimal entanglement witnesses.

Since an entanglement witness is a Hermitian operator, it can be realized experimentally in order to verify, i.e., detect, the entanglement of a given state. This also implies a significant advantage from an experimental point of view because, generally, significantly fewer measurement settings are necessary to detect entanglement of an unknown state than via state tomography. More details can be found in Sec.~\ref{BEexperiment}.

In the following sections, we review important properties of entanglement witnesses such as decomposability, mirrored entanglement witnesses, and nonlinear witnesses, with a focus on detecting and characterizing bound entangled states.

\subsubsection{Decomposable and Non-Decomposable Entanglement Witnesses}\label{sec:decomposible}

\begin{mydef}[Decomposable Entanglement Witness]
$W$ is called a decomposable entanglement witness if it can be written as $W=A+B^\Gamma$ with $A, B\geq0$.
Here, $B^\Gamma$ denotes the partial transposition of $B$.
\end{mydef}

From that, the following theorem can be deduced:
\begin{mytheorem}\label{nondecomposible}
Decomposable entanglement witnesses cannot detect PPT entangled states.
\end{mytheorem}
This is true because we have $\myTr(W\,\rho) = \myTr(A\,\rho)+\myTr(B^\Gamma\rho) = \myTr(A\,\rho)+\myTr(B\,\rho^\Gamma) \geq 0$, where we used that the trace is invariant under transposition and that $\rho^\Gamma\geq0$ for all $\rho\in\mathrm{PPT}$.

Equivalently, one can define an entanglement witness $W$ as decomposable if $\myTr(W\rho)\geq 0$ holds for all PPT states $\rho$ (see, e.g., Ref.~\cite{chruscinski_entanglement_2014}). In summary, non-decomposable witnesses have to be constructed to detect bound entangled states. Consequently, finding non-decomposable witnesses or bound entangled states is a dual problem and, as such, a generally difficult problem. 

\subsubsection{Mirrored Entanglement Witnesses}\label{sec:mirroredEW}

Surprisingly, the authors of Ref.~\cite{bae_mirrored_2020} found that the very definition of an entanglement witness (Def.~\ref{DefEW}) does often also have a non-trivial second bound $U_{W}$, namely
\begin{eqnarray}
0\;\leq\;\myTr(W\;\sigma)&\leq&U_{W} \;, \qquad \forall\; \sigma\in \mathrm{SEP} \;,
\end{eqnarray}
such that one can also find entangled states that violate this new bound, i.e., $\myTr(W\,\rho)>U_{W}$ for some $\rho\in \mathrm{ENT}$. This can be formalized by the following definition~\cite{bae_mirrored_2020}:

\begin{mydef}[Mirrored Entanglement Witness]\label{mirroredEW}
  Given an entanglement witness $W$, one defines a mirrored operator by
\begin{equation}   \label{M!}
    W_{\rm M} = \mu\; \mathbb{1}_{d_A} \otimes \mathbb{1}_{d_B} - W,
\end{equation}
with the smallest $\mu>0$ such that $\tr(W_{\rm M} \; \sigma) \geq 0$ for all $\sigma \in \mathrm{SEP}$.
Moreover, if the maximal eigenvalue of $W$ satisfies $\lambda_{\rm max} > \mu$, then $W_{\rm M}$ is an entanglement witness, and hence one has a pair $(W, W_{\rm M})$ of mirrored entanglement witnesses. 
\end{mydef} 

\begin{figure}
    \centering
    \includegraphics[width=0.65\linewidth]{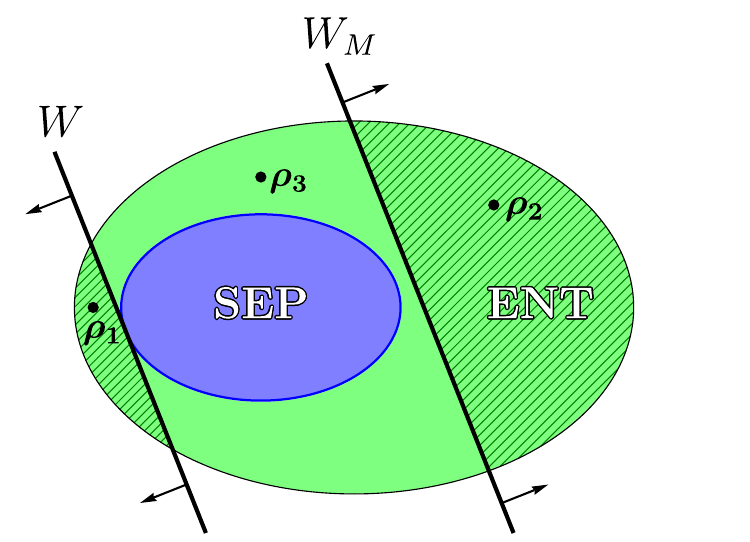}
    \caption{This picture depicts the quantum state space, separated into the sets SEP and ENT. The witness $W$ detects the entanglement of the state $\rho_1$, but neither of $\rho_2$ nor $\rho_3$. On the other hand, the mirrored witness $W_M$ based on the same local observables can also witness the entanglement of $\rho_2$. Therefore, the pair of witnesses $\{W, W_M\}$ can detect, in general, more entangled states than only considering one of the two witnesses.}
    \label{fig:mirroredEW}
\end{figure}

The concept is summarized in Fig.~\ref{fig:mirroredEW}.
The idea comes from the observation that for any positive semi-definite matrix $K\geq 0$ with $\myTr(K)=1$ and $K\in \linop(\hs_A\otimes\hs_B)$, which can be interpreted as a state on a composite system, we can find an upper bound $U$ and a lower bound $L$ defined by $U(K):=\max_{\sigma}(\myTr(K\; \sigma))$ and $L(K):=\min_\sigma(\myTr(K\; \sigma))$ with $\sigma\in \mathrm{SEP}$. Obviously, it follows that
\begin{eqnarray}
    L(K)&\leq& \myTr(K\; \sigma)\;\leq\; U(K)\quad\forall\quad \sigma\in \mathrm{SEP}\;.
\end{eqnarray}
Let us denote by $\lambda_{min}(K)$ and $\lambda_{max}(K)$ the minimal and maximal eigenvalues of $K$, respectively. Then, if $\lambda_{min}(K)<L(K)$ and/or $U(K)<\lambda_{max}(K)$ the bounds can be used to detect entangled states. Namely, one finds
\begin{align}\label{eq:separabiltywindow}
    \myTr(K\; \rho)\notin[L(K),U(K)]\; \Longrightarrow \; \rho \in \mathrm{ENT}\;.
\end{align}
The interval $[L(K), U(K)]$ is called the \emph{separability window} of the observable $K$.

Due to linearity, each bound can be associated with an entanglement witness, i.e.,
\begin{eqnarray}
    W^{(+)}&=&\frac{1}{1-d_1 d_2 L(K) }\left(K-L(K)\;\id_{d_A}\otimes\id_{d_B}\right)\;,\\
     W^{(-)}&=&\frac{1}{d_1 d_2 U(K)-1}\left(-K+U(K)\;\id_{d_A}\otimes\id_{d_B}\right)\;.
\end{eqnarray}
Thus, we have
\begin{gather}
    \begin{aligned}
        0\;\leq\; \myTr(W^{(+)} \,\rho)&\Longleftrightarrow L(K)\leq \myTr(K\; \rho)\;,\\
        0\;\leq\; \myTr(W^{(-)} \,\rho)&\Longleftrightarrow \myTr(K\; \rho)\leq U(K)\;.
    \end{aligned}
\end{gather}
This means that from one observable $K$ with a non-trivial separability window, we constructed a pair of witnesses $\{W^{(+)}, W^{(-)}\}$. Two witnesses can potentially detect more entangled states than one. Moreover, note that those witnesses are constructed from the same experimental settings. 

As proven in detail in Ref.~\cite{bae_mirrored_2020}, there is a universal relation between the pair $\{W^{(+)}, W^{(-)}\}$, so that mirrored entanglement witnesses can directly be constructed from each other, which coincides with the Def.~\ref{mirroredEW}. 

Furthermore, the authors of Ref.~\cite{chruscinski_mirrored_2025} showed that there exist mirrored pairs of entanglement witnesses that are both optimal and non-decomposable.
Consequently, the concept of mirrored entanglement witness can aid in detecting PPT-entangled states in an optimal way.
This also refutes a recent conjecture posed in Ref.~\cite{bera_structure_2023}.

\subsubsection{Nonlinear Entanglement Witnesses}

Entanglement witnesses are linear hyperplanes that distinguish entangled states from the set of separable states in the space of Hermitian operators. Since the set of separable states is convex, one may improve a given entanglement witness by introducing nonlinearities such that the convex set is better approximated and a larger set of entangled states can be detected. 

The first ideas in this direction came from uncertainty relations which are inequalities in terms of variances~\cite{hofmann_violation_2003,guhne_nonlinear_2006}.
The variance of an observable $K$ for a state $\rho$ is defined as $\Delta(K)_{\rho} = \sqrt{\myTr(K^2 \rho) - \myTr(K\rho)^2}$, which is clearly nonlinear in $\rho$, and is shown to be capable to detect entanglement in Ref.~\cite{guhne_nonlinear_2006}.

Another, more systematic method of constructing nonlinear entanglement witnesses is introduced in Ref.~\cite{guhne_nonlinear_2006}. Suppose that some entangled state $\rho$ is detected by a positive but not completely positive map $\Lambda$, that is, $(\id \otimes\Lambda)(\rho)$ possesses an eigenvector $|\varphi\rangle \in \hs_A \otimes \hs_B$ corresponding to a negative eigenvalue. Then the formula
\begin{equation}
W = (\mathbb{1} \otimes \Lambda^{*})( |\varphi\rangle \langle \varphi |) \;, 
\end{equation}
where $\Lambda^{*}$ denotes a dual map, defines an entanglement witness detecting $\rho$, i.e., $\myTr(W \rho)<0$. A nonlinear functional is then constructed by subtracting polynomials from an entanglement witness. To be specific, let us consider subtracting a single nonlinear polynomial as follows \cite{guhne_nonlinear_2006},
\begin{equation} \label{N2}
{\cal F} [\rho] = \myTr\big( W \rho\big) - \frac{1}{s^2(\chi)} \myTr\big( \widetilde{X}\; \rho\big)\cdot\myTr\big( \widetilde{X}^\dagger\; \rho\big),
\end{equation}
where one can choose $\widetilde{X} =  ({\rm \mathbb{1}} \otimes \Lambda^*)( |\varphi\rangle \langle \chi|)$ for some state $|\chi \rangle \in \hs_A \otimes \hs_B$ and $s(\chi)$ is the largest Schmidt coefficient of $|\chi \rangle$. The parameters are chosen such that $\mathcal{F}[\sigma]\geq 0$ for all $\sigma\in \mathrm{SEP}$. It is clear that a witness $W$ does not detect a larger set of entangled states than its nonlinear functional ${\cal F}$ since $\mathcal{F}[\rho] \leq \myTr(W\;\rho)$.

\subsection{Entanglement Witnesses based on MUBs and SICs and Generalizations}\label{POVMWitnesses}

Any entanglement witnesses can be decomposed into local measurements $A_i \otimes B_j \in \linop(\hs_A \otimes \hs_B)$:
\begin{eqnarray}\label{witnessdecomositioninlocal}
    W = \sum_{i,j} w_{i,j}\; A_i\otimes B_j\;.
\end{eqnarray}
This is the important connection of the theory to any experimental realizations, where it is of practical importance to minimize the number of local observables.

An interesting class of entanglement witnesses arises from sets of mutually unbiased bases (MUBs) that are experimentally very feasible (see also Sec.~\ref{BEexperiment}). They also gradually approach a set that can be used for quantum state tomography as they converge to a complete set of MUBs\footnote{A set of MUBs is called complete if it contains the maximal number of MUBs. For prime-power dimensions, this is known to be $d+1$.}. Therefore, it is clear that a well-chosen complete set of MUBs should, in principle, be able to detect bound entanglement, at least indirectly. Surprisingly, as presented in Sec.~\ref{sec:MUBwitness}, a complete set is not necessary to witness bound entanglement. A smaller set of experimental settings can sometimes be found, which is a sign of a substructure of the set of bound entangled states.

A complete set of MUBs is a $2$--design~\cite{bae_linking_2019} such as SICs (symmetrical informationally complete), Sec.~\ref{sec:SICs}. SIC witnesses can be constructed by applying the same method as for a set of MUBs. As shown in Sec.~\ref{sec:mirroredEW}, every entanglement witness has a mirrored one. For MUB-- and SIC--witnesses the corresponding nontrivial upper bounds become strongly dependent on the dimension $d$ because different choices exist of complete MUB or SIC bases. This reveals that the entanglement structure is involved when the dimension $d$ increases and, by that, the detection of bound entanglement. 

\subsubsection{A Complete MUB set is not needed to detect Bound Entanglement}\label{sec:MUBwitness}

Mutually unbiased bases are sets of different bases with the property that a measurement in one basis reveals no information about the measurement statistics for a measurement in another.
Mathematically, this can be defined as follows:
\begin{mydef}[Mutually Unbiased Bases]\label{def:unbiased}
    A pair of orthonormal bases $\mathcal{B}_k$ and $\mathcal{B}_l$ of the Hilbert space $\mathbb{C}^d$ are mutually unbiased if the basis vectors satisfy $|\langle i_k | j_l \rangle |^2 = \frac{1}{d} (1-\delta_{k,l}) +\delta_{k,l}\delta_{i,j}$ for all $\ket{i_k}\in\mathcal{B}_k$ and $\ket{j_l}\in\mathcal{B}_l$.
    A set of pairwise mutually unbiased bases $\{\mathcal{B}_k\}_k$ is referred to as mutually unbiased bases (MUBs).
\end{mydef}
 This condition formalizes mathematically the concept of complementarity, i.e., if the system resides in an eigenstate of one observable, the outcome of measuring a complementary observable is maximally uncertain. Generalizing to larger collections of measurements, we say a set of bases is mutually unbiased if each pair of bases is mutually unbiased. For prime-power dimensions, it has been shown that there exists a complete set of exactly $d+1$ MUBs~\cite{ivonovic_geometrical_1981,wootters_optimal_1989} and explicit constructions thereof, e.g., based on the Weyl-Heisenberg operators \eqref{eq:weyl_ops} (see \cite{durt_mutually_2010} for a review). However, it is conjectured that fewer MUBs exist in all other cases~\cite{boykin_real_2005}. In particular, for dimension $d=6$, so far, only a set of $3$ MUBs has been found.

The first direct connection between MUBs and entanglement witnesses is made in Ref.~\cite{spengler_entanglement_2012}, which we show below. The one for dimension $d=3$ was used for the first experiment witnessing bound entanglement (see Sec.~\ref{BEexperiment}).  Generalizations of the witnesses to a wider class of measurements, including $2$-designs and mutually unbiased measurements (MUMs), for which MUBs are an example, are investigated in Ref.~\cite{chruscinski_entanglement_2014,bae_linking_2019,hiesmayr_detecting_2021}.

We derive entanglement witnesses from a collection of MUBs. Let us consider $m$ MUBs in a $d$-dimensional Hilbert space,
\begin{equation}
\mathcal{M}_m = \{ \mathcal{B}_0,\mathcal{B}_1,\ldots,\mathcal{B}_{m-1}\}\,,
\end{equation}
where $\mathcal{B}_k =\{ |i_k\rangle \}_{i=0}^{d-1}$, and $\mathcal{B}_0$ is chosen as the computational basis $\{ |i_0\rangle \}=\{|0\rangle,\dots,|d-1\rangle\}$. From these MUBs, we construct the operators
\begin{eqnarray}
 \mathbf{B}(\mathcal{M}_m,s) &=& \sum_{l=0}^{d-1} |l \rangle\langle l| \otimes  | l +s \rangle\langle l+s|  +  \sum_{k=1}^{m-1} \sum_{i=0}^{d-1} |i_k \rangle\langle i _k| \otimes |i_k \rangle\langle i_k|\;, 
\end{eqnarray}
and its partial transpose
\begin{eqnarray}
 \mathbf{B}^\Gamma(\mathcal{M}_m,s) &=& \sum_{l=0}^{d-1} |l\rangle\langle l| \otimes  |l +s \rangle\langle l+s|  +  \sum_{k=1}^{m-1} \sum_{i=0}^{d-1} |i_k \rangle\langle i _k| \otimes |i_k^* \rangle\langle i _k^*|\;,\end{eqnarray}
where the complex conjugation (and transposition) is performed with respect to the canonical basis $\mathcal{B}_0$. It is proven in Ref.~\cite{spengler_entanglement_2012} that for any separable state $\sigma \in \mathrm{SEP}$, the two operators are bounded from above by the same quantity, depending on the number of MUBs, $m$, and the dimension $d$. In particular,
\begin{equation}
\myTr(\mathbf{B}(\mathcal{M}_m,s) \, \sigma) \leq \frac {d+m-1}{d} \;,
\end{equation}
and
\begin{equation}
\myTr(\mathbf{B}^\Gamma(\mathcal{M}_m,s)\, \sigma) \leq \frac {d+m-1}{d} \;.
\end{equation}
Hence, if the operators
\begin{equation}\label{MUBwitness}
  W(\mathcal{M}_m,s) = \frac {d+m-1}{d}\; \mathbb{1}_d \otimes \mathbb{1}_d -  \mathbf{B}(\mathcal{M}_m,s) ,
\end{equation}
and $W^\Gamma(\mathcal{M}_m,s)$ possess a strictly negative eigenvalue, they satisfy the very definition of an entanglement witness~(Def.~\ref{DefEW}) with
\begin{gather}
    \begin{aligned}
        0&\leq \myTr(W(\mathcal{M}_m,s)\, \sigma)\quad\textrm{for all}\quad \sigma\in \mathrm{SEP}\;,\\
        0&\leq \myTr(W^\Gamma(\mathcal{M}_m,s)\, \sigma)\quad\textrm{for all}\quad \sigma\in \mathrm{SEP}\;.
    \end{aligned}
\end{gather}

Only non-decomposable entanglement witnesses can detect PPT entangled states (see Thm.~\ref{nondecomposible}). This property depends on the choice $s$ \cite{spengler_entanglement_2012}:
\begin{mytheorem}
    The entanglement witnesses $W(\mathcal{M}_{m=d+1},s=0)$ and $W^\Gamma(\mathcal{M}_{m=d+1},s=0)$ are decomposable.
\end{mytheorem}
Consequently, they can not detect PPT entangled states.
But note that the entanglement witnesses $W^\Gamma(\mathcal{M}_m,s=0)$ are anyhow powerful witnesses as they detect
all possible entangled states of the family of isotropic
states  $\{\rho_{iso}(p)=p P_{k,l}+ (1-p) \frac{1}{9}\mathbb{1}_9\}_{k,l=0}^{d-1}$ (where $P_{k,l}$ is one of the Bell states~\eqref{eq:bell_states}). Via this family of states,
it is also proven that the following theorem holds~\cite{spengler_entanglement_2012}:
\begin{mytheorem}
   No more than $d+1$ MUBs can exist. Otherwise, $W^\Gamma(\mathcal{M}_m,s=0)$ would witness separable states as entangled states.
\end{mytheorem}

We are interested in non-decomposable entanglement witnesses since only those can detect bound entangled states. As already mentioned above, the non-decomposability depends on the choice of $s$. In particular, it has been found that \cite{bae_how_2022}:
\begin{mytheorem}\label{mubforexp}
   Let $s>0$ and $d>2$.  Then  $W^\Gamma(\mathcal{M}_{m=d+1,s)}$ defines a non-decomposable entanglement witness if and only if  $s \neq \frac{d}{2}$.
\end{mytheorem}

Finally, in strong contrast to SICs (see next section~\ref{sec:SICs}), a fascinating fact for entanglement witnesses based on MUBs has been proven, namely a complete set of MUBs is not needed to detect bound entangled states. In detail:
\begin{mytheorem}\label{notcompletestetofMUBs}
     Let $s \neq \frac{d}{2}$ with $s>0$ and $d>2$. If $m > \frac d2 + 1$, then  $W^{\Gamma}(\mathcal{M}_m,s)$ defines a non-decomposable entanglement witness, i.e., it can detect bound entangled states.
\end{mytheorem}

Particularly, for $d=3$, exactly $4$ mutually unbiased bases exist. However, bound entangled states can also be detected by choosing only $3$ MUBs, considerably reducing the number of different experimental setups. In Fig.~\ref{fig:notcompleteMUB}, we present two slices through the Magic Simplex~\eqref{eq:magic_simplex} visualizing the effect of a MUB witness based on $4$ versus $3$ MUBs.

\begin{figure}
    \centering
    \includegraphics[width=1\linewidth]{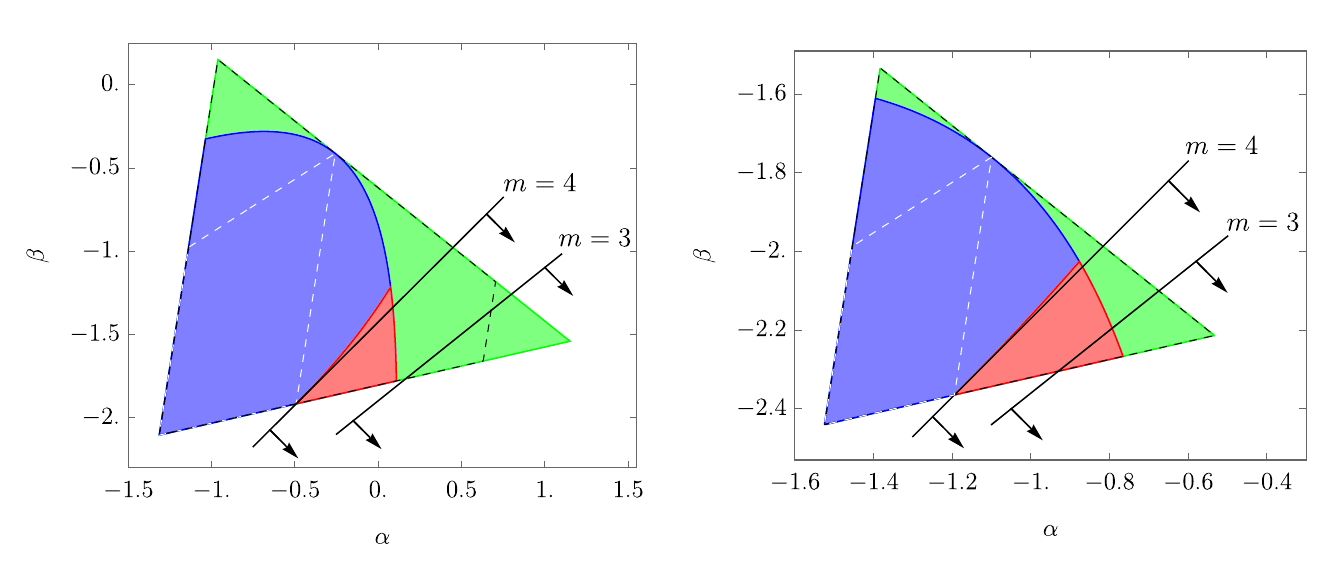}
\caption{
These graphics depict two slices, $B_1$ (left; $\gamma = -1/\sqrt{3}$) and $B_2$ (right; $\gamma = -0.83$), through the Magic Simplex for bipartite qutrits in the state family $\rho_B=\frac{1-\frac{\alpha}{5}-\frac{\beta}{4}-\gamma}{9}\mathbb{1}_9+\frac{\alpha}{5} P_{0,0}+\frac{\beta}{8} \sum_{i=1}^2 P_{i,0}+\frac{\gamma}{3} \sum_{i=0}^2 P_{i,1}$. NPT states are colored green. The red area corresponds to bound entangled states detected by the realignment criterion (Thm.~\ref{thm:realignment_criterion}). The MUB witnesses $W^\Gamma(\mathcal{M}_m,1)$, solid lines with arrows, are drawn for the maximal number of MUBs $m=4$ and $m=3$. Only for family $B_2$ can it detect entanglement in the PPT region (blue and red). The white dashed lines correspond to the kernel polytope (\ref{eq:kernel}); all states inside must be separable. The black dashed lines correspond to the enclosure polytope (\ref{eq:enclosure_polytope}); all states outside must be NPT entangled.}
\label{fig:notcompleteMUB}
\end{figure}

This brings us to the next question: Can a different subset of MUBs be more or less efficient than another subset in detecting entanglement? Indeed, this is the case~\cite{hiesmayr_detecting_2021} and is discussed in the next section.

\subsubsection{On the Universality of MUBs for Non-Decomposable Witnesses}\label{sectionsubsets} 

It has been shown that certain sets of MUBs are more useful than others in detecting entanglement~\cite{hiesmayr_detecting_2021}---a consequence of the unitary inequivalence of different classes of MUBs. In this sense, MUBs are not universal. Thus, it is expected that MUB witnesses do not only depend on the number of MUBs $m$ (see Thm.~\ref{notcompletestetofMUBs}) but also on the particular choice of the set of MUBs. In Ref.~\cite{hiesmayr_detecting_2021}, it is shown that for $d\geq 4$, inequivalent sets differ in their efficiency of detecting entanglement. Moreover, there are not only unitarily inequivalent sets of complete MUBs but also unextendible MUB sets, i.e., sets that can not be extended to a complete set of MUBs, which are more effective in detecting entanglement.

The next section discusses the mirrored entanglement witnesses of MUB witnesses, where the lower bounds show their dependence not only on the number of MUBs but the particular choice of MUBs within the MUB set.

In summary, Theorem~\ref{notcompletestetofMUBs} shows that, since \emph{any} set of $m$ MUBs ($m>d/2+1$) can be used to construct a non-decomposable entanglement witness that detects bound entangled states, the non-decomposability of such a witness is a universal property of mutually unbiased bases. However, since different equivalence classes of MUBs give rise to different entanglement witnesses, an interesting interplay is revealed between the universality of MUBs for non-decomposable witnesses and the dependence on the local choice of MUBs in detecting a given bound entangled state. The nature of such connections would be worth investigating further since they may shed light on the structure of the set of bound entangled states.

\subsubsection{Entanglement Witnesses based on Sets of SICs}\label{sec:SICs}

Let us consider a special set of POVMs (see also Sec.~\ref{sec:POVM}) called \textit{symmetric informationally complete} POVMs (abbreviated SICs). 

\begin{mydef}[Symmetric Informationally Complete POVMs] Let $\{\Pi_s\}_{s=1}^{d^2}$ be a set of rank-$1$ projectors with 
\begin{equation}
    \label{def:SIC}
    \myTr(\Pi_s \Pi_r)=\frac{1}{d+1} (1+ d\; \delta_{sr})\;, 
\end{equation} 
such that the set $\{\frac{1}{d} \Pi_s\}_{s=1}^{d^2}$ forms an informationally complete POVM, i.e., a set of $d^2$ POVM elements that span the whole Hilbert space.
Then, the set $ \{\Pi_s\}_{s=1}^{d^2} $ is called symmetric informationally complete (SIC) and referred to as SICs.
\label{def:true_def_SIC}
\end{mydef}

SICs are a special case of generalized measurements on a Hilbert space like MUBs, therefore considering the full set of MUBs or SICs needs to be equivalent for entanglement detection because both are so-called $2$-designs~\cite{bae_linking_2019}. In contrast to MUBs, the existence of SICs is conjectured for any $d$, the largest dimension for which an exact solution has been found is $d=323$~\cite{scott_sics_2017}.

\setlength{\arrayrulewidth}{0.5pt}
\begin{table}
\begin{center}
\begin{tabular}{||c|c|c|c|c||c|c|c||}
\hhline{|t=====:t:===t|}
&\multicolumn{4}{c||}{$\vphantom{\biggl\lbrace}$ Lower Bounds $L$}&\multicolumn{3}{c||}{$\vphantom{\biggl\lbrace}$ Upper Bounds $U$}\\
\hhline{||-|-|-|-|-||-|-|-||}
&$\cellcolor{gray!20}d=2$&$\cellcolor{gray!20}d=3$&\multicolumn{2}{c||}{\cellcolor{gray!20}$d=4$}& \cellcolor{gray!20}$d=2$&\cellcolor{gray!20}$d=3$&\cellcolor{gray!20}$d=4$\\
\hhline{||-|-|-|-|-||-|-|-||}
$m$&$\vphantom{\biggl\lbrace} L_{m,2}^{(MUB)}$ ~&~ $L_{m,3}^{(MUB)}$ ~&~ $L_{m,4}^{- (MUB)}$ ~&~ $L_{m,4}^{+ (MUB)}$  ~&~ $U_{m,2}^{(MUB)}$ ~&~ $U_{m,3}^{(MUB)}$
~&~  $U_{m,4}^{(MUB)}$    \\
\hhline{|:=====||===:|}
$\vphantom{\big\lbrace}2$  & $0.5$& 0.211& 0  & 0  & 1.5 & 1.333 & 1.125  \\ \hhline{||-|-|-|-|-|-|-|-||}
 $\vphantom{\big\lbrace}3$  &  1& 0.5& \textbf{0.25} & \textbf{0.5}   & 2    & 1.666 &  1.5 \\ \hhline{||-|-|-|-|-|-|-|-||}
$\vphantom{\big\lbrace}4$ & & 1& 0.5  & 0.5          &   &    2  & 1.725 \\ \hhline{||-|-|-|-|-|-|-|-||}
$\vphantom{\big\lbrace}5$ & & & 1  & 1  &  &  &  2 \\
\hhline{|b:=====:b:===:b|}
\end{tabular}
\end{center}
\caption{[Table taken from Ref.~\cite{bae_linking_2019}.] This table shows the results of the optimization over all separable states for the observable~ \eqref{MUBKriterion}, which is based on a set of $m$ MUBs. The upper bounds $U$ only depend on the number $m$ of MUBs; however, the lower bounds $L$ are starting to depend on the choice of the MUBs for $d\geq 4$. Results for higher dimensions can be found in Ref.~\cite{hiesmayr_detecting_2021}.}
\label{tab:MUB}
\end{table}

Following the idea behind entanglement witnesses based on MUBs (presented in Sec.~\ref{sec:MUBwitness}), we can construct a witness based on a set of complete basis-POVMs in the same manner.

In an experiment, the following correlation functions $C$ are bounded from below and from above:
\begin{eqnarray}
\label{MUBKriterion}   L^{MUB}_{m,d}\;\leq\; &C_{m}:=\sum_{k=1}^{m}\sum_{i=0}^{d-1}\;\myTr(|i_k\rangle\langle i_k|\otimes|i_k\rangle\langle i_k|\;\rho)&\;\leq\;U^{MUB}_{m,d}\\
\label{SICKriterion} L^{SIC}_{\tilde{m},d}\;\leq\; &C_{\tilde{m}}:=\sum_{k=1}^{\tilde{m}}\;\myTr(|s_k\rangle\langle s_k|\otimes|s_k\rangle\langle s_k|\;\rho)&\;\leq\;U^{SIC}_{\tilde{m},d} 
\end{eqnarray}
with $|\langle i_k|j_l\rangle|^2= \frac{1}{d}  (1-\delta_{k,l})+\delta_{k,l} \delta_{i,j}$ and $|\langle s_k|s_j\rangle|^2=\frac{1}{d+1} (1-\delta_{k,j})+\delta_{k,j}$.
Here $U$ and $L$ denote the bounds over all separable states, i.e., a violation of those inequalities detects entanglement. The bounds depend on the number of MUBs $m$ and on the number of SICs $\tilde{m}$ that are applied, respectively. Note that $[L, U]$ is the separability window~\eqref{eq:separabiltywindow} of the pair of mirrored witnesses which can be constructed from $C_m$ or $C_{\tilde{m}}$ (see Sec.~\ref{sec:mirroredEW}).

The upper bound for the MUBs has been derived for all dimensions to be $U^{M}_{m,d}=1+\frac{m-1}{d}$ and is independent of the choice of the set of MUBs~\cite{spengler_entanglement_2012}. A particular reordering of the MUB or SIC projectors for one of the two subsystems in $C_m$ or $C_{\tilde{m}}$ changes, in general, the upper and lower bounds, which can be exploited to detect bound entanglement (cf. Sec.~\ref{sec:MUBwitness}). In Tab.~\ref{tab:MUB} and Tab.~\ref{tab:SICs}, taken from Ref.~\cite{bae_linking_2019}, we present results of the lower and upper bounds on the expression $C_m$ and $C_{\tilde{m}}$ as defined above. In Sec.~\ref{BEexperiment}, we present an experiment that utilizes the above expression $C_m$ based on MUBs to detect bound entanglement (by applying a particular reordering in one subsystem).

For the MUB case, the lower bounds depend on the choice of a set of MUBs for dimensions greater than $d=4$. Here, inequivalent sets of $m$ MUBs exist~\cite{grassl_small_2017}. Optimizing $C_m$ over all separable states, one finds the smallest and the largest lower bound $L^\pm$ if $m=3$ (Tab.~\ref{tab:MUB}). Thus, an experimenter needs to optimize the MUB observables with respect to the source's state or try out all options. Otherwise, they may not be able to detect the entanglement.

In the SIC case, in strong contrast to MUBs, different sets of SICs affect lower and upper bounds, $L^\pm, U^\pm$, for some $\tilde{m}$ numbers of SICs (Tab.~\ref{tab:SICs}), and it starts already with the first nontrivial dimension, i.e., $d=3$. In practice, this means that witnesses based on SIC observables need local information about the state to be optimized.

In Fig.~\ref{fig:MUBSIC}, we visualize the entanglement witnesses depending on different numbers of MUBs or SICs for a slice in the Magic Simplex for two qutrits, \eqref{eq:magic_simplex}, of the mixture of the totally mixed states and two Bell states. 
Both MUBs and SICs are equivalent over a 2-design if they form a complete set, and therefore, the corresponding witnesses are also equivalent~\cite{bae_linking_2019}.
In this slice, they represent the optimal witnesses for the family of isotropic states $\{\rho_{iso}(p)=p P_{k,l}+ (1-p) \frac{1}{9}\mathbb{1}_9\}_{k,l=0}^{d-1}$. As can be deduced from Fig.~\ref{fig:MUBSIC}, fewer MUBs or SICs detect a smaller region of the entangled states; however, which region is detected is very different for MUBs and SICs. Thus, incomplete sets of those two POVMs can reveal different entangled states in the Hilbert space. In this particular case, the MUB witness ``rotates'' in this slice, whereas for the SIC witness, we find a ``parallel shift'' for this slice in the Magic Simplex.

Less is known for higher dimensions since there is little knowledge about the constructions of the MUBs and SICs in general dimensions. For dimensions up to $d=9$, results for MUBs are presented in Ref.~\cite{hiesmayr_detecting_2021}.

\setlength{\arrayrulewidth}{0.5pt}
\begin{table}
\begin{center}
\begin{tabular}{||c | c | c | c || c | c | c||}
\hhline{|t====:t:===t|}
&\multicolumn{3}{c||}{$\vphantom{\biggl\lbrace}$ Lower Bounds $L$}&\multicolumn{3}{c||}{$\vphantom{\biggl\lbrace}$ Upper Bounds $U$}\\
\hhline{||-|-|-|-||-|-|-||}
&\cellcolor{gray!20}$d=2$&\multicolumn{2}{c||}{\cellcolor{gray!20}$d=3$}& \cellcolor{gray!20}$d=2$&\multicolumn{2}{c||}{\cellcolor{gray!20}$d=3$}\\
\hhline{||-|-|-|-||-|-|-||}
 $\tilde{m}~ $   &  $\vphantom{\biggl\lbrace}  ~L_{\tilde{m},2}^{(SIC)}$   ~ & ~$ L_{\tilde{m},3}^{ - (SIC ) } $ ~ &~ $L_{\tilde{m}, 3}^{+(SIC)} $  ~&  ~ $U_{\tilde{m} ,2}^{(SIC)}$  ~ &~ $ U_{\tilde{m}, 3}^{+ (SIC ) } $ ~ & ~  $U_{\tilde{m}, 3}^{- (SIC)}$ \\
\hhline{|:====||===:|}
 $\vphantom{\big\lbrace}2$  ~&  0 & -      & - &      1.244      & - & -  \\ \hhline{||----||---||}
 $\vphantom{\big\lbrace}3$  ~& 0.266 &  0      & 0 &      1.333      & \textbf{1.254} & \textbf{1.125}  \\ \hhline{||----||---||}
 $\vphantom{\big\lbrace}4$  ~&  0.666 & 0      & 0 &    1.333    &\textbf{1.4} &  \textbf{1.25}  \\ \hhline{||----||---||}
 $\vphantom{\big\lbrace}5$  ~&  - &     0  & 0   &  -      &\textbf{1.463}  &  \textbf{1.400} \\\hhline{||----||---||}
 $\vphantom{\big\lbrace}6$  ~&   -     & \textbf{0}      & \textbf{0.112}&    - & \textbf{1.5}   &  \textbf{1.482}  \\ \hhline{||----||---||}
 $\vphantom{\big\lbrace}7$  ~&   -    & 0.15 & 0.15     &  - & 1.5      &   1.5 \\ \hhline{||----||---||}
 $\vphantom{\big\lbrace}8$  ~&   -    & 0.375  & 0.375    &    - & 1.5      &    1.5 \\ \hhline{||----||---||}
$\vphantom{\big\lbrace}9$  ~&    -   & 0.75   &  0.75  &     - & 1.5      &    1.5 \\
\hhline{|b====:b:===b|}
\end{tabular}
\caption{[Table taken from Ref.~\cite{bae_linking_2019}.] The lower and upper bounds, $L$ and $U$, respectively, for entanglement witnesses based on $\tilde{m}$ SICs are shown for $d=2,3$. We use the Hesse SICs for $d=3$ defined in Ref.~\cite{stacey_sporadic_2017}. Note that for $d=2$, the smallest and largest lower/upper bounds, $L^\pm/U^\pm$, are equal, whereas for $d>2$, they may differ (marked in bold). This strongly contrasts MUBs, where only the lower bounds can vary.}\label{tab:SICs}
\end{center}
\end{table}

\begin{figure}
    \centering
    \includegraphics[width=1\linewidth]{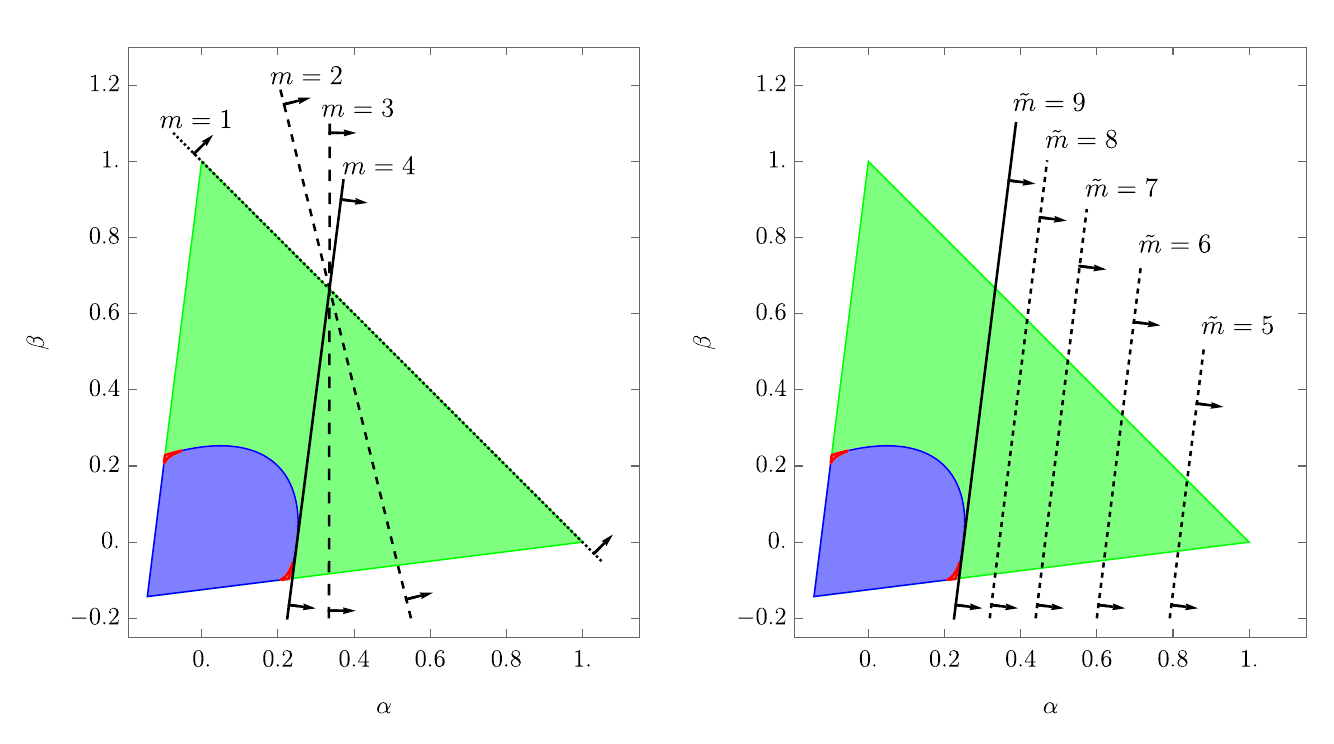}
\caption{These graphics show a $2$-dimensional slice through the $8$-dimensional Magic Simplex for two qutrits for the state family $\rho_A=\frac{1-\alpha-\beta}{9}\, \mathbb{1}_9+\alpha\,P_{0,0}+\beta\, P_{1,0}$. The green-colored region corresponds to NPT states. The red area corresponds to PPT states for which the quasi-pure approximation~(cf. Sec.~\ref{quasipure}) detects bound entanglement. On the left, the MUB witnesses for different numbers $m$ of MUBs are shown, and on the right, the SIC witnesses for different numbers $\tilde{m}$ of SICs. None of the witnesses can detect bound entanglement; however, these graphics illustrate that those observables detect different sets of entangled states.}
\label{fig:MUBSIC}
\end{figure}

Due to the mentioned equivalence over a $2$-design, all bound entangled states that are detected by entanglement witnesses based on MUBs are also detected by the corresponding SIC witnesses. In contrast to MUBs (see Thm.~\ref{notcompletestetofMUBs}), using an incomplete set of SICs, no example of a bound entangled state that is detected by the corresponding witness has been found so far. This also holds for different generalizations of SICs (see also Sec.~\ref{sec:SIC_criteria}).

In summary, sets of MUBs or SICs are typically easily constructed observables in experiments (see Sec.~\ref{BEexperiment}) and, as shown, capable of detecting bound entangled states since, in general, non-decomposable entanglement witnesses can be constructed. Moreover, those entanglement witnesses can be mirrored, giving rise to two bounds, $L$ and $U$.
Generally, these bounds depend on the specific set of MUBs. Interestingly, the upper bounds for the MUB witnesses depend solely on the number of MUBs.

\subsection{Criteria Based on a Correlation Tensor}
\label{sec:cor_ten_crit}
In Ref.~\cite{sarbicki_family_2020}, a family of entanglement criteria is introduced that is based on the Bloch representation of bipartite quantum states $\rho \in \denop(\hs_A \otimes \hs_B)$ of dimension $d = d_Ad_B$.

Let $G^{A}_\alpha$ be an orthonormal basis of $\linop(H_{A})$ with respect to the Hilbert-Schmidt inner product $\myTr(G^{A \dagger}_\alpha G^{A}_\beta)=\delta_{\alpha,\beta}$ and $G^{A}_0 = \frac{1}{d_{A}}\id_{A}$ (similar for system $B$). The Bloch representation of a state $\rho$ is then 
\begin{align}
    \rho &= \frac{1}{d_A d_B} \id_A \otimes \id_B  
    + \sum_{i=1}^{d_A^2-1} \frac{r_i^A}{d_B} G_i^A \otimes  \id_B
    + \sum_{j=1}^{d_B^2-1} \frac{r_j^B}{d_A} \id_A  \otimes G_j^B
    + \sum_{i=1}^{d_A^2-1} \sum_{j=1}^{d_B^2-1} t_{i,j} G_i^A \otimes G_j^B \\
    &=: \sum_{\alpha = 0}^{d_A^2-1} \sum_{\beta = 0}^{d_B^2-1}  C_{\alpha, \beta} G_\alpha^A \otimes G_\beta^B \;.
\end{align}
Let us define two diagonal matrices $D_x^A := diag(x, 1, \dots, 1)$ and $D_y := diag(y, 1, \dots, 1)$ of dimension $d_A^2$ and $d_B^2$, respectively. With the matrix $C=(C_{\alpha, \beta})$, the following entanglement criterion can be formulated \cite{sarbicki_family_2020}:
\begin{mytheorem}[Correlation Tensor Criterion]
    \begin{align}
        \exists ~x,y \geq 0~with~ \myTr \left( \sqrt{D_x^A C D_y^B  (D_x^A C D_y^B )^\dagger} \right)  > \sqrt{\frac{d_A -1 +x^2}{d_A}} \sqrt{\frac{d_B -1 +y^2}{d_B}} 
        \implies \rho \in \ENT.
     \end{align}    
\end{mytheorem}
This family of criteria reproduces several other known entanglement criteria. For $x=y=1$, it reproduces the realignment criterion (cf. Sec.~\ref{sec:realignment}), while for $x=y=0$ it is equivalent to the so-called \emph{de Vicente criterion} \cite{de_vicente_separability_2007}. 
It is shown in Ref.~\cite{sarbicki_family_2020} that the correlation tensor criterion is also equivalent to the ESIC criterion (cf. Sec.~\ref{sec:SIC_criteria}) for $(x,y)=(\sqrt{d_A+1}, \sqrt{d_B+1})$. Moreover, criteria based on this correlation tensor allow deriving a corresponding class of entanglement witnesses (cf. Sec.~\ref{sec:EWs}; Ref.~\cite{sarbicki_family_2020}).

\subsection{Realignment Criterion} 
\label{sec:realignment}
The realignment criterion (also named cross norm criterion) is presented in Ref.~\cite{chen_matrix_2003} in matrix notation and in Ref.~\cite{rudolph_further_2005} in Dirac notation. It is based on the singular values of a realigned bipartite density matrix. 

The realignment operator $R$ can be defined via its action on the basis states of a bipartite system
\begin{align}
      R: |i\rangle \langle j | \otimes |k\rangle \langle l| 
      \mapsto 
      (|i\rangle \langle j | \otimes |k\rangle \langle l|)_{R} 
      := |i\rangle \langle k | \otimes |j\rangle \langle l |
\end{align}
and extends to general density matrices $\rho \in \denop(\hs_A \otimes \hs_B)$ by linearity:
\begin{align}
    R: \rho \mapsto \rho_R 
\end{align}
In Ref.~\cite{chen_matrix_2003, rudolph_further_2005} it is shown that if $\rho \in \SEP$, then $\myTr(\sqrt{\rho^\dagger_R \rho_R}) \leq 1$. Conversely, if the sum of singular values of a realigned density matrix is greater than one, then the state must be entangled. Thus, the realignment or cross norm criterion reads:
\begin{mytheorem}[Realignment Criterion / Cross Norm Criterion]
\label{thm:realignment_criterion}
\begin{align}
    \myTr(\sqrt{\rho^\dagger_R \rho_R}) > 1 \implies \rho \in \ENT.
\end{align}    
\end{mytheorem}

The realignment criterion is neither stronger nor weaker than the PPT criterion (Thm.~\ref{thm:PPT-crit}), i.e., it does not detect all NPT entangled states but can detect PPT entanglement.

Considering Bell-diagonal states (cf. Sec.~\ref{sec:bell_system}), the realignment criterion seems very effective for the detection of PPT entanglement for $d=3$ and $d=4$ (cf. Refs.~\cite{popp_comparing_2023,popp_special_2024}).

\subsection{Entanglement Criteria based on SICs and GSICs}
\label{sec:SIC_criteria}
SICs, i.e., symmetric informationally complete POVMs, provide another way to derive entanglement criteria.
Let $\rho \in \denop(\hs_A \otimes \hs_B)$, $\Pi_\alpha^A$ and $\Pi_\beta^B$ be SICs (Def.~\ref{def:true_def_SIC}) for $\hs_A$ and $\hs_B$ and define the correlation tensor $\mathcal{P}$ with elements 
\begin{align}
\label{eq:corr_tensor}
    \mathcal{P}_{\alpha, \beta} := \myTr( \Pi_\alpha^A \otimes \Pi_\beta^B\;\rho).
\end{align} In Ref.~\cite{shang_enhanced_2018}, the following entanglement criterion is derived:
\begin{mytheorem}[ESIC Criterion]
    \begin{align}
        \myTr(\sqrt{\mathcal{P}^\dagger \mathcal{P}})
         > \frac{2}{\sqrt{d_A(d_A+1)d_B(d_B+1)}} \implies \rho \in \ENT.
    \end{align}
\end{mytheorem}
A similar criterion can be stated for \emph{general symmetric informationally complete POVMs (GSICs)}. Compared to SICs, for GSICs the POVM elements $\lbrace \Pi^{A}_{\alpha} \rbrace$ are not required to be of rank-1, but instead satisfy 
\begin{align}
    \myTr(\Pi^{A}_\alpha \Pi^A_{\alpha}) &= a_{A} \\
    \myTr(
        \Pi^{A}_{\alpha}
        \Pi^{A}_{\gamma}
        )
       &= \frac{1-d_{A}a_{A}} {d_{A}(d^2_{A}-1)}, ~\text{for } \alpha \neq \gamma
\end{align}
for system $\hs_A$ and similarly for POVM elements $\lbrace \Pi^{B}_{\beta} \rbrace$ for $\hs_B$. With the corresponding correlation tensor \eqref{eq:corr_tensor}, the GSIC criterion reads \cite{lai_entanglement_2018}:
\begin{mytheorem}[GSIC Criterion]
    \begin{align}
        \myTr(\sqrt{\mathcal{P} \mathcal{P}^\dagger}) > 
        \sqrt{\frac{a_Ad^2_A+1}{d_A(d_A+1)}} \sqrt{\frac{a_Bd^2_B+1}{d_B(d_B+1)}} 
        \implies \rho \in \ENT.
    \end{align}
\end{mytheorem}

\subsection{Concurrence and Quasi-Pure Approximation}\label{quasipure}

The quasi-pure approximation $\mathcal{C}_{qp}$ \cite{bae_detection_2009} of the concurrence $\mathcal{C}$, presented in Ref.~\cite{wootters_entanglement_1998}, is another useful sufficient criterion to detect entanglement in systems of finite dimension. 

The concurrence, introduced in Ref.~\cite{wootters_entanglement_1998}, is an entanglement monotone closely related to the entanglement of formation (cf. Def.~\ref{def:entanglement_of_formation}). It can only be explicitly calculated for a pair of qubits, while for systems of larger size, only numerical solutions or estimates exist \cite{hiesmayr_multipartite_2008,wootters_entanglement_1998, mintert_concurrence_2005, mintert_measures_2005}. The quasi-pure approximation provides a lower bound for the concurrence of a quantum state
\begin{align}
    \mathcal{C}_{qp}(\rho) \leq \mathcal{C}(\rho).
\end{align}
Given the spectral decomposition  $\rho = \sum_{i=0}^{d-1} \lambda_i\; |\psi_i\rangle \langle\psi_i|$, a matrix $\mathcal{T}$ is constructed, whose singular values determine the quasi-pure approximation. 
Let $A := 4 \sum_{i<j, k<l} (|ikjl\rangle-|jkil\rangle-|iljk\rangle+|jlik\rangle) \times (h.c.)$ ($h.c$ being shorthand for hermitian conjugate) and let $|\psi_0\rangle$ be the dominant eigenvector. With $\ket{\chi} \propto A \ket{\psi_0} \otimes \ket{\psi_0}$, $\mathcal{T}$ is defined as
\begin{align}
    \label{eq:T}
    \mathcal{T}_{ij} := \sqrt{\lambda_i \lambda_j} = (\langle \psi_i | \otimes \langle \psi_j|) | \chi \rangle.
\end{align}
Let $\mathcal{S}_i$ be the singular values of $\mathcal{T}$ and denote the largest one by $\mathcal{S}_0$.
The quasi-pure approximation of the concurrence is defined as
\begin{align}
    \mathcal{C}_{qp}(\rho) := \text{max} \left(0,~\mathcal{S}_0 - \sum_{i>0} \mathcal{S}_i\right)\;.
\end{align}
For any state in the Magic Simplex, i.e., standard Bell-diagonal states (cf. Sec.~\ref{sec:magic_simplex}), there exists an explicit form for the singular values of $\mathcal{T}$ (cf. Ref.~\cite{bae_detection_2009}). 
Given a state of the form of Eq.~\eqref{eq:magic_simplex}, let $P_{n,m}$ be the Bell projector with the largest weight $c_{k,l}$. The singular values $\mathcal{S}_{k,l}$ are then explicitly given by
\begin{align}
    \mathcal{S}_{k,l} = \sqrt{\frac{d}{2(d-1)}c_{k,l}\left[\left( 1-\frac{2}{d}\right) c_{n,m} \delta_{k,n} \delta_{l,m} +\frac{1}{d^2}c_{2n-k,2m-l)} \right]}\;.
\end{align}
Any state satisfying $\mathcal{C}(\rho) > 0$ is entangled (cf. Ref.\cite{wootters_entanglement_1998}). Providing a lower bound, the quasi-pure approximation can therefore be directly used to detect entanglement:
\begin{mytheorem}[Quasi-Pure Concurrence Criterion]
\label{thm:qp_concurrence_criterion}
\begin{align}
    \mathcal{C}_{qp}(\rho) > 0 \implies \rho \in \ENT
\end{align}    
\end{mytheorem}
Like the realignment criterion (Thm.~\ref{thm:realignment_criterion}), this entanglement criterion can detect PPT entangled states but generally not all NPT entangled states. In Ref.~\cite{popp_comparing_2023}, it is shown that the overlap of PPT entangled Bell-diagonal states detected by both the quasi-pure concurrence criterion and the realignment criterion is relatively small compared to the joint number of detected bound entangled states. This indicates a complementary behavior for these criteria, as a state detected by one criterion is likely not detected by the other (cf. Sec.~\ref{boundBDS}).

\subsection{Range Criterion} \label{sec:range_crit}
The range criterion is of historical importance as it led to the discovery of the first known PPT entangled state.
It is formalized as follows \cite{horodecki_separability_1997}:
\begin{mytheorem}[Range Criterion] \label{thm:range_crit}
    Let $\rho\in\denop(\hs_A \otimes \hs_B)$ be separable with $\dim(\hs_A \otimes \hs_B) = d^2$.
    Then there exists a set $\{|\psi_i\rangle \otimes |\phi_j\rangle \}_{i,j=0}^{n-1} \subset \hs_A \otimes \hs_B$ of $n\leq d^2$ product vectors such that
    \begin{enumerate}[label=(\alph*)]
        \item the vectors $\{|\psi_i\rangle \otimes |\phi_j\rangle \}_{i,j=0}^{n-1}$ span the range of $\rho$ and
        \item the vectors $\{|\psi_i\rangle \otimes |\phi^*_j\rangle \}_{i,j=0}^{n-1}$ span the range of $\rho^\Gamma$.
    \end{enumerate}
    Conversely, if no such vectors exist, $\rho$ is entangled.
\end{mytheorem}
Here $\rho^\Gamma$ is the partial transposition of $\rho$ with respect to $\hs_B$, and $|\phi^*_j\rangle$ denotes the complex conjugate of $|\phi_j\rangle$ in the basis where the partial transposition is taken.
An analogous version of the theorem also holds for the partial transposition $\rho^{T_A}$ of $\rho$ with respect to $\hs_A$.
In this case, the vectors $\{|\psi^*_i\rangle \otimes |\phi_j\rangle \}$ span the range of $\rho^{T_A}$ if $\rho$ is separable.

\subsection{Reduction Criterion and more Criteria from PNCP Maps}
\label{sec:reduction_criterion}
Like the PPT criterion (Thm.~\ref{thm:PPT-crit}), the reduction criterion (Ref.~\cite{horodecki_reduction_1999}) is based on a PNCP map (cf. Sec.~\ref{sec:PNCP-maps_for_ent_detection}). Consider the PNCP map $\Lambda: \linop(\hs) \rightarrow \linop(\hs)$, $\dim(\hs) = d$, acting as
\begin{align}
    \Lambda(\sigma) := \id_d \myTr(\sigma) - \sigma \;.
\end{align}
$\Lambda$ is a PNCP map, so it provides a necessary condition for separability for density matrices:
\begin{align}
    \label{eq:PNCP_SEP_necessary_condition}
    \rho \in \mathrm{SEP} \implies (\Lambda \otimes \id_B) (\rho) \geq 0 \;.
\end{align}
The contrapositive of this statement can be formulated using the reduced density matrix $\rho_B := \pTrA(\rho)$:
\begin{mytheorem}[Reduction Criterion]
\label{thm:reduction_criterion}
\begin{align}
    \id_A \otimes \rho_B  - \rho \ngeq 0 \implies \rho \in \mathrm{ENT} \;.
\end{align}
\end{mytheorem}
Ref.~\cite{horodecki_reduction_1999} shows that the reduction criterion is weaker than the PPT criterion. Any state detected as entangled by the reduction criterion is also detected by the PPT criterion. Consequently, NPT entangled states exist that are not detected by the reduction criterion. Additionally, it is shown that all entangled states detected by the reduction criterion (Thm.~\ref{thm:reduction_criterion}) can be \emph{distilled}, so only states that are detected by the PPT criterion and not by the reduction criterion are candidates for NPT bound entangled states.

However, there exist structurally similar criteria based on PNCP maps that can detect bound entanglement. In Ref.~\cite{bhattacharya_generating_2021}, a family of PNCP maps is presented that allows to detect two PPT entangled qutrits due to Eq.~\eqref{eq:PNCP_SEP_necessary_condition}.

\section{Characteristics of Bound Entanglement} \label{charOfBE}

This chapter reviews several key properties and characteristics of bound entangled states.
The goal of this investigation is twofold.
The first is a logical follow-up to the last chapter, which summarizes our primary methods of detecting bound entanglement and shows that there are still many challenges.
In particular, a hierarchical structure between the presented entanglement criteria is absent.
Solving these issues requires further insight into the entanglement structure and the properties of bound entangled states.
In this regard, the following sections clarify the role of impurity for bound entangled states and present results on the (nonzero) volume of the set BE, showing that bound entanglement cannot be neglected in general.
Moreover, Sec.~\ref{sec:BE_and_ent_measures} introduces commonly used entanglement measures, and by showing the inequivalence of two such measures, the ``bound nature'' of the entanglement in bound entangled states is elucidated.
Furthermore, the relation of bound entanglement to Bell inequalities and special relativity is discussed.
The second goal is to address the role of bound entangled states in known quantum information processing tasks such as quantum teleportation or quantum key distribution.
The advantage of bound entangled resource states over fully separable ones is limited to some of these schemes.
Further insight into the structure of bound entanglement and its characteristics may contribute to a better understanding of entanglement as a resource and, consequently, to novel applications in quantum technologies.

\subsection{Rank of Bound Entangled States}

From the discussion in Sec.~\ref{sec:ent_dist}, it is clear that all bound entangled states must be mixed.
This follows from the fact that every entangled pure state has a Schmidt rank greater than one and can thus be distilled with a nonzero distillation rate.
Consequently, there are no pure, i.e., rank-1 bound entangled states.
The logical follow-up question is: what are the general restrictions on the rank for bipartite bound entangled states?

A first advance to answering this question is presented in Ref.~\cite{horodecki_rank_2003}.
The authors prove that if the rank of a bipartite quantum state $\rho$ is strictly less than the rank of either of its reduced density matrices $\pTrA(\rho)$ and $\pTrB(\rho)$, then $\rho$ is distillable.
The contrapositive of this statement is formalized by
\begin{mytheorem} \label{thm:rank_BE_undistillable}
    Let $\rho \in \denop(\hs_A\otimes\hs_B)$ be a bipartite quantum state and let $\rho_A = \pTrB(\rho)$ and $\rho_B = \pTrA(\rho)$ be its reduced states with respect to the Hilbert spaces $\hs_A$ and $\hs_B$, respectively.
    If $\rho$ is undistillable, then     
    $\mathrm{rank}(\rho) \geq \max\{\mathrm{rank}(\rho_A), \mathrm{rank}(\rho_B)\}$.
\end{mytheorem}
In specific cases even the strict inequality $\mathrm{rank}(\rho) > \max\{\mathrm{rank}(\rho_A), \mathrm{rank}(\rho_B)\}$ holds for undistillable states.
In particular, it is shown in Ref.~\cite{chen_distillability_2011} that this includes all states violating the PPT criterion (Thm.~\ref{thm:PPT-crit}).
Therefore, no NPT bound entangled states with $\mathrm{rank}(\rho) = \max\{\mathrm{rank}(\rho_A), \mathrm{rank}(\rho_B)\}$ exist.

As observed in Ref.~\cite{horodecki_rank_2003}, a consequence of Thm.~\ref{thm:rank_BE_undistillable} is that undistillable quantum states with rank $n$ have support on at most an $n\times n$-dimensional subspace of the underlying bipartite Hilbert space.
Consequently, an undistillable rank-2 state can be viewed as a qubit state, and in this case, Thm.~\ref{thm:qubit_distillation} implies that the state is separable.
Thus, there are no rank-2 bound entangled states.
Additionally, it is demonstrated in Ref.~\cite{horodecki_operational_2000} that no rank-3 bound entangled states satisfying the PPT criterion exist.
The authors of Ref.~\cite{chen_rank-three_2008} show that this also holds for potential rank-3 NPT bound entangled states.

These findings leave open the possibility of bound entangled states with rank 4 or greater.
From explicit constructions (cf. Sec.~\ref{sec:construction} and Ref.~\cite{brus_construction_2000, divincenzo_unextendible_2003}), we know that such rank-4 states do exist.
The authors of Ref.~\cite{bej_unextendible_2021, bandyopadhyay_non-full-rank_2005} moreover argue that one can use these states of minimum rank to construct bound entangled states of rank 5 up to the full dimension $d_A d_B$ of the bipartite Hilbert space.
The above results can be summarized by
\begin{mytheorem} \label{thm:rank_BE}
    No bound entangled states exist with rank 1, 2, or 3.
    Moreover, bound entangled states with rank 4 and greater can be explicitly constructed.
\end{mytheorem}

\subsection{Frequency of Bound Entangled States} \label{sec:frequency_bound}

A natural question to ask about the characteristics of entanglement classes (NPT, PPT, SEP, BE, FE) is: Given a certain quantum system and a randomly chosen quantum state, what are the probabilities of that state belonging to a certain entanglement class? Or equivalently: How big are the relative volumes of NPT, PPT, SEP, BE, and FE within the quantum state space? 

Even for the Hilbert space of bipartite quantum states $\denop(\hs_A \otimes \hs_B)$, this question is not straightforward to answer as the results strongly depend on the properties of the subsystems. 
In fact, it is shown in Ref.~\cite{horodecki_bound_2003} that for infinite-dimensional (i.e., continuous-variable) systems, bipartite quantum states are generically distillable. At the same time, Ref.~\cite{zyczkowski_volume_1998} demonstrates that the volume of undistillable states is finite for finite-dimensional systems. 
A fundamental difficulty in quantifying the volumes within the set of all quantum states is that calculations generally depend on the chosen probability measure. The choice of a ``natural'' measure on $\denop(\hs_A \otimes \hs_B)$ is not clear (cf. Ref.~\cite{slater_priori_1999}). In Sec.~\ref{genBoundEntStates}, we present results related to the frequency of bound entangled states that can be obtained despite this challenge.

The geometric representation of Bell-diagonal states as a $(d^2-1)$-simplex in real space (cf. Sec.~\ref{sec:bell_system}) allows determining frequencies of states in the entanglement classes SEP, BE and FE in this subset of states. These results are compared for the standard and generalized Bell-diagonal states in Sec.~\ref{boundBDS}.

Since only PPT bound entangled states are detected, reported frequencies of states in BE correspond to the relative volume of $\PPT\cap \ENT$. This provides a lower bound for the volume of BE because the volume of NPT bound entangled states might be nonzero in certain systems. So far, however, no bound entangled state with negative partial transposition is known.

\subsubsection{General Bipartite Bound Entangled States} \label{genBoundEntStates}

For general systems, the relative amount of bound entangled states in the bipartite Hilbert space strongly depends on the dimensions of the subsystems.

As mentioned above, a unique, ``natural'' measure for the space $\denop(\hs_A \otimes \hs_B)$ is not clear, a priori. Let $d_{AB} = d_A d_B$ be the dimension of the bipartite system. In Ref.~\cite{zyczkowski_volume_1998, zyczkowski_volume_1999}, a family of product measures is proposed, which follows from the spectral decomposition of any density matrix $\rho = U C U^{\dagger}$. The unitary matrix $U \in \mathcal{U}(d_{AB})$ ($\mathcal{U}(d_{AB})$ denotes the set of unitary operators on $\hs_A \otimes \hs_B$) is not unique. $C=diag(c_1,\dots,c_{d_{AB}})$ is the diagonal matrix of eigenvalues of $\rho$ with the trace condition $\sum_i c_i = 1$. A measure $\mu$ on the space $\denop(\hs_A \otimes \hs_B)$ can then be defined as $\mu = \nu_C \cross \nu_H$. Here, $\nu_C$ is a measure on the $(d_{AB}-1)$-dimensional simplex formed by the eigenvalues $c_i$. $\nu_H$ is the Haar measure defined on $\mathcal{U}(d_{AB})$. While the Haar measure $\nu_H$ can indeed be motivated as a ``natural'' measure on $\mathcal{U}(d_{AB})$ due to its invariance under unitary transformations, the choice of $\nu_C$ is more arbitrary. The normalized Lebesgue measure $\mathcal{L}_{d_{AB}-1}$ corresponds to a uniform and rotationally invariant distribution on the simplex as a submanifold of $\mathbb{R}^{d_{AB}}$. However, based on information-theoretic arguments, other distributions and induced measures can also be considered (cf. Ref.~\cite{slater_priori_1999, zyczkowski_volume_1999}). 

Using the product measure $\mu = \mathcal{L}_{d_{AB}-1} \cross \nu_H$, it is shown in Ref.~\cite{zyczkowski_volume_1998} that there exists a neighborhood of the maximally mixed state $d^{-1}_{AB} \id_{d_{AB}}$ which only contains separable states and has nonzero measure (less than one). Moreover, it is shown that the measure of PPT entangled states is larger than zero if $d_{AB}>6$. As those states are bound entangled, the same holds for the set BE. In Ref.~\cite{zyczkowski_volume_1999}, it is argued that this is also true for any nonsingular measure $\mu$ for finite dimension $d_{AB}$, and we can state the following:
\begin{mytheorem}
    The measures $\mu(\SEP)$ of separable states in $\denop(\hs_A \otimes \hs_B)$ for finite $d_{AB}$-dimensional Hilbert space $\hs_A \otimes \hs_B$ is nonzero and less than one. If $d_{AB}> 6$, the same holds true for the measure $\mu(\BE)$ of bound entangled states.
\end{mytheorem}
Investigations in Ref.~\cite{slater_priori_1999} show that the volumes depend on the chosen measure. Nonetheless, numerical results of Ref.~\cite{zyczkowski_volume_1999} indicate that there exists a dominant dependence of the volume of states with positive partial transposition (PPT) on the dimension $d_{AB}$ for all nonsingular measures. In particular, it is shown that the volume of PPT decreases approximately exponentially with the dimension: $\mu(\PPT) \propto K_{\mu} \exp(-t_{\mu}d_{AB}) $ for some constants $K_{\mu}$ and $t_{\mu}$ depending on the chosen measure $\mu$. From $\mu(\PPT) = \mu(\SEP) + \mu(\BE \cap \PPT)$ it follows that also the volumes of separable and PPT entangled states decrease approximately exponentially with the dimension.

Finally, and exceeding the primary scope of this article, we mention a result showing that bipartite bound entanglement occupies a finite volume only if the Hilbert space is finite-dimensional. Consider infinite-dimensional bipartite Hilbert spaces of square-integrable functions ($L^2(\mathbb{R})$), i.e., continuous-variable systems. In Ref.~\cite{horodecki_bound_2003}, the following result is shown:
\begin{mytheorem}
    The subset of undistillable states is nowhere dense in the set of all continuous-variable states in $\denop(\hs_A \otimes \hs_B)$ for $\hs_A \cong \hs_B \cong L^2(\mathbb{R})$.
\end{mytheorem}
The theorem directly implies that quantum states are generically distillable in continuous-variable systems. Because the set of undistillable states does not contain any open ball, it is of measure zero independent of the chosen measure. As subsets of undistillable states, the sets SEP, PPT, and BE are nowhere dense for infinite-dimensional systems. Note that this result is not limited to $\PPT \cap \BE$ but holds for possibly existing NPT bound entangled states as well.

\subsubsection{Bell-diagonal States} \label{boundBDS}

In this section, results of Ref.~\cite{popp_almost_2022, popp_comparing_2023,popp_special_2024} regarding the frequency of bound entangled states of Bell-diagonal systems (Sec.~\ref{sec:bell_system}) are summarized. It is shown that a significant share of states is bound entangled for small dimensions of the subsystems ($d = d_A = d_B \geq 3$). Interestingly, the standard Bell-diagonal system, the Magic Simplex \eqref{eq:magic_simplex}, stands out among all generalized Bell-diagonal systems \eqref{eq:gen_BDS}. This indicates that the entanglement structure, including the frequency of bound entanglement, strongly depends on the choice of the Bell basis.

As discussed in the previous section, the choice of a measure to determine volumes within the set of quantum states is not unique since it depends on the designation of a ``natural'' distribution of states. Following a geometric approach to entanglement via the identification of $d^2$-dimensional Bell-diagonal states with points in the standard $(d^2-1)$-simplex, specific volumes of entanglement classes can be well approximated by Monte Carlo simulations. More precisely, the normalized Lebesgue measure $\mathcal{L}_{d^2-1}$ can be approximated by the frequencies of a uniformly distributed sampling of points in the $(d^2-1)$-simplex. As described in Ref.~\cite{popp_almost_2022}, the volumes (regarding the induced uniform measure) of certain subsets of Bell-diagonal states can be determined as follows. First, a sufficiently large set of uniformly distributed points in the simplex is generated and associated with states in $\denop(\hs_A \otimes \hs_B)$. These states are then classified as separable (SEP), PPT entangled (PPT-ENT), or NPT entangled (NPT) using the methods presented in Sec.~\ref{sec:detection of BE}. The relative frequencies approximate the volumes of these subsets in the state space. Note that due to the challenge of detecting bound entangled states, not all PPT states are classified as SEP or PPT-ENT.

In Ref.~\cite{popp_almost_2022}, this method is used to analyze bipartite standard Bell-diagonal qutrits ($d=3$) for the standard Bell basis, i.e., in the Magic Simplex, and to determine the entanglement class of more than $95\%$ of states. In Ref.~\cite{popp_comparing_2023}, similar methods are applied for higher dimensions. As for general systems (cf. Sec.~\ref{genBoundEntStates}) the relative volume of PPT states decreases with the dimension (cf. Fig.~\ref{fig:freq_BDS}). For the standard Bell basis, the relative volume of PPT decreases from $50\%$ for $d=2$ to less than $1\%$ for $d\geq 6$. 
For $d=3$ and $d=4$, also a significant amount of PPT Bell-diagonal states can be detected as bound entangled. Including the uncertainty of PPT states that can not be identified as separable or entangled, the relative volume of PPT-ENT lies in the interval $[5.4\%,7.4\%]$ for qutrits and in $[0.2\%,2.8\%]$ for ququarts.
The results for the obtained frequencies of states in NPT, PPT, SEP, and PPT-ENT are summarized in Tab.~\ref{tab:frequencies} and visualized in Fig.~\ref{fig:freq_BDS}.

\begin{table}[h]
    \centering
    \begin{tabular}{c|c|c|c|c}
         & $d=2$ & $d=3$ & $d=4$ & $d=5$ \\
       \hline
       rel. NPT  & $0.50$ & $0.61$ & $0.884$ & $0.93$ \\
       \hline
       rel. PPT  & $0.50$ & $0.39$ & $0.116$ & $0.07$ \\
       \hline
       rel. SEP  & $0.50$ & $0.33 \pm 0.01$ & $0.101 \pm 0.013$ & $-$\\
       \hline
       rel. PPT-ENT & $0.0$ & $0.06 \pm 0.01$ & $0.015 \pm 0.013$ & $-$\\
    \end{tabular}
    \caption{
    Relative volumes of entanglement classes in the Magic Simplex. For the dimensions $d=2,3,4,5$, the relative frequencies of NPT, PPT, SEP and PPT-ENT states among all standard Bell-diagonal states are shown. For $d=3$ and $d=4$, a significant amount of PPT entangled states are detected, but not all PPT states can be classified as SEP or PPT-ENT and the relative volumes of SEP and PPT-ENT lies in the given interval.
    }
    \label{tab:frequencies}
\end{table}

\begin{figure}[H]
    \centering
    \begin{center}
\includegraphics[width=0.65 \linewidth]{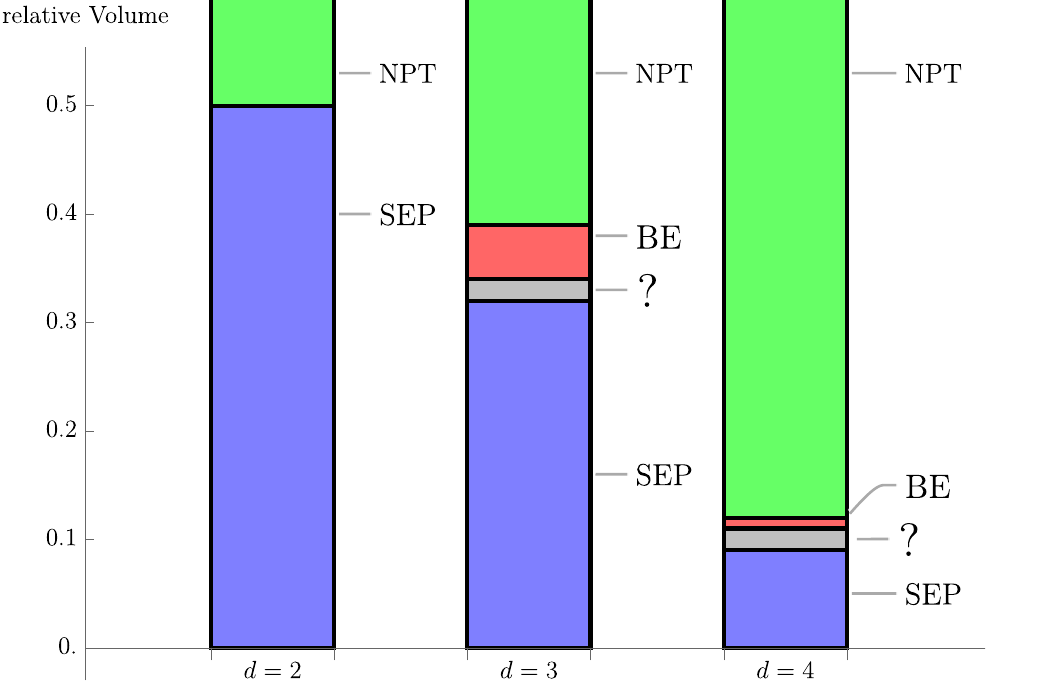}
\end{center}
    \caption{The frequencies of PPT and NPT states for Bell-diagonal systems of dimension $d=2,3,4
    $ are illustrated. The (not completely shown) green bar denotes the proportion of NPT states, while the blue, red, and gray bars denote PPT states. The fraction of separable and bound entangled states is depicted in blue and red, respectively. Unclassified PPT states are denoted by ``?'' (gray). One observes a significant decrease in PPT states for increasing dimension.}
    \label{fig:freq_BDS}
\end{figure}

Ref.~\cite{popp_special_2024} investigates generalized Bell-diagonal systems (cf. Sec.~\ref{sec:gen_bell_system}) and shows that the standard Bell basis has some very special properties that distinguish it from other generalized Bell bases. This is also reflected in the entanglement structure, particularly in the frequency of PPT bound entangled states among all generalized Bell-diagonal states. 
For thousands of generalized Bell bases for $d=3$, the standard Bell-diagonal system, i.e., the Magic Simplex, has the smallest relative volume of NPT states among the corresponding Bell-diagonal states. In particular, the relative volumes of PPT and PPT-ENT are generally smaller for generalized Bell-diagonal states than for the Magic Simplex. 
The influence of the chosen Bell basis on the sets PPT, NPT, PPT-ENT, and SEP can also be visualized by considering certain subsets of Bell-diagonal states. In Fig.~\ref{fig:gen_bell_families}, it is shown how the sets of the entanglement classes within a two-dimensional family of states significantly depend on the Bell basis.

Interestingly, when comparing different generalized Bell-diagonal systems, the relative frequency of PPT entangled states among all PPT states strongly correlates positively with the volume of PPT itself. This means that the frequency of PPT entangled states increases disproportionately with the PPT volume. The more PPT states exist, the higher the frequency of entangled states among them.

These observations indicate that the relations of basis states in a Bell basis strongly affect the entanglement structure of corresponding Bell-diagonal states, including the frequency of  PPT bound entangled states. 
The underlying group structure, characterizing the standard Bell basis and the Magic Simplex via the Weyl relations \eqref{eq:weyl_relations}, is connected with the frequency of PPT states and PPT entanglement for Bell-diagonal states. If generalized Bell bases are not related by such structure, then fewer Bell-diagonal PPT states and PPT entangled states can be observed, despite the fact that all states represent mixtures of maximally entangled Bell states.

\begin{figure}
    \centering
    \includegraphics[width=1 \linewidth]{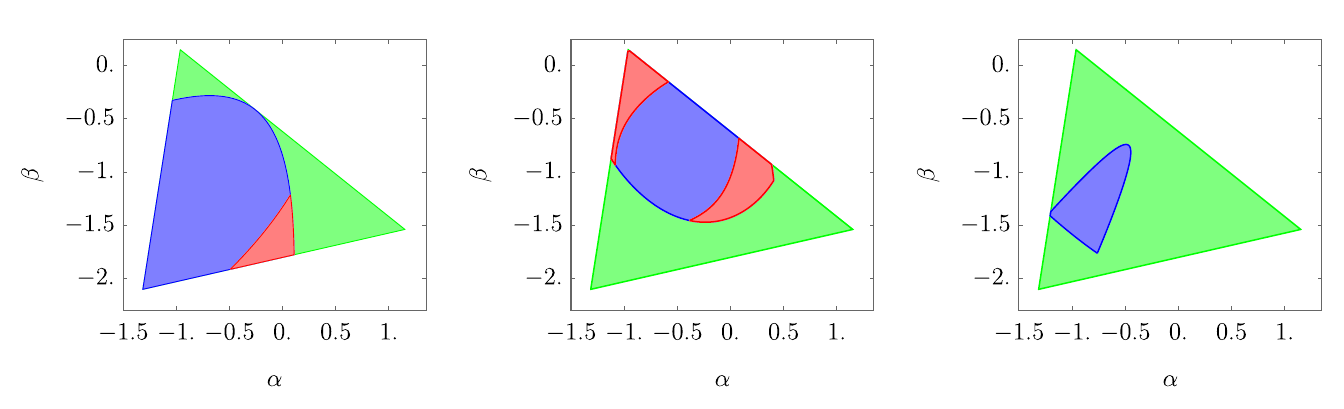}
\caption{
These graphics illustrate a particular slice (state family $B_1$; see caption of Fig.~\ref{fig:notcompleteMUB} for the definition) through the Magic Simplex of bipartite qutrits for the standard Bell basis on the left. The middle and the right panel show the same parameterized family for two generalized Bell bases. The green area corresponds to NPT entangled states. The red area depicts PPT bound entangled states that are detected by the realignment criterion~(Thm.~\ref{thm:realignment_criterion}). These three graphics illustrate the dependence of the entanglement structure of Bell-diagonal states on the Bell basis.
}
\label{fig:gen_bell_families}
\end{figure}

\subsection{Bound Entanglement and Information Processing} \label{sec:BE_and_info_processing}

Many quantum information protocols require some entangled states shared between two parties as a resource.
Most of them operate best if the entanglement is concentrated in maximally entangled pure states.
However, the resource states in experiments are generally mixed, so the standard protocols do not apply unconditionally.
In principle, if one has a close to unrestricted supply of free entangled states, any protocol can be preceded by entanglement distillation.
In this case, the protocol for maximally entangled states is applicable after the distillation process.
Consequently, all free entangled states can be utilized in specific quantum information processing tasks to outperform classical schemes.
This strategy is, by definition, impossible for bound entangled states, yet they are still inseparable and thus display characteristics different from separable states.
Therefore, the question arises whether they can serve as a resource for quantum information protocols. 
This section presents results that answer this question for most currently known applications in the negative.

\subsubsection{Quantum Teleportation and Distillation} \label{sec:BE_teleportation}

The most general setup of quantum teleportation is as follows \cite{horodecki_general_1999}.
Two parties, Alice and Bob, share a single resource state $\rho_{AB}\in\denop(\hs_A\otimes\hs_B)$ with $\dim(\hs_A) = \dim(\hs_B) = d \geq 2$.
Additionally, Alice is in possession of an unknown $d$-dimensional pure\footnote{Teleportation of an unknown mixed state $\sigma = \sum_i p_i |\phi_i\rangle\langle\phi_i|$ is readily obtained by linearity of the protocol.}  quantum state $|\phi\rangle\in\hs_{A'}$.
Teleportation aims to transmit $|\phi\rangle$ from Alice to Bob by some LOCC scheme, i.e., by only using local quantum operations and classical communication (cf. Sec.~\ref{sec:ent_dist}).
Thus, teleportation through the resource state $\rho_{AB}$ constitutes an LOCC-restricted quantum channel $\mathcal{E}_{\rho_{AB}}: \denop(\hs_{A'}) \rightarrow \denop(\hs_B)$.
In the following, we denote the (generally mixed) teleported state by $\sigma_\phi = \mathcal{E}_{\rho_{AB}}(|\phi\rangle \langle \phi|) \in \denop(\hs_B)$, thereby omitting its explicit dependence on $\rho_{AB}$ and the particular LOCC protocol.

A customary figure of merit for a given teleportation scheme is the \emph{fidelity} \cite{jozsa_fidelity_1994} between the teleported state $\sigma_\phi$ and the input state $|\phi\rangle$,
\begin{align} \label{eq:fidelity}
    F ( \sigma_\phi, |\phi\rangle) = \myTr\left(\sqrt{(|\phi\rangle\langle\phi|)^{1/2}\; \sigma_\phi\;(|\phi\rangle\langle\phi|)^{1/2} }\right)^2=\langle\phi|\, \sigma_\phi \, |\phi\rangle \;.
\end{align}
As the state $|\phi\rangle$ is unknown, it is necessary to integrate over all potential input states to obtain the averaged teleportation fidelity.
We thus arrive at the \emph{fidelity of transmission} \cite{horodecki_general_1999}
\begin{align} \label{eq:fid_of_transmission}
    f(\rho_{AB}) := \int \diff \phi\, F ( \sigma_\phi, |\phi\rangle ) \;,
\end{align}
where the integral is taken with respect to the Haar measure on the space of input states $|\phi\rangle$.
Since $F(\sigma_\phi , |\phi\rangle) \in [0,1]$, we also have $f(\rho_{AB})\in [0,1]$ for any $\rho_{AB}$ and any teleportation procedure.
The upper bound $f_{max} = 1$ is achieved by the original teleportation protocol of Ref.~\cite{bennett_teleporting_1993} when a maximally entangled pure resource state is used.
On the contrary, assuming a scenario with a separable resource state $\rho_{AB} = \rho_{A}\otimes\rho_{B}$, a nonzero bound of the teleportation fidelity can be obtained.
In this case, one can apply a measure-and-prepare protocol \cite{linden_bound_1999} and obtain $f_{SEP} = 2/(d+1)$.\footnote{This can be shown by direct calculation (see the supplementary information). Alternatively, it follows from Thm. \ref{thm:teleportation_and_singlet_fraction}. For this, note that there are separable states $\sigma$ with $\mathcal{F}_{max}=1/d$, e.g., any subgroup state \eqref{eq:subgroup_state} in the Magic Simplex. For such a state, \eqref{eq:teleportation_and_singlet_fraction} applies and yields the desired bound.}

For what follows we also need to define the \emph{singlet fraction} of a state $\rho_{AB}$ by
\begin{align} \label{eq:singlet_fraction}
    \mathcal{F}(\rho_{AB}) := F(\rho_{AB}, |\Omega_{0,0}\rangle)  = \myTr(P_{0,0}\, \rho_{AB}) \;,
\end{align}
where $|\Omega_{0,0}\rangle$ is a maximally entangled state\footnote{Inconsequentially, for dimension $d=2$, $|\Omega_{0,0}\rangle$ is not the singlet state but can be transformed into the true singlet by local operations.} and $P_{0,0}$ is its projector, defined in \eqref{eq:max_ent_state_omega_00} and \eqref{eq:bell_states}, respectively.
Note that distillation aims to maximize the singlet fraction of the initial state by possibly using up multiple additional copies of $\rho_{AB}$.
A state is distillable if $\mathcal{F} = 1$ is reached with nonzero probability after some number of LOCC rounds (cf. Sec.~\ref{sec:ent_dist}).

Ref.~\cite{horodecki_general_1999} establishes an important connection between the singlet fraction of a state and the fidelity of transmission:
\begin{mytheorem} \label{thm:teleportation_and_singlet_fraction}
    For a given resource state $\rho_{AB} \in \denop(\hs_A \otimes \hs_B)$ with $\dim(\hs_A) = \dim(\hs_B) = d$, the maximal fidelity of transmission $f_{max}$ achievable by any teleportation protocol is given by
    \begin{align} \label{eq:teleportation_and_singlet_fraction}
        f_{max} = \frac{d \,\mathcal{F}_{max}+1}{d+1} \;.
    \end{align}
    Here, $\mathcal{F}_{max}$ denotes the maximal singlet fraction attainable by LOCC from a single copy of $\rho_{AB}$.
\end{mytheorem}

Applied to bound entangled states, we get
\begin{mycor} \label{cor:BE_and_tele}
    Bound entangled resource states do not permit better-than-classical quantum teleportation since $\mathcal{F}_{max} \leq 1/d$.
\end{mycor}

Otherwise, they would be distillable by the protocol of Ref.~\cite{horodecki_reduction_1999}, contradicting the definition of bound entanglement.
As a consequence, the maximal fidelity of transmission for bound entangled states satisfies $f_{max}^{BE} \leq 2/(d+1)$.
This is exactly the bound of the measure-and-prepare protocol mentioned above.
Thus, bound entangled states do not allow for better transmission fidelity than any purely classical channel.
Hence, the best teleportation strategy for Alice and Bob sharing a bound entangled state is to discard it and perform classical teleportation.

Finally, Ref.~\cite{chitambar_duality_2024} recasts quantum teleportation into a state discrimination problem.
Doing so allows them to characterize the resource states that offer no advantage over separable resource states in one-way teleportation protocols.
Their novel approach allows for an alternative proof of Cor. \ref{cor:BE_and_tele}.

\subsubsection{Dense Coding}

The concept of \emph{dense coding} \cite{bennett_communication_1992} (also called super-dense coding) is intimately connected to teleportation.
Indeed, in a certain sense, the two problems are dual to one another \cite{chitambar_duality_2024, wilde_quantum_2017}.
Dense coding aims to transmit a certain number of bits (i.e., classical information) from Alice to Bob by exchanging fewer qudits.
We assume a noiseless quantum channel through which the qudits are sent, whereas the shared states are generally mixed and noisy.
The procedure starts with Alice and Bob sharing $n$ identical bipartite quantum states $\rho_{AB}^{\otimes n}$.
Alice can then perform any local quantum operation on her part of the shared states to encode some classical bits and then send her altered $n$ qudits to Bob.
Now Bob tries to retrieve $m>n$ classical bits of information encoded by Alice.

The problem as presented above can be restated in terms of the capacity of the (noisy) entanglement-assisted classical channel \cite{bennett_entanglement-assisted_1999, bennett_entanglement-assisted_2002, horodecki_classical_2001}.
A channel capacity greater than $1$ indicates that more than one bit of classical information is transmitted for every qudit sent through the noiseless quantum channel.
A quantum channel's maximum capacity is $2 \log_2(d)$, where $d$ is the dimension of the qudit (see Ref.~\cite{wilde_quantum_2017}).
One of the results of Ref.~\cite{horodecki_classical_2001} is
\begin{mytheorem}
    Bound entanglement-assisted classical channels do not admit a capacity greater than 1.
\end{mytheorem}
Consequently, they do not outperform a fully classical protocol, i.e., one where the bits are sent via a classical channel from Alice to Bob.

A different approach to dense coding with bound entangled states is presented in Ref.~\cite{chitambar_duality_2024}.
Therein, the authors reexamine the duality between dense coding and teleportation.
As a consequence of Cor. \ref{cor:BE_and_tele}, they likewise prove that bound entanglement is not useful in general dense coding schemes beyond the classical level.

\subsubsection{Activation of Bound Entanglement} \label{sec:Activation_of_BE}

In Ref.~\cite{linden_purifying_1998, kent_entangled_1998}, it is proven that for generic mixed states, it is impossible to distill a maximally entangled state with nonzero probability from only a single copy.
This means that no LOCC operations applied to a generic bipartite mixed state (of equal dimensions $d_A = d_B$) can bring the singlet fraction \eqref{eq:singlet_fraction} arbitrarily close to $1$ with non-vanishing probability.
Consequently, by Thm.~\ref{thm:teleportation_and_singlet_fraction}, it is generally impossible to teleport faithfully using only a single entangled mixed state, even if it is free entangled.

An example\footnote{This example and the specific activation protocol of Ref.~\cite{horodecki_bound_1999} mentioned below are reproduced in the supplementary information.} of such a free entangled qutrit-qutrit state $\rho^{\mathrm{FE}}$ is given in Ref.~\cite{horodecki_general_1999}.
Its initial singlet fraction is $\mathcal{F}_0$, and this cannot be increased beyond some $\mathcal{F}_{max}<1$ by LOCC and using only a single copy of the state.
They further show that this upper bound even holds for LOCC protocols that have a probability of success converging to zero.
According to Thm.~\ref{thm:teleportation_and_singlet_fraction} the fidelity of teleportation through $\rho^{\mathrm{FE}}$ as the sole resource state is thus upper bounded by $f_{max} = \frac{1}{4}(3 \,\mathcal{F}_{max}+1) < 1$.
Hence, faithful teleportation through $\rho^{\mathrm{FE}}$ is not guaranteed.
Surprisingly, an additional supply of bound entangled states can considerably improve teleportation fidelity.
This is unexpected because bound entangled states can neither be distilled nor allow faithful teleportation on their own (cf. Cor.~\ref{cor:BE_and_tele}).

A concrete LOCC protocol that achieves this is presented in Ref.~\cite{horodecki_bound_1999}.
It consumes a suitable bound entangled state $\rho^{\mathrm{BE}}$ and transforms $\rho^{\mathrm{FE}}$ to $\rho_1^{\mathrm{FE}}$ with singlet fraction $\mathcal{F}_1 > \mathcal{F}_0$.
Heuristically speaking, the entanglement from the bound entangled state is transferred to the free entangled pair.
As the states $\rho^{\mathrm{FE}}$ and $\rho_1^{\mathrm{FE}}$ are structurally identical, the protocol can be repeated.
After $n$ repetitions, the singlet fraction of $\rho_n^{\mathrm{FE}}$ is $\mathcal{F}_n$ which approaches 1 as $n$ goes to infinity.
Importantly, for some $n=n_0$, one eventually obtains $\mathcal{F}_{n_0} > \mathcal{F}_{max}$.
Therefore, the teleportation fidelity of $\rho_{n_0}^{\mathrm{FE}}$ exceeds $f_{max}$ and ultimately approaches $1$ for large $n_0$, thus making faithful teleportation possible.

The recent work in Ref.~\cite{ozaydin_nonlocal_2022} investigates a different protocol of bound entanglement activation compared to the one in Ref.~\cite{horodecki_bound_1999}.
The authors utilize quantum Zeno dynamics to similarly increase the singlet fraction \eqref{eq:singlet_fraction} of a free entangled state by consuming one or more bound entangled states.

In scenarios like this, where bound entangled states improve the performance of some quantum information theoretic task that requires free entanglement, we say that the bound entanglement has been \textit{activated}.
On the other hand, some authors interpret this the other way around:
In the above example, they would say that the teleportation capability of the state $\rho^{\mathrm{FE}}$ has been activated by the bound entanglement in $\rho^{\mathrm{BE}}$.

A different kind of activation of bound entanglement is presented in Ref.~\cite{vollbrecht_activating_2002}.
The authors show that bound entanglement can more generally activate a quantum state's distillability.
In particular, bound entangled states can make any NPT state $1$-distillable (cf. Sec.~\ref{sec:ent_dist}).
To show this, they use bound entangled states constructed from two highly symmetric families of states, Werner and isotropic states.
It is worth noting that the resulting bound entangled activator states have full rank.
Moreover, some of them are universal activators, i.e., they can activate the distillability of any given NPT state.
Some bound entangled universal activator states lie arbitrarily close to the set of separable states.

\subsubsection{Pure State Conversion assisted by Bound Entanglement}\label{sec:pureconversionBE}

A similar effect to the activation of bound entanglement can be achieved for pure state conversion.
In particular, it is known that if only one copy of a pure bipartite quantum state is available, its Schmidt rank cannot be increased by any (probabilistic) LOCC scheme \cite{lo_concentrating_2001, vidal_entanglement_1999}.
However, when additional specific bound entangled states are available as a resource, this restriction can be circumvented \cite{ishizaka_bound_2004}:
\begin{mytheorem}
    A single copy of any bipartite entangled pure state $|\psi\rangle \in \hs_A \otimes \hs_B$ can be converted with nonzero success probability to any other pure bipartite state $|\phi\rangle \in \hs_A \otimes \hs_B$ by LOCC with the help of suitable bound entangled states.
\end{mytheorem}
The crucial part of the conversion protocol involves the state
\begin{align}
    \label{eq:higher_dim_bell_state}
    |\Omega_{0,0}^{m} \rangle = \frac{1}{\sqrt{m}} \sum_{i = 1}^m \, |i\rangle \otimes  |i\rangle \in \mathbb{C}^d \otimes \mathbb{C}^d \;,
\end{align}
where $2\leq m \leq d$.
It is shown that this state can be probabilistically converted to the $d$-dimensional maximally entangled state $|\Omega_{0,0}^d\rangle$
by LOCC with the assistance of the (unnormalized) bound entangled state
\begin{align} \label{eq:pure_state_conversion_BE_state}
    \sigma_{m\rightarrow d} = \frac{1}{m^2(d-1)} \Big( (m-1) (P_{0,0}^d)_{A_1 B_1} \otimes (P_{0,0}^m)_{A_2 B_2} + \frac{1}{d+1} (\id_{d^2} - P_{0,0}^d)_{A_1 B_1} \otimes (\id_{m^2} - P_{0,0}^m)_{A_2 B_2} \Big) \;.
\end{align}
Here, $P_{0,0}^d$ and $P_{0,0}^m$ are the projectors onto $|\Omega_{0,0}^d \rangle$ and $|\Omega_{0,0}^m \rangle$, respectively.
The subscripts $A_1$, $A_2$, $B_1$, and $B_2$ indicate the corresponding Hilbert spaces on which the state acts.
This state is indeed PPT in the cut $A_1 A_2 | B_1 B_2$.
Furthermore, as the conversion of $|\Omega_{0,0}^m\rangle$ to $|\Omega_{0,0}^d\rangle$ is not possible for $d>m$ by LOCC alone, $\sigma_{m\rightarrow d}$ must be entangled.
Hence, it is bound entangled.

The explicit protocol of probabilistically converting any entangled state $|\psi\rangle$ to any other state $|\phi\rangle$ by LOCC (and potentially using bound entangled states) is as follows \cite{ishizaka_bound_2004}.
Suppose the initial state $|\psi\rangle$ is entangled with Schmidt rank $r>1$.
If the target state $|\phi\rangle$ has Schmidt rank $r'\leq r$, then the usual conversion techniques \cite{lo_concentrating_2001, vidal_entanglement_1999} are applicable.
However, if $r'>r$, the initial state $|\psi\rangle$ is first converted to $|\Omega_{0,0}^r \rangle$ by LOCC \cite{vidal_entanglement_1999}.
Next, $|\Omega_{0,0}^r \rangle$ is converted into $|\Omega_{0,0}^{r'} \rangle$ by LOCC and the use of $\sigma_{r\rightarrow r'}$ as in \eqref{eq:pure_state_conversion_BE_state}.
Finally, $|\Omega_{0,0}^{r'} \rangle$ can be converted to any desired state $|\phi\rangle$ with Schmidt rank $r'$.
Importantly, the probability of success is nonzero in either case.
Consequently, bound entangled states can be useful in pure state conversion tasks.

\subsubsection{Bound Entangled States and Quantum Channels}

There are various ways of identifying quantum states with quantum channels.
Two natural ones have been investigated for bipartite bound entangled states \cite{horodecki_binding_2000, divincenzo_unextendible_2003}, and both relate to the concept of binding entanglement channels:
\begin{mydef}[Binding Entanglement Channel] \label{def:binding_ent_channel}
    A quantum channel $\mathcal{E}$ is called a binding entanglement channel if it can be used to create bipartite bound entangled states between two parties, but it cannot be used to share free entangled states.
\end{mydef}
By definition, binding entanglement channels are never entanglement breaking \cite{horodecki_entanglement_2003}.
The authors of Ref.~\cite{horodecki_binding_2000} characterized binding entanglement channels by proving
\begin{mytheorem}
    A quantum channel $\mathcal{E}$ is a binding entanglement channel if and only if its Choi state \eqref{eq:choi_operator} is a bound entangled state.
\end{mytheorem}
Thus, utilizing the Choi-Jamiołkowski isomorphism~\cite{choi_completely_1975, jamiolkowski_linear_1972}, it is possible to unambiguously associate a binding entanglement channel $\mathcal{E}$ with a bound entangled state $\rho$ whose reduced state is maximally mixed, i.e., $\rho_A = \myTr_A(\rho) = (1/d) \id_d$.
In particular, the identification is given by
\begin{align} \label{eq:BE_channel_identification}
    \rho = (\mathcal{E}\otimes \id ) P_{0,0} \;,
\end{align}
where $P_{0,0}$ is the Bell projector defined in \eqref{eq:bell_states}.
The original bound entangled state is thus obtained by sending one half of a maximally entangled state $|\Omega_{0,0}\rangle$ as in \eqref{eq:max_ent_state_omega_00} through the corresponding channel.
As any Bell-diagonal state is locally maximally mixed (cf. Sec.~\ref{sec:bell_system}), Eq.~\eqref{eq:BE_channel_identification} fully characterizes any bound entangled state within this subset of states.
For bound entangled states with reduced states not maximally mixed (but with full rank), it is still possible to find such an identification.
However, an additional preliminary step of local filtering is necessary in this case~\cite{horodecki_binding_2000}.

An alternative route from bound entangled states to binding entanglement channels is presented in Ref.~\cite{divincenzo_unextendible_2003}.
There, the binding entanglement channel is given by the teleportation through the bound entangled state (cf. Sec.~\ref{sec:BE_teleportation}).
No free entanglement can be shared by utilizing this channel.
Otherwise, one could use it to construct a distillation scheme for the initial bound entangled state.
Furthermore, sending one half of the maximally entangled state $|\Omega_{0,0}\rangle$ through the channel, the resulting state is bound entangled.
Consequently, the channel satisfies Def.~\ref{def:binding_ent_channel} and is thus a binding entanglement channel.

\subsubsection{Quantum Key Distribution}\label{sec:Keydistribution}

The main objective of quantum key distribution (originally devised in Ref.~\cite{bennett_quantum_2014}) is to ensure secure classical communication between two parties, Alice and Bob.
This is achieved by using ``effectively distributed'' entanglement \cite{curty_entanglement_2004} to establish a private key, i.e., a classical bit string known exclusively to them.
In this creation process, any potential interference of a third party, Eve, is detectable by Alice and Bob with nonzero probability.
Alice and Bob can subsequently use the shared key to encode classical data before transmission and decode it upon receipt.

A common measure of a state's capability of allowing for a secure key generation protocol is given by \cite{horodecki_quantum_2007}:
\begin{mydef}[Private Key Distillation Rate]
        The private key distillation rate (or simply distillable key) $K_D$ is the maximally attainable asymptotic rate of converting $n$ copies of some initial quantum state $\rho_{AB} \in \denop(\hs_A\otimes\hs_B)$ via LOCC operations into a private qudit state of dimension $d_n$ of the form
    \begin{align} \label{eq:private_state}
        \gamma_{ABA'B'}^{(d_n)} = \frac{1}{d_n} \sum_{i,j=0}^{d_n-1} |ii\rangle\langle jj|_{AB} \otimes U_i \, \Tilde{\rho}_{A'B'} \, U_j^\dagger \;,
    \end{align}
    where $\Tilde{\rho}_{A'B'}\in\denop(\hs_{A'}\otimes\hs_{B'})$ and the $U_i$ are some unitaries acting on $\hs_{A'}\otimes\hs_{B'}$.
    In particular, the (device-dependent) distillable key is thus given by
    \begin{align}
        \label{eq:distillable_key}
        K_D = \sup_\mathcal{P} \; \limsup_{n,d_n\rightarrow\infty}\frac{\mathrm{log}_2 (d_n)}{n} \;,
    \end{align}
    where the supremum is taken over all possible LOCC protocols $\mathcal{P}$ converting $\rho_{AB}^{\otimes n}$ into $\gamma_{ABA'B'}^{(d_n)}$.
\end{mydef}
The name \textit{private state} for states of the form \eqref{eq:private_state} comes from Ref.~\cite{horodecki_secure_2005} where it is shown that the most general state from which a private key can be obtained by LOCC is of this form.
The private key distillation rate is very similar to the asymptotic entanglement distillation rate of Sec.~\ref{sec:ent_dist}, with the main difference being the different target state $\gamma_{ABA'B'}$ \cite{horodecki_general_2009, horodecki_secure_2005}.
In general, distilling a private state is an easier task than distilling a maximally entangled state \cite{horodecki_quantum_2009}.

For some quantum key distribution protocols, entangled states shared between the parties are a valuable resource.
It is shown in Ref.~\cite{horodecki_secure_2005} that an arbitrarily secure private key can also be obtained from bound entangled states.
To this end, a specific bound entangled state with $K_D>0$ is constructed.
The downside of this constructive proof is that it requires many resource states on a high-dimensional Hilbert space.
As demonstrated in Ref.~\cite{horodecki_low-dimensional_2008}, this high-dimensionality is not generally necessary.
Therein, bound entangled states with dimensions $d_A = d_B = 4$ are constructed from which a private key can be obtained via the one-way distillation protocol from Ref.~\cite{devetak_distillation_2005}.
This example is generalized in Ref.~\cite{pankowski_low-dimensional_2011} to a wider class of states (requiring a two-way protocol for secret key distillation).
Ref.~\cite{horodecki_bounds_2015} shows that the possibility to distill a secure key from certain PPT-entangled states relates to how much they violate a Bell inequality (c.f. Sec.\ref{sec:BIandBE}.
A different construction of bound entangled states with a positive distillable key is presented in Ref.~\cite{ozols_bound_2014}. This establishes a link between bound entangled states and classical probability theory.
An important consequence of these examples is that a zero entanglement distillation rate, a feature of all bound entangled states, does not imply a zero private key distillation rate.

The authors of Ref.~\cite{horodecki_low-dimensional_2008} argue that the volume of bound entangled states $\rho_{AB}\in\denop(\hs_A\otimes\hs_B)$ with positive distillable key, $K_D>0$, is nonzero in the case of finite Hilbert space dimensions $d_A, d_B \geq 4$.
This makes the experimental implementation of their theoretical considerations more feasible because a nonzero volume implies, in some sense, robustness against small perturbations (cf. Sec.~\ref{BEexperiment}).
Furthermore, in Ref.~\cite{banaszek_quantum_2012}, so-called ``privacy witnesses'' are introduced.
They are an analog of entanglement witnesses (cf. Sec.~\ref{sec:EWs}) but for detecting a positive distillable key of quantum states.
The advantage of these privacy witnesses is that they can facilitate the experimental verification of (bound entangled) private states compared to other methods.

The nonzero distillable key of bound entangled states might also be relevant for quantum repeaters \cite{briegel_quantum_1998}.
Commonly, such repeaters utilize entanglement distillation and teleportation, making them inapplicable in scenarios with bound entanglement.
An alternative realization is proposed in Ref.~\cite{bauml_limitations_2015}.
This does not require distillability but only a nonzero distillable key and can, therefore, be used with bound entangled states.

Interestingly, the impossibility of using bound entangled states for distillation-based quantum repeaters can be turned into an advantage. Ref.~\cite{sakarya_hybrid_2020} proposes a hybrid quantum network architecture based on bound entangled states. While allowing the encryption of classical communication between two authorized parties, the impossibility of distilling such states protects against unauthorized secret-key generation via entanglement swapping. As distillable states used in such a network do not provide this security, bound entanglement can represent a unique resource for information processing tasks.

A further focus is on device-independent quantum key generation \cite{zhang_device-independent_2022}.
In this context, the authors of Ref.~\cite{christandl_upper_2021} found PPT-entangled states such that the device-independent distillable key $K_{DI}$ is strictly less than the device-dependent one in \eqref{eq:distillable_key}, i.e., $K_{DI} < K_D$ for certain shared states.
Regarding bound entanglement in device-independent protocols, a revised Peres conjecture\footnote{For the original Peres conjecture, see the discussion in Sec.~\ref{sec:BIandBE}.} is brought forward in Ref.~\cite{arnon-friedman_upper_2021}:
\begin{myconj}
    Bound entangled states are not useful for device-independent quantum key distribution.
\end{myconj}

\subsection{Bound Entanglement and Entanglement Measures} \label{sec:BE_and_ent_measures}
A natural extension to whether a quantum state is entangled or not is to ask how much the state is entangled. \emph{Entanglement measures} (cf. Ref.~\cite{plenio_introduction_2007}) aim to quantify the amount of entanglement. However, up to this date, no universal measure has been found. Instead, several entanglement measures have been proposed with different operational meanings. For mixed states, these measures are generally hard to compute.
However, for pure bipartite states  the \emph{entropy of entanglement} (cf. Ref.~\cite{bennett_concentrating_1996} and defined below), based on the \emph{von Neumann entropy} provides a unique measure of entanglement (cf. Ref.~\cite{donald_uniqueness_2002}). All proposed bipartite entanglement measures should reduce to it for pure states. Depending on the context, other properties like non-increasing on average under LOCC, (sub-- or super--) additivity or continuity can be required or desired. 

In this section, we shortly introduce the most relevant entanglement measures for bipartite mixed states and elaborate on their relation to bound entanglement.
Of special interest is the appearance of irreversibility of certain quantum information processing tasks. The difference between two important proposed entanglement measures, namely the so-called \emph{entanglement cost} and the \emph{distillable entanglement}, implies an irreversible loss of entanglement as a resource for some protocols. This phenomenon is closely related to bound entanglement. Another relation to bound entanglement is reflected in the additivity of the distillable entanglement and the existence of NPT bound entangled states.

The bipartite \emph{entropy of entanglement} $E$ of a pure state is defined as the \emph{von Neumann entropy} $S(\rho) := -\myTr(\rho \log_2(\rho))$ of the reduced state:
\begin{mydef}[Entropy of Entanglement] Let $|\psi \rangle\langle \psi| \in \denop(\hs_A\otimes\hs_B)$. The (von Neumann) entropy of entanglement of $|\psi \rangle $ is defined as
\begin{align}
    E(|\psi \rangle\langle \psi|) := S(\pTrA(|\psi \rangle\langle \psi|)) =  S(\pTrB(|\psi \rangle\langle \psi|)).
\end{align}
\end{mydef}
A unique or ``canonic'' extension to mixed states is not known. One way to generalize it to mixed states is by so-called ``convex hull constructions'' based on all possible decompositions of $\rho$ into pure states $|\psi_i\rangle$ with probabilities $p_i$.  The \emph{entanglement of formation} $E_F$ is such an extension.
\begin{mydef}[Entanglement of Formation]
\label{def:entanglement_of_formation}
Let $\rho \in \denop(\hs_A\otimes\hs_B)$. Given all decompositions of $\rho$ into pure states $\rho = \sum_i p_i\; |\psi_i\rangle\langle\psi_i|$, the entanglement of formation is defined as
\begin{align}
    E_F(\rho) := \inf_{\{p_i, \psi_i\} } \left\{ \sum_i p_i\; E(|\psi_i\rangle\langle\psi_i|)\right\},
\end{align}
where the infimum is taken over all possible decompositions.
\end{mydef}
The entanglement of formation measures the minimum number of maximally entangled qubits required on average to produce the given state via LOCC \cite{bennett_concentrating_1996}. The authors of Ref.~\cite{horodecki_quantum_2009} show that this measure is non-additive, i.e., in general, $E_F(\rho_1\otimes\rho_2) \neq E_F(\rho_1) + E_F(\rho_2)$.\\
A closely related measure is the \emph{entanglement cost} $E_C$. It is defined in the asymptotic limit of infinitely many copies and quantifies the maximal rate $r$ at which maximally entangled states can be converted to copies of $\rho$ using LOCC. The conversion is allowed to be approximate with vanishing error in the limit of a large number of copies. Let $P_{0,0}^{\otimes m}$ be $m$-copies of the bipartite maximally entangled state of dimension $d^2$ \eqref{eq:bell_states}. Denoting general LOCC operations by $\mathcal{E}_{LOCC}$ and the trace norm as $\myTr|\rho| := \sqrt{\rho^\dagger\rho}$, we can define the entanglement cost as follows:
\begin{mydef}[Entanglement Cost] \label{def:entanglement_cost }Let $\rho \in \denop(\hs_A\otimes\hs_B)$. The entanglement cost is defined as
\begin{align}
    E_C(\rho) := \inf
        \left \{ 
        r: \lim_{{n \to \infty}}  
        \left [ \inf_{\mathcal{E}_{LOCC}} \myTr|\rho^{\otimes n} - \mathcal{E}_{LOCC}(P_{0,0}^{\otimes rn})| \right ]
        = 0 
        \right \}.
\end{align}
\end{mydef}
The close relation of the entanglement cost and the entanglement of formation is reflected in the equality of the \emph{regularized} entanglement of formation and the entanglement cost (cf. Ref.~\cite{hayden_asymptotic_2001}),
\begin{align}
    E_C(\rho) = \lim_{{n \to \infty}} \frac{E_F(\rho^{\otimes n})}{n}.
\end{align}
The \emph{distillable entanglement} $E_D$ quantifies entanglement from the opposite viewpoint compared to the entanglement cost. While $E_C$ represents the rate for transforming copies of a maximally entangled state to copies of the state $\rho$, $E_D$ represents the rate for the inverse process of transforming multiple copies of $\rho$ to maximally entangled states, i.e., entanglement distillation (cf. Sec.~\ref{sec:ent_dist}):
\begin{mydef}[Distillable Entanglement] \label{def:distillable_entanglement}Let $\rho  \in \denop(\hs_A\otimes\hs_B)$. The distillable entanglement is defined as
\begin{align}
    E_D(\rho) := \sup
        \left \{ 
        r: \lim_{{n \to \infty}}  
        \left [ \inf_{\mathcal{E}_{LOCC}} \myTr|\mathcal{E}_{LOCC}(\rho^{\otimes n}) - P_{0,0}^{\otimes rn}| \right ]
        = 0 
        \right \}.
\end{align}
\end{mydef}
Note that in the definitions of $E_C$ and $E_D$, often $d=2$ is assumed for the Bell state $P_{0,0}$. As discussed in Sec.~\ref{sec:ent_dist}, for $d>2$ the asymptotic rates change according to $r \rightarrow r / [log_2(d)]$. 

Another important measure is the \emph{relative entropy of entanglement} $E_R$. Using the \emph{quantum relative entropy} $S(\rho || \sigma):= \myTr(\rho \log_2(\rho) - \rho\log_2(\sigma))$, the relative entropy of entanglement can be defined as follows:
\begin{mydef}[Relative Entropy of Entanglement]
\label{def:relative_entropy_of_entanglement}
Let $X \subset \denop(\hs_A\otimes \hs_B)$ be a set that is mapped onto itself under LOCC. The relative entropy of entanglement of $\rho$ with respect to $X$ is
\begin{align}
    E_{R,X}(\rho) &:= \inf_{\sigma \in X} S(\rho||\sigma).
\end{align}
The regularized relative entropy of entanglement is
\begin{align}
    E^{\infty}_{R,X}(\rho) &:= \lim_{{n \to \infty}} \frac{E_{R,X}(\rho^{\otimes n})}{n}.
\end{align}    
\end{mydef}
If $X = \SEP$, we simply write $E_R$. However, depending on which states are considered as ``resource-neutral'', i.e., states for which zero measure is specified, also the sets $\PPT$ or the set of undistillable states can be considered as the set $X$. \\
Many other proposed entanglement measures exist (see Ref.~\cite{plenio_introduction_2007} for an introduction to the topic), including the already introduced \emph{distillable key} $K_D$ (cf. Eq.~\eqref{eq:distillable_key}) or a measure related to the concurrence (cf. Ref.~\cite{hiesmayr_two_2009}) and its quasi-pure approximation of Sec.~\ref{quasipure}.

\subsubsection{Irreversibility and PPT Bound Entanglement} 
From the definitions of the entanglement cost and the distillable entanglement, it becomes clear that if they are equal, the processes of entanglement distillation using multiple copies of a state $\rho$ is asymptotically reversible because $\rho$ can be created from distilled maximally entangled states via LOCC at the same rate. However, in Refs.~\cite{divincenzo_unextendible_2003,vidal_irreversibility_2001}, it is shown that this does not generally hold and that the process of entanglement distillation is irreversible. \\
In particular, it is shown for a specific PPT bound entangled state $\rho_b$ with $E_D(\rho_b) = 0$, that $E_C(\rho_b) > 0$.
This is generalized in Ref.~\cite{yang_irreversibility_2005} to all bound entangled states.
Consequently, interconverting the mixed-state entanglement of $\rho$ and the entanglement of the pure Bell states via asymptotic LOCC operations is not reversible. 

This also holds for catalytic LOCC operations of the form $\rho_b^{\otimes n} \otimes P_{0,0}^{\otimes m} \underset{LOCC}{\longrightarrow} P_{0,0}^{\otimes (m+l) }$ (cf. Ref.~\cite{vidal_irreversibility_2001}), even if correlations of the catalyst with the systems are allowed at the output \cite{lami_no-go_2024}. This means that, even in the asymptotic limit, the entanglement needed to locally create a PPT bound entangled state $\rho_b$ cannot be distilled using (catalytic) LOCC. Interpreting entanglement as a resource for information processing tasks, bound entangled states bind resources that cannot be accessed by distillation and thus motivate their name.

\subsubsection{Additivity of Distillable Entanglement and NPT Bound Entanglement}
The potential existence of NPT bound entanglement is linked to the additivity of the distillable entanglement. As mentioned in Sec.~\ref{sec:Activation_of_BE}, with the help of certain bound entangled states, any NPT entangled state, including a potential NPT bound entangled states $\sigma_b$, can be activated and made distillable by local operations (cf. Ref.~\cite{vollbrecht_activating_2002, horodecki_bound_1999}). Since any universal activator bound entangled state $\rho_b$ and any NPT bound entangled state $\sigma_b$ have by definition zero distillable entanglement, the distillability of the activated state would imply $E_D(\rho_b \otimes \sigma_b) >  E_D(\rho_b) + E_D(\sigma_b) = 0$ and consequently the non-additivity of the distillable entanglement.

\subsubsection{Inequalities for Entanglement Measures}
In Sec.~\ref{sec:Keydistribution}, it is discussed that there exist bound entangled states $\rho_b$ for which the distillable key \eqref{eq:distillable_key} is nonzero, i.e., $K_D(\rho_b) > 0 = E_D(\rho_b)$. In Ref.~\cite{horodecki_secure_2005}, it is shown that the following hierarchy between the entanglement measures exists:
\begin{align}
    \label{eq:inequality_entanglement_measurers}
    E_D \leq K_D \leq E_R^{\infty} \leq E_C.
\end{align}
While this inequality holds for all states, also the strict inequalities $E_D < K_D < E_C$ and $E_D < E_R^{\infty}$ are possible for certain states. Bound entangled states (with $E_D=0$) can thus exhibit nonzero  $K_D$, $ E_R^{\infty}$  and $E_C$.

\subsection{Bound Entanglement and Quantum Metrology}\label{sec:metrology}

The field of quantum metrology studies precise measurement techniques with the aim of extracting physical parameters from quantum systems.
One of the tasks is to estimate the dynamical parameter $\theta$ of a unitary evolution $U(\theta) = e^{-i \theta H}$ acting on a given quantum state $\rho$ (see, e.g., Ref.~\cite{toth_quantum_2014, liu_quantum_2019}).
This evolution is induced by a Hamiltonian $H$.
For any measurement, the estimation sensitivity is bounded by the \emph{quantum Cramér-Rao bound}, 
\begin{align}
    (\Delta \theta)^2 \geq \frac{1}{F_Q(\rho, H)} \;,
\end{align}
which involves the inverse of the \emph{quantum Fisher information} (QFI).
The latter is thus often the central object of study, and for the case under consideration, it is given by
\begin{align} \label{eq:QFI_def}
    F_Q(\rho, H) = 2 \sum_{i,j = 0}^{d-1} \frac{(\lambda_i - \lambda_j)^2}{\lambda_i + \lambda_j} \; \abs{\langle \psi_i | H | \psi_j \rangle}^2 \;.
\end{align}
Here, the $|\psi_i\rangle$ are eigenvectors of $\rho \in \denop(\hs)$, $\dim(\hs) = d$, with corresponding eigenvalue $\lambda_i$, and the sum is only over terms with $\lambda_i + \lambda_j \neq 0$.
One of the results in Ref.~\cite{braunstein_statistical_1994} states that the QFI is upper bounded for any state $\rho$ by
\begin{align} \label{eq:QFI_upper_bound}
    F_Q(\rho, H) \, \leq \, F_{Q, \max}(\rho, H) = 4  (\Delta_\rho H )^2  \;,
\end{align}
where $(\Delta_\rho H)^2 = \langle H^2 \rangle_\rho - \langle H \rangle_\rho^2$ is the variance of $H$.
Furthermore, there exists an upper bound of the QFI for separable states, obtained by maximizing \eqref{eq:QFI_def} over all states $\rho\in \mathrm{SEP}$ \cite{pezze_entanglement_2009}.
This is generally lower than the maximal upper bound given by \eqref{eq:QFI_upper_bound}.
A quantum state $\rho$ is said to be \emph{useful for quantum metrology} (for a given Hamiltonian $H$) if the QFI $F_Q(\rho, H)$ exceeds the separable bound.

Some investigations were conducted to find an answer to the question of whether bipartite bound entangled states can be useful for quantum metrology.
The first examples are presented in Ref.~\cite{toth_quantum_2018}.
Semidefinite programming is used to find numerical instances of such states, including ones in Hilbert space dimensions $2\times4$ and $d\times d$ for $3\leq d\leq 12$.
In addition to the numerical findings, the authors also give the analytical form of an example in dimension $4\times 4$.

Subsequently, starting from these states, as well as from ones in Ref.~\cite{badziag_bound_2014}, the authors of Ref.~\cite{pal_bound_2021} constructed two distinct families of ($2d\times 2d$)-dimensional bound entangled states for any $d\geq 2$ which are metrologically useful.
On top of that, they even show that for their specific Hamiltonian, the general upper bound \eqref{eq:QFI_upper_bound} is satisfied.
Hence, these states are maximally useful for parameter estimation with a specific Hamiltonian.

In conclusion, even though bound entangled states can sometimes be useful for quantum metrology, generally, maximally entangled pure states offer a greater metrological advantage.

\subsection{Bound Entanglement and Bell Inequalities}\label{sec:BIandBE}

The origin of today's quantum information theory can be dated back to the publication of John Stewart Bell in 1964~\cite{bell_einstein_1964}, where he showed that the predictions of quantum theory of two joint quantum systems are in general incompatible with a local realistic hidden variable theory. Or in the words of Bell~\cite{bell_speakable_1987}:\\ ``\textit{If [a hidden-variable theory] is local, it will not agree with quantum mechanics, and if it agrees with quantum mechanics, it will not be local.}''

In mathematical terms, any local realistic model assumes that the measurement statistics of a joint system are reproduced by
\begin{equation} \label{eq:prob_bell_ineq}
P(A,a; B, b)=\int \diff \lambda \; \rho(\lambda)\; p(A, a,\lambda)\cdot p(B, b,\lambda)\;, 
\end{equation}
where $A,B$ are observables chosen independently by Alice and Bob and $a,b$ are their measurement outcomes. Here, $\rho(\lambda)$ is any normalized distribution of hidden parameters $\lambda$. The factorization of the joint probability distribution of $A$ and $B$ into the individual probabilities $p$ is the assumption denoted as \textit{local} by Bell. 

Combining different probabilities of the form \eqref{eq:prob_bell_ineq} allows us to find upper and lower bounds on those combinations, called Bell inequalities. Those discriminate between local realistic models and other models, particularly quantum theory. A plethora of experiments in many distinct quantum systems show a violation of such Bell inequalities and, herewith, a contradiction to the assumptions behind Bell inequalities. A violation of a Bell inequality is only found in the case the joint system is entangled, but not all entangled systems violate a particular Bell inequality.

Here we present the result of Ref.~\cite{vertesi_disproving_2014}, where the authors found a particular bound entangled state that does violate a particular Bell inequality. Consequently, the Peres conjecture~\cite{peres_all_1999}  -- that bound entangled states are equal to states that exhibit a local realistic explanation -- is falsified. An alternative formulation of the conjecture is that any entangled state that leads to Bell inequality violation must be free entangled.

Let us consider the qutrit-qutrit state
\begin{equation}
   \rho_{\textrm{bound}}= \left(
\begin{array}{ccccccccc}
 \frac{9789}{27536} & 0 & 0 & 0 & \frac{3239}{27536} & \frac{5 \sqrt{\frac{131}{2}}}{13768} & \frac{45 \sqrt{\frac{131}{2}}}{6884} & 0 & 0 \\
 0 & \frac{3311}{27536} & \frac{5 \sqrt{\frac{131}{2}}}{13768} & \frac{3239}{27536} & 0 & 0 & 0 & -\frac{45 \sqrt{\frac{131}{2}}}{6884} & -\frac{9}{6884} \\
 0 & \frac{5 \sqrt{\frac{131}{2}}}{13768} & \frac{1}{13768} & \frac{5 \sqrt{\frac{131}{2}}}{13768} & 0 & 0 & 0 & -\frac{9}{6884} & 0 \\
 0 & \frac{3239}{27536} & \frac{5 \sqrt{\frac{131}{2}}}{13768} & \frac{3311}{27536} & 0 & 0 & 0 & -\frac{45 \sqrt{\frac{131}{2}}}{6884} & \frac{9}{6884} \\
 \frac{3239}{27536} & 0 & 0 & 0 & \frac{9789}{27536} & -\frac{5 \sqrt{\frac{131}{2}}}{13768} & -\frac{45 \sqrt{\frac{131}{2}}}{6884} & 0 & 0 \\
 \frac{5 \sqrt{\frac{131}{2}}}{13768} & 0 & 0 & 0 & -\frac{5 \sqrt{\frac{131}{2}}}{13768} & \frac{1}{13768} & \frac{9}{6884} & 0 & 0 \\
 \frac{45 \sqrt{\frac{131}{2}}}{6884} & 0 & 0 & 0 & -\frac{45 \sqrt{\frac{131}{2}}}{6884} & \frac{9}{6884} & \frac{81}{3442} & 0 & 0 \\
 0 & -\frac{45 \sqrt{\frac{131}{2}}}{6884} & -\frac{9}{6884} & -\frac{45 \sqrt{\frac{131}{2}}}{6884} & 0 & 0 & 0 & \frac{81}{3442} & 0 \\
 0 & -\frac{9}{6884} & 0 & \frac{9}{6884} & 0 & 0 & 0 & 0 & \frac{9}{6884} \\
\end{array}
\right) \;.
\end{equation}
The state is bound entangled, which is detected by the realignment criterion (see Sec.~\ref{sec:realignment}).
The following observable, in the form of an entanglement witness (see Sec.~\ref{sec:EWs}), can be defined:
\begin{equation}
W_{Bell}=
\left(
\begin{array}{ccccccccc}
 \frac{1}{25} & 0 & 0 & 0 & -\frac{1}{25} & \frac{\sqrt{2}}{25} & -\frac{\sqrt{21}}{25} & 0 & 0 \\
 0 & \frac{3}{25} & \frac{\sqrt{2}}{25} & -\frac{1}{25} & 0 & 0 & 0 & \frac{\sqrt{21}}{25} & \frac{\sqrt{42}}{25} \\
 0 & \frac{\sqrt{2}}{25} & \frac{46}{25} & \frac{\sqrt{2}}{25} & 0 & 0 & 0 & \frac{\sqrt{42}}{25} & 0 \\
 0 & -\frac{1}{25} & \frac{\sqrt{2}}{25} & \frac{3}{25} & 0 & 0 & 0 & \frac{\sqrt{21}}{25} & -\frac{\sqrt{42}}{25} \\
 -\frac{1}{25} & 0 & 0 & 0 & \frac{1}{25} & -\frac{\sqrt{2}}{25} & \frac{\sqrt{21}}{25} & 0 & 0 \\
 \frac{\sqrt{2}}{25} & 0 & 0 & 0 & -\frac{\sqrt{2}}{25} & \frac{46}{25} & -\frac{\sqrt{42}}{25} & 0 & 0 \\
- \frac{\sqrt{21}}{25} & 0 & 0 & 0 & \frac{\sqrt{21}}{25} & -\frac{\sqrt{42}}{25} & \frac{21}{25} & 0 & 0 \\
 0 & \frac{\sqrt{21}}{25} & \frac{\sqrt{42}}{25} & \frac{\sqrt{21}}{25} & 0 & 0 & 0 & \frac{21}{25} & 0 \\
 0 & \frac{\sqrt{42}}{25} & 0 & -\frac{\sqrt{42}}{25} & 0 & 0 & 0 & 0 & \frac{8}{25} \\
\end{array}
\right) \;.
\end{equation}
This satisfies for all local realistic hidden-variable theories~\cite{vertesi_disproving_2014}
\begin{equation}
0\leq \myTr({W_{Bell}\; \rho})\;.
\end{equation}
Computing $\myTr(W_{Bell}\;\rho_{\textrm{bound}})=-0.000263144$ shows a violation and thus the Peres conjecture is falsified. Undistillable states contradict the assumptions behind the derivation of Bell inequalities (the extension to the dimension $d=4,\dots ,8$ is presented in Ref.~\cite{pal_family_2017}). In Ref.~\cite{horodecki_bounds_2015} the question is asked how much one can violate a Bell inequality with PPT-entangled states and it is concluded that a violation cannot be large with respect to different measures, which affects quantum key generation (see Sec.~\ref{sec:Keydistribution}).

Researchers have tried to find different families of bound entangled states that violate other Bell inequalities. For example, the systematic set of Bell inequalities for qudit systems~\cite{collins_bell_2002}, so-called CGLMP-Bell inequalities, have been discussed for a subset of families within the Magic Simplex~\cite{spengler_geometric_2011} and found a negative answer, i.e., no bound entangled state violated these Bell inequalities. It remains open whether a general set of Bell inequalities and bound entangled states can be constructed in a systematic way. With that, the relation of undistillability and assumptions behind Bell inequalities, such as non-locality, can be further enlightened.

\subsection{Bound Entanglement and Lorentz Invariance} \label{sec:lorentz_invariance_BE}

One of the biggest open questions in physics is whether the two great theories, relativity and quantum theory, can be unified. Entanglement and quantum information theoretic tasks have therefore been intensively discussed in a relativistic setting (see, e.g., Ref.~\cite{friis_relativistic_2010} and references therein). In Ref.~\cite{caban_bound_2023}, the authors put forward the question of whether bound entangled states differ from free entangled states by being invariant under a Lorentz boost. They found a negative answer, i.e., boosted bound entangled states can become separable, bound, or free entangled states depending on the applied relativistic boost.
We sketch their derivation of this fact in the remainder of this section.

A strong guiding principle is that physics should not change for different observers. Let us choose as an observable an entanglement witness $W$ (see Sec.~\ref{sec:EWs}), then the following relation has to hold
\begin{equation}\label{relativisticwitness}
    0\leq \tr(W \rho)=
 \tr(W^{\mathsf{boosted}}\rho^{\mathsf{boosted}})\;,
\end{equation}
where we denoted by ``boosted'' the observable that a Lorentz-boosted observer is assigning to the experimental setting.

Let us choose as a particular entanglement witness, the MUB-witness capable of detecting bound entangled states, defined in Eq.~(\ref{MUBwitness}), $W=W_{\textrm{spin}}(\mathcal{M}_{d+1},s\not=\frac{s}{2})\otimes \mathbb{1}_{\textrm{momentum}}$ (see Sec.~\ref{sec:MUBwitness}). Furthermore, let us choose $d=3$ and $s=1$ and a boost along the $z$-axis with a particular choice of energy (details can be found in Ref.~\cite{caban_bound_2023,caban_lorentz-covariant_2005}). For any relativistic setting, one has to consider the spin and the momentum degrees of freedom of a quantum particle.
Let us choose the Magic Simplex state $\rho_{\textrm{spin}}(x)=2 x\; P_{0,0}+x\; P_{1,1}+(\frac{1}{3}-x)\;\sum_{i=0}^{2} P_{2,i}$ for two spin-$1$ particles with $x\in[0,\frac{1}{3}]$ for the spin part and a momentum state, e.g., $|\psi_{\textrm{momentum}}\rangle=\frac{1}{\sqrt{2}}\lbrace |\vec{p}_+\vec{p}_-\rangle+|\vec{p}_-\vec{p}_+\rangle\rbrace$ with $\vec{p}_\pm$ denoting opposite momenta. Then for $\rho=\rho_{\mathsf{spin}}(x)\otimes \rho_{\mathsf{momentum}}$ the inequality \eqref{relativisticwitness} turns to
\begin{equation}
    0\leq 1-x\;,
\end{equation}
which is not violated for any $x$. This is as expected since we have chosen the spin and the momentum part of our total state $\rho$ to be separable.

The situation changes when we consider the spin or momentum subspaces, which we achieve by partial tracing over the respective subsystems. For the spin part, we find for the unboosted case
\begin{equation}
    0\leq\myTr\left(\myTr_\mathsf{momentum} ( W )\myTr_\mathsf{momentum} (\rho)\right)=-x\;.
\end{equation}
Thus, a violation is found for any $x$ that is not equal to zero. Indeed, the state is separable for $x=0$, PPT entangled for $x\in\{0,\frac{2}{15}]$ and else free entangled. Here, the MUB witness detects the full range of entanglement. On the other hand, for a particular relativistic boost, one finds that for no choice of $x$, a violation can be found, i.e.,
\begin{equation}
    0\;\leq\;\myTr\left(\myTr_\mathsf{momentum} (W^{\mathsf{boosted}}) \myTr_\mathsf{momentum}(\rho^{\mathsf{boosted}})\right)=-x+0.45\;.
\end{equation}
Further detailed analyses show that the spin states $\myTr_\mathsf{momentum}(\rho^{\mathsf{boosted}})$ are in this case separable.

In summary, a boosted state can lose its property of being bound or free entangled, i.e., the Lorentz invariance does not hold for either bound or free entangled states; it is not connected to the PPT property.
From a practical point of view, different boosts allow to produce new sets of bound entangled states out of a particular bound entangled state.

\section{Experiments Verifying Bound Entanglement}\label{BEexperiment}

In recent years, impressive progress has been made in controlling the preparation, dynamics, and measurement of high-dimensional quantum systems \cite{wang_qudits_2020, chi_programmable_2022, ringbauer_universal_2022,carvacho_experimental_2017,hiesmayr_observation_2016}. Quantum computers, which work with qudits instead of qubits, are also no longer out of sight. With this progress, investigations into the preparation and observation of bound entangled states have also begun. In 2013, the first experiment based on MUB witnesses (cf. Sec.~\ref{sec:MUBwitness}) was performed for two photons.
We present the sketch of this experiment in Sec.~\ref{Firstexperiment}, which confirmed for the first time the existence of this peculiar type of entangled states.

In general, verifying bound entanglement in the laboratory is an intrinsically hard problem since a large coherent volume of bound entangled states in Hilbert space must be found to enable verification despite experimental errors. In addition, a method must be found to distinguish between bound and free entanglement. With the discovery of the one-parameter family of states by the Horodecki family in 1998 \cite{horodecki_mixed-state_1998}, it became clear that the search for experimental verification requires a larger family of bound entangled states to fulfill the experimental requirements. This led to investigating states within a subset of the entire Hilbert space, the Magic Simplex (see Sec.~\ref{sec:bell_system}). These analyses show that bound entangled states are not rare in this subset, e.g., for $d=3$ a relative volume of about $10\%$ is bound entangled (see Sec.~\ref{sec:frequency_bound}) and that bound entangled states can occupy large connected volumes in the Bell-diagonal state space that allow experimental investigations.

Once a family of states can be generated experimentally, the next task is to verify their entanglement. This can be done by any entanglement witness that can be realized on these physical systems but does not distinguish between bound and free entangled states. Alternatively, one can perform state tomography and then apply the tools presented in Chapters~\ref{sec:free_vs_bnd} and \ref{sec:detection of BE} to detect bound entanglement. This may not be the best strategy as it is time-consuming, often not realizable for experimental reasons, and usually requires additional assumptions. Another strategy is proposed in Ref.~\cite{sentis_bound_2018} via 
a hypothesis test tailored to the detection of bound entanglement and
with a measure of statistical significance, but not yet applied experimentally. Ideas for preparing bound entangled states and proving their bound nature are urgently needed.

\subsection{First Experiment Verifying the Existence of Bound Entangled States for Two Qutrits}\label{Firstexperiment}

In 2013 the authors of 
Ref.~\cite{hiesmayr_complementarity_2013} chose a particular state family of the Magic Simplex (cf. Def.~\ref{def:magic_simplex}) and applied a maximum complementarity protocol based on MUB observables to witness the entanglement of two photons entangled in their orbital angular momentum. Furthermore, they performed state tomography to verify that the generated states were PPT. Thus the existence of bound entanglement in Nature was first experimentally verified $15$ years after its discovery~\cite{horodecki_mixed-state_1998}.

Spatially entangled photon pairs can be generated by shining with a Gaussian laser beam ($\lambda= 413$~nm wavelength) onto a collinear Type-I spontaneous parametric downconversion (SPDC) crystal. The phase-matching in the crystal is tuned via
temperature control to produce a similar amount of downconverted photons in the $l=0$ and $l=\pm 1$ orbital angular momentum (OAM) mode. The conservation of orbital angular momentum forces opposite $l$, which allows only for the following state (which defines for this section the computational basis by setting $|l\rangle\mapsto |l+1\rangle$)
\begin{equation}
|\psi_{SPDC}\rangle=\frac{1}{\sqrt{3}}\big( |l=0\rangle\otimes|l=0\rangle+|l=+1\rangle\otimes|l=-1\rangle+|l=-1\rangle\otimes|l=+1\rangle\big)\;.
\end{equation}
This state is a Bell state that allows the construction of any other Bell state by applying a Weyl operator onto one of the subsystems (cf. Def.~\ref{def:bell_states}). The experiment achieves this via spatial light modulators (SLM), i.e., nonlinear materials. Projective measurements are done by appropriate mode transformation and subsequent imaging onto the core of a single mode fiber. The mixing was achieved by two different methods, one by considering only a subset of all possible measurements and the other by stochastic rotations applied to one of the photons. The photons are finally guided to an avalanche photodiode which functions as a single-photon counter.

In the experiment, four different correlation functions were measured, i.e., ($\mathcal{B}_0$ is the computational basis, and $|i_k\rangle$ are the basis states of the mutually unbiased bases $\mathcal{B}_k$; see Def.~\ref{def:unbiased}),
\begin{eqnarray}
C_{\mathcal{B}_0,\mathcal{B}_0}&=&\sum_{i=0}^{d-1}\;\myTr(|i\rangle\langle i|\otimes|i+2\rangle\langle i+2|\,\rho)\;,\\
C_{\mathcal{B}_k,\mathcal{B}_k}&=&\sum_{i=0}^{d-1}\;\myTr(|i_k\rangle\langle i_k|\otimes|i_k^*\rangle\langle i_k^*|\,\rho)\;,\quad\textrm{with}\quad 1\leq k\leq 3 \;.
\end{eqnarray}
The sum of these four correlation functions relates to the partial transpose of the MUB witness~\eqref{MUBwitness} for $m=4$ and $s=2$. If the sum is found to be greater than $2$, entanglement is detected (cf. Eq.~\eqref{MUBKriterion}). According to Thm.~\ref{mubforexp} this MUB witness forms a non-decomposable witness and can therefore detect bound entanglement. 
The experimental results were~\cite{hiesmayr_complementarity_2013}:
\begin{center}
  \begin{tabular}{c|c|c}
 Correlation function    & Theory  & Experiment\\
 \hline
  $C_{\mathcal{B}_0,\mathcal{B}_0}$   &$0.675$  &$0.667\pm 0.005$\\
  $C_{\mathcal{B}_1,\mathcal{B}_1}$ &$0.468$ &$0.463\pm 0.005$\\
  $C_{\mathcal{B}_2,\mathcal{B}_2}$ &$0.468$ &$0.463\pm 0.005$\\
  $C_{\mathcal{B}_3,\mathcal{B}_3}$ &$0.468$ &$0.463\pm 0.005$\\
  \hline
  $\sum_{i=0}^{3} C_{\mathcal{B}_i,\mathcal{B}_i}$& $2.079$& $2.066\pm 0.02$
\end{tabular}  
\end{center}
Here, the state was chosen from the Magic Simplex family $B_1$ with $\alpha =-0.07$, $\beta =-1.73$, defined and illustrated in Fig.~\ref{fig:notcompleteMUB}.

The entanglement is detected since a value greater than $2$ is only possible for entangled states. The experiment also showed excellent control since all entanglement class conserving symmetries (cf. Sec.~\ref{sec:bell_system}) have been used to implement equivalent states.

The PPT property was verified via state tomography for different runs of the experiment and different choices of bases~\cite{loffler_detection_2013}. State tomography can always be performed by measuring all possible correlation functions based on MUBs, i.e., $\tilde{C}_k=\sum_{i,j=0}^{d-1}\;\myTr(|i_k\rangle\langle i_k|\otimes|j_k\rangle\langle j_k|\,\rho)$.

\section{Construction of Bound Entangled States} \label{sec:construction}

Constructing bound entangled states is generally difficult as no efficient solution exists to discriminate the sets SEP, BE, and FE of general bipartite systems.
There are various ways of approaching this problem, e.g., depending on the purpose of the construction or whether some bound entangled states are initially known.

An intensively studied strategy to construct bound entangled states relies on unextendible product bases (UPBs).
We discuss this method in detail in Sec.~\ref{sec:UPB_construction}.
Alternatively, one can also analytically or numerically construct bound entangled states in different ways (e.g., Ref.~\cite{popp_almost_2022, li_sequentially_2021, sindici_simple_2018, moerland_bound_2024}).
This can be accomplished by enforcing that the states satisfy the PPT criterion (Thm.~\ref{thm:PPT-crit}) while violating at least one other entanglement criterion (cf. Sec.~\ref{sec:detection of BE}).
In addition, depending on the problem investigated, further useful restrictions on the bound entangled states may be implemented this way.

The downside of the numerical approaches is that they often do not offer a deep insight into the analytic structure of bound entangled states.
However, other properties, such as the frequency of bound entangled states, may become more feasible to investigate (cf. Sec.~\ref{sec:frequency_bound}).

The situation is different if some bound entangled states are known initially (e.g., as the result of an analytic or numeric construction described above).
Depending on the properties of these states, it is sometimes possible to construct further bound entangled states from them \cite{zhao_construction_2016, bandyopadhyay_non-full-rank_2005,bej_unextendible_2021, bandyopadhyay_non-full-rank_2005}.
Let us mention two strategies that build on concepts discussed in previous sections.
First, one can easily utilize geometric symmetries to construct additional bound entangled states if the initial one is a Magic Simplex state, i.e., of standard Bell-diagonal form (cf. Sec.~\ref{sec:magic_simplex}).
Second, as discussed in Sec.~\ref{sec:lorentz_invariance_BE}, one may find new bound entangled states by Lorentz boosting others.

\subsection{Construction from Unextendible Product Bases} \label{sec:UPB_construction}

One way to analytically construct PPT bound entangled states utilizes the concept of unextendible product bases, defined by:
\begin{mydef}[Unextendible Product Basis] \label{def:UPB}
    Given a multipartite Hilbert space $\hs = \bigotimes_{i=1}^{n} \hs_i$ of $n\geq 2$ parties and a set $S= \{ |\phi_j\rangle = |\phi^1_j\rangle \otimes |\phi^2_j\rangle \otimes ... \otimes |\phi^n_j\rangle \}_{j=1}^{m}$ consisting of $m$ mutually orthogonal product states in $\hs$, i.e., $\langle\phi_i | \phi_j\rangle = \delta_{ij}$ for all $i,j = 1, ..., m$.
    We denote by $\hs_S$ the subspace of $\hs$ spanned by the vectors in $S$, and by $\hs_S^\perp$ the complementary subspace of $\hs$ spanned by the set of vectors orthogonal to the ones in $S$.
    The set $S$ is called unextendible product basis (UPB) if $\dim(\hs_S^\perp)\neq 0$ and $\hs_S^\perp$ contains no product vectors.
\end{mydef}

In Ref.~\cite{bennett_unextendible_1999, divincenzo_unextendible_2003}, it is shown that such UPBs exist for specific dimensions. Given a UPB, one can construct PPT bound entangled states as follows \cite{bennett_unextendible_1999}:
\begin{mytheorem} \label{thm:BE_from_UPB}
    Let $S = \{ |\phi_j \rangle \}_{j=1}^{m}$ be a UPB on a Hilbert space $\hs = \bigotimes_{i=1}^{n} \hs_i$ with $\dim(\hs) = \Pi_{i=1}^n d_i = D$ and $1 + \sum_{i=1}^n (d_i - 1) \leq m \leq D-4$.
    Then the state
    \begin{align} \label{eq:BE_from_UPB}
        \rho_{UPB} = \frac{1}{D-m} \left( \id_D - \sum_{j=1}^m |\phi_j\rangle\langle \phi_j| \right) 
    \end{align}
    is bound entangled.
\end{mytheorem}
To see that this is true, note that by construction, the state $\rho_{UPB}$ is PPT because $\id_D$ is invariant under partial transposition and the states $|\phi_j \rangle$ are product states.
In particular, $(|\phi_j\rangle\langle\phi_j|)^\Gamma = |\Tilde{\phi}_j \rangle\langle\Tilde{\phi}_j|$, where $|\Tilde{\phi}_j\rangle$ are product states, and thus $\rho_{UPB}^\Gamma$ is again a quantum state.
Furthermore, Def.~\ref{def:UPB} entails that if $S$ is a UPB, the complementary subspace $\hs_S^\perp$ is non-empty and contains only entangled states.
Thus, by the range criterion (Thm.~\ref{thm:range_crit}), every state $\rho \in \linop(\hs_S^\perp)$ is entangled, and consequently $\rho_{UPB}$ is bound entangled.

The state $\rho_{UPB}$ in \eqref{eq:BE_from_UPB} is proportional to the projection operator onto the complementary subspace $\hs_S^\perp$.
As for the range of $m$ in the theorem, it is proven in Ref.~\cite{bennett_unextendible_1999} that a lower bound on the number of elements in a UPB is given by $m_{min} = 1 + \sum_{i=1}^n (d_i - 1)$, where $d_i = \dim(\hs_i)$ (see also Ref.~\cite{alon_unextendible_2001}).
Furthermore, it is clear that $\mathrm{rank}(\rho_{UPB}) = D-m$.
As no bound entangled states exist with rank smaller than 4 (cf. Thm.~\ref{thm:rank_BE}), a UPB cannot consist of more than $m_{max} = D-4$ product states.
Conversely, such maximal UPBs can be used to construct bound entangled states with minimal rank.

Even though Thm.~\ref{thm:BE_from_UPB} establishes a way of constructing bound entangled states, it does not completely characterize them in general, i.e., not all bound entangled states can be constructed via a UPB.
This is demonstrated in Ref.~\cite{bennett_unextendible_1999} by considering the bipartite bound entangled state from Ref.~\cite{horodecki_separability_1996} with Hilbert space dimensions $d_1 = 2$ and $d_2 = 4$.
As no UPBs exist in dimensions $2\times n$ for any $n$, such a qubit-ququart state cannot be constructed in the above way.
However, in special cases, UPBs enable a complete characterization of specific classes of bound entangled states.
As shown in Ref.~\cite{skowronek_three-by-three_2011}, for bipartite qutrit-qutrit systems, all rank-4 PPT bound entangled states can be obtained (up to local linear transformations) by maximal UPBs using Thm.~\ref{thm:BE_from_UPB}.

Lastly, we mention that Thm.~\ref{thm:BE_from_UPB} does not present the only way to construct bound entangled states from UPBs.
A different (but similar) approach is taken in Ref.~\cite{pittenger_unextendible_2003}.
The authors consider the convex hull of the states in a UPB and obtain a more involved yet stronger result than the one from Ref.~\cite{bennett_unextendible_1999}.
In particular, they construct further bound entangled states in the vicinity of $\rho_{UPB}$ from the same UPB.

\section{Conclusion and Outlook} \label{sec:outlook}

This review illuminates that many features of one of our most fundamental physical theories are still not fully understood.
The tensor-product structure of the Hilbert space underlying bipartite finite-dimensional quantum systems is the origin of non-classical correlations and, thus, entanglement.
In this context, the separability problem, i.e., distinguishing separable from non-separable quantum states, is one of the main open questions.
Solving this is shown to be NP-hard for general quantum systems of dimension $d>6$.
At its bottom lies the notion of local quantum operations and classical communications (LOCC) as entangled states are precisely those that can not be created from separable states by LOCC alone.
Even though we currently lack a good way to detect and quantify entanglement, it is clear that some states are more entangled than others in the sense that they exhibit stronger non-classical correlations.
This is especially true when considering probabilistic mixtures of pure quantum states, i.e., mixed states, which allows one to single out maximally entangled pure states from the rest.
One question concerning the production of strongly entangled states via LOCC is: Given multiple copies of some weakly entangled states, is it possible to manipulate them by LOCC to obtain fewer but more strongly entangled states?
Such a procedure is called entanglement distillation.
It is proven that a necessary condition for the distillability of a quantum state is that it violates the PPT criterion.
Entangled states that can not be distilled to maximally entangled states are called bound entangled, and they naturally arise in the study of distillation.
They seem to play a crucial role in our way to a full understanding of entanglement and its utilization as a resource.

Despite extensive research, the exact role of bound entanglement in connection to LOCC, entanglement distillation, and the PPT criterion still needs to be determined.
On the positive side, any new insight into bound entangled states will likely also lead to advances regarding the separability problem and vice versa.
Many possible approaches to tackling this issue are summarized in this review.
Thus, studying bound entangled states is crucial for fully understanding entanglement and the general non-classical nature of quantum physics.
Bound entangled states show limited use in currently known quantum technologies.
Nonetheless, further investigation into the unique properties of bound entanglement may lead to the conception of specialized tasks only realizable with those states.

Setting aside these general motivations for the ongoing study of bound entanglement, the rest of this section will review the most relevant open questions from the previous sections.
It thus presents a more detailed outlook on potential future research.
Some questions that have found a definitive answer and are discussed in this review are collected in Tbl.~\ref{tab:Summary}.

The current state of research concerning the detection of (bound) entanglement is summarized in Sec.~\ref{sec:free_vs_bnd} and Sec.~\ref{sec:detection of BE}.
To date, the only sufficient criterion for the undistillability of entangled states is the PPT criterion (Thm.~\ref{thm:PPT-crit}).
This poses the question of whether NPT bound entangled states exist, a topic discussed in more detail at the end of this section.
Sec.~\ref{sec:detection of BE} emphasizes that almost all the presented entanglement criteria do not follow a strict hierarchy regarding the number of entangled states they detect.
Hence, a universal criterion subsuming most of them has yet to be discovered.
Finding such a universal criterion is generally difficult as the mathematical structure of the known criteria is often very different.
One possible way forward is investigating the relation of bound entangled states (and their properties like entropic quantities, mixedness, etc.) to the criteria for detecting them.
Concerning entanglement witnesses presented in Sec.~\ref{sec:EWs} and Sec.~\ref{POVMWitnesses}, discovering an efficient way to find non-decomposable witnesses would lead to a novel method of constructing bound entangled states.
One possible approach is via entanglement witnesses based on mutually unbiased bases (MUBs).
The exact relation of MUBs and general non-decomposable witnesses is still not fully understood (cf. Sec.~\ref{sectionsubsets}).
Furthermore, non-decomposable witnesses can be constructed from non-decomposable PNCP maps via the Choi-Jamiołkowski isomorphism, possibly leading to the discovery of novel bound entangled states with interesting properties (cf. Ref.~\cite{korbicz_structural_2008}).
Compared to other methods, entanglement witnesses constitute observables, allowing the detection of bound entangled states in experiments.
As experimentally more accessible bound entangled states may be found using such a construction, verifying and subsequently utilizing bound entanglement in experiments could become easier.

Transitioning to the various characteristics of bound entangled states discussed in Sec.~\ref{charOfBE}, one fundamental question concerns the frequency of these states.
For a long time, it was argued that bound entanglement is a rare phenomenon.
More recent studies show that the relative volume of bound entangled states in the state space is not zero (cf. Sec.~\ref{sec:frequency_bound}), and hence, bound entanglement is not negligible in general.
However, any volume calculation presumes the choice of a probability measure for the state space, yet no unique canonical measure has been found to date.
Thus, a definitive relative volume determination is impossible.
Investigating candidates for such a canonical measure might yield insight into bound entanglement and the geometry of the quantum state space.
Additionally, numerical investigations suggest that PPT entanglement decreases approximately exponentially, independent of the chosen measure.
Further insight into different state space measures could also clarify why this seems to be the case.
As for the frequency of bound entanglement in the subset of Bell-diagonal states, it is remarkable that the largest share of bound entanglement can be found for the standard Bell basis (i.e., the Magic Simplex; cf. Sec.~\ref{sec:magic_simplex}).
The standard Bell basis possesses additional algebraic properties compared to non-standard Bell bases.
Consequently, one might conjecture that bound entanglement is linked to the underlying symmetries of the chosen basis.
Investigating this specific family of states might help to understand the set of bound entangled states in the total state space.

In Sec.~\ref{sec:BE_and_info_processing}, it is elaborated that bound entangled states do not offer an advantage over separable states for most known quantum information processing tasks.
The security in some networks utilizing secure key generation and the activation of bound entanglement are two exceptions.
The latter usually refers to scenarios where limited access to free entanglement with an additional supply of bound entangled states performs better in some quantum information protocols than the free or bound entangled states would allow on their own.
Known examples include activating a quantum state's teleportation capability and its distillability.
Activating the $1$-distillability of any NPT state by PPT bound entangled states has very interesting implications if NPT bound entangled states would exist.
In this case, one could obtain maximally entangled qudits from only bound entangled states and LOCC.
Hence, the existence of NPT bound entanglement implies the non-additivity of the distillable entanglement.

Including the distillable entanglement, some of the most relevant entanglement measures that attempt to quantify entanglement as a quantum resource are introduced in Sec.~\ref{sec:BE_and_ent_measures}. As any such measure is generally very hard to evaluate for mixed states, precise results for bound entangled states are rare. Crucially, however, a fundamental irreversibility in (catalytic) LOCC operations is proven to exist via the difference between the distillable entanglement and the entanglement cost for bound entangled states. Bound entangled states require a nonzero amount of maximally entangled states for their formation, which, in turn, cannot be extracted by entanglement distillation. 
This irreversibility relates to, or may be, the main reason for the difficulties of defining a general entanglement measure quantifying entanglement as a resource. No such measure is known, and the use of a specific measure is usually motivated by a corresponding operational task. Finding and relating entanglement measures to operational tasks contributes to the identification of the most relevant properties they should satisfy. Investigating properties like (sub-/super-) additivity, asymptotic continuity, or lockability of proposed measures contributes to such operational interpretations. As bound entanglement is, by definition, related to entanglement distillation, the distillable entanglement is certainly a measure of interest. Evaluating and comparing different measures for certain NPT and PPT entangled states, such as Werner or Bell-diagonal states, may provide information about the entanglement structure and properties of bound entangled states. As these computations are generally not feasible, finding stricter bounds for entanglement measures, like the Rains bound ~\cite{rains_semidefinite_2001}, is of crucial importance. While any other entanglement measure provides an upper bound for the distillable entanglement, lower bounds can, for instance, be found by analyzing Bell-diagonal states with respect to measures presented in this review or others like squashed entanglement~\cite{christandl_squashed_2004}. 
Even though irreversibility is present in the LOCC paradigm, other paradigms beyond LOCC have recently been proposed in which reversibility may be achieved \cite{regula_reversibility_2024, lami_distillable_2023}. In contrast, however, the paradigms of non-entangling operations and PPT operations have been shown to exhibit the irreversibility related to bound entanglement \cite{wang_irreversibility_2017, lami_no_2023}. Notably, irreversibility under non-entangling operations distinguishes entanglement from other quantum resources (cf. Ref.~\cite{lami_no_2023}). Since irreversibility lies at the heart of bound entanglement, investigations that explore this phenomenon under more general non-global paradigms than LOCC may contribute to our understanding of this type of entanglement. 

One further question with direct physical motivation concerns bound entanglement in relativistic settings.
This is particularly intriguing in the context of the non-local features of bound entangled states (cf. Sec.~\ref{sec:BIandBE}) and the local nature of special relativity.
Examples show that Lorentz boosts can change the entanglement class of a state from free entangled to bound entangled or separable.
However, a universal understanding is still missing.
Such a comprehension could reveal crucial insights into the interplay of quantum physics and relativity and open new possibilities for constructing bound entangled states analytically.

Regarding general bound entanglement construction schemes described in Sec.~\ref{sec:construction}, much research still needs to be done.
Nowadays, numerically generating a plethora of bound entangled states with certain properties, e.g., inside the Magic Simplex, is possible.
On the other hand, the main analytic method with some generality is via unextendible product bases (UPBs).
As any construction necessarily exploits some property of bound entangled states, it will be crucial to investigate further constructions to gain a deeper understanding of the underlying physics and mathematics.
So far, all schemes we know can only construct PPT bound entangled states.

This brings us to the central question of the existence of NPT bound entangled states, i.e., potential undistillable entangled states that violate the PPT criterion (Thm.~\ref{thm:PPT-crit}).
See Fig.~\ref{fig:FE_BE_Set_relations} for a visualization of the problem.
In this regard, the main open challenge is proving or disproving whether the PPT criterion is also necessary for the undistillability of bipartite quantum systems with general Hilbert space dimensions $d_A \times d_B$ (cf. Thm.~\ref{thm:dist_implies_NPT}).
A consequence of Thm.~\ref{thm:qubit_distillation} and Thm.~\ref{thm:PPT-crit_qubit_qubit} is that this is true for qubit-qubit ($2\times 2$) systems.
In Ref.~\cite{divincenzo_evidence_2000}, it is furthermore shown that it extends to systems of Hilbert space dimension $2\times d$ with $d>2$.
The difficulty of investigating undistillability lies in the definition of distillability itself (cf. Sec.~\ref{sec:ent_dist}).
A state is said to be distillable if, for some $n>0$, there is a LOCC protocol transforming $n$ identical copies of the state to some maximally entangled state.
An equivalent characterization of distillability uses any entangled qubit state as the target state of distillation (cf. Def.~\ref{def:n-distillability} and Thm.~\ref{thm:qubit_distillation}).
However, for every $d>2$ and $n>0$, it is shown that there are states with Hilbert space dimensions $d^2\times d^2$ that are not $n$-distillable, but are $(n+1)$-distillable (cf. Thm.~\ref{thm:watrous_n_distillability}).
Thus, it is not enough to check that a state is $n$-undistillable for some finite $n>0$ to establish that it is generally undistillable.
Besides this complication, there is one major simplification to the problem of NPT bound entanglement.
Ref.~\cite{horodecki_reduction_1999} shows that to solve the general problem, it is sufficient to prove whether NPT bound entangled states exist among the one-parameter family of Werner states \cite{werner_quantum_1989}.
The reason for this is that every NPT state can be transformed by LOCC to an NPT Werner state.
However, proving the undistillability of these highly symmetric Werner states has not been successful yet.
There have been numerous attempts to characterize them.
For every $n>0$, Werner states are shown to be $n$-undistillable for a certain parameter range \cite{divincenzo_evidence_2000, dur_distillability_2000}.
Whether this parameter range vanishes as $n\rightarrow\infty$ is not known.
Furthermore, it is conjectured that every 1-undistillable Werner state is also 2-undistillable \cite{dokovic_two-distillable_2016}, and numerical analyses show that such states are probably also 3-undistillable \cite{dur_distillability_2000}.
Further approaches and references can be found in Ref.~\cite{pankowski_few_2010, clarisse_distillability_2006}.
We mention in passing that, due to the Choi-Jamiołkowski isomorphism, the problem can also be restated in terms of positive maps \cite{divincenzo_evidence_2000, muller-hermes_positivity_2016}.
One major consequence of the existence of NPT bound entangled states is that it would imply the non-additivity and non-convexity of the distillable entanglement \cite{shor_nonadditivity_2001}.
This is exemplified by the activation of 1-distillability discussed in Sec.~\ref{sec:Activation_of_BE} (cf. also Ref.~\cite{vollbrecht_activating_2002}).
In particular, suppose $\rho_{PPT}$ is a PPT bound entangled universal activator state and $\rho_{NPT}$ is an NPT bound entangled state.
Using Def.~\ref{def:distillable_entanglement}, one would obtain $E_D(\rho_{NPT}\otimes \rho_{PPT}) > 0 = E_D(\rho_{NPT}) + E_D(\rho_{PPT})$.

In summary, bound entanglement is a unique form of quantum
entanglement that cannot be used for distillation. It
remains a complex and not fully grasped phenomenon, indicating
the limitations of our current understanding of these quantum
correlations.
The complexities associated with the detection, characterization, and
quantification of bound entanglement underscore the incompleteness of our
theoretical framework. The inability to fully quantify it
as a resource in quantum information processing tasks further
highlights the gaps in our knowledge, leaving the potential role of
bound entanglement in practical applications of quantum information
processing largely unexplored. 
Recent advancements in the detection and
characterization of bound entanglement in specific systems offer more
than just a glimmer of hope. Future research contributing to these
developments could potentially lead us towards 
a comprehensive understanding of entanglement and its effective
utilization in quantum technologies.
The journey towards unraveling the mysteries of bound entanglement
continues, paving the way for exciting discoveries in the realm of
quantum physics and quantum information theory.

\setlength{\arrayrulewidth}{1pt}
\renewcommand{\arraystretch}{2.0}
\begin{table} 
\begin{center}
\begin{tabular}{|p{11.5cm}|c|l|}
\hhline{|t===t|}
\multicolumn{3}{||c||}{\cellcolor{gray!30} \Huge{\textbf{Bound Entangled (BE) States}}}\\
\hhline{||=|=|=||}
\large{\textbf{Question}}& \large{\textbf{Answer}}& \large{\textbf{Section}}\\
\hhline{||=|=|=||}
Can BE states be distilled?& No& \ref{sec:basics}, \ref{sec:free_vs_bnd}\\
\hhline{||-|-|-||}
Does $n$-copy-undistillability for a state imply its undistillability? & No& \ref{sec:ent_dist}\\
\hhline{||-|-|-||}
Is BE invariant under LOCC?& No&\ref{sec:basics},  \ref{sec:free_vs_bnd}\\ 
\hhline{||-|-|-||}
Is there a single entanglement criterion that is strictly stronger than all others?& No& \ref{sec:detection of BE}\\
\hhline{||-|-|-||}
Can all entanglement witnesses detect BE states?& No& \ref{sec:decomposible}\\
\hhline{||-|-|-||}
Is a complete set of MUBs needed for the detection of BE states?& No& \ref{sec:MUBwitness}\\
\hhline{||-|-|-||}
Are BE states of measure zero in finite-dimensional Hilbert spaces?& No& \ref{sec:frequency_bound}\\
\hhline{||-|-|-||}
Are BE states useful as resources in quantum teleportation?& No& \ref{sec:BE_teleportation}\\
\hhline{||-|-|-||}
Are BE states useful in device-dependent schemes as an additional resource?& Yes& \ref{sec:Activation_of_BE}, \ref{sec:pureconversionBE} \\
\hhline{||-|-|-||}
Do all bound entangled states have zero entanglement of formation? & No& \ref{sec:BE_and_ent_measures}\\
\hhline{||-|-|-||}
Can BE be useful for quantum metrology?& Yes& \ref{sec:metrology}\\
\hhline{||-|-|-||}
Can BE states violate a Bell inequality?& Yes& \ref{sec:BIandBE}\\
\hhline{||-|-|-||}
Are BE states Lorentz invariant?& No& \ref{sec:lorentz_invariance_BE}\\
\hhline{||-|-|-||}
Have BE states been verified experimentally?& Yes& \ref{BEexperiment}\\
\hhline{||-|-|-||}
Is there an analytic construction for all BE states?& No& \ref{sec:construction}\\
\hhline{|b:===:b|}
\end{tabular}
\end{center}
\caption{This table summarizes some basic questions discussed in detail in this review. We have only selected dichotomous answers that always refer to all aspects mentioned in the respective sections.}
\label{tab:Summary}
\end{table}

\section*{Acknowledgments}
\addcontentsline{toc}{section}{Acknowledgments}
B.C.H. and C.P. acknowledge gratefully that this research was funded in whole,
or in part, by the Austrian Science Fund (FWF) project P36102-N (Grant DOI:  10.55776/P36102). For the purpose of open access, the author has applied a CC BY public copyright license to any Author Accepted Manuscript version arising from this submission. 
The funder played no role in study design, data collection, analysis and interpretation of data, or the writing of this manuscript.

\section*{Author Contributions Statement}
All authors contributed equally to this work.

\section*{Competing Interests}
 All authors declare no financial or non-financial competing interests.

\pagebreak

\appendix
\renewcommand{\thesubsection}{\Alph{subsection}}

\section*{Appendix}
\addcontentsline{toc}{section}{Appendix}

\subsection{Minimal Fidelity of Teleportation} \label{app:min_fid_of_transmission}

Here, we briefly recap how the measure-and-prepare teleportation protocol (cf. Ref.~\cite{linden_bound_1999}) works and calculate its fidelity of transmission.
The general setup is identical to the one described at the beginning of Sec.~\ref{sec:BE_teleportation}.
In the first step of the protocol, Alice measures her unknown quantum state $|\phi\rangle$ in some orthonormal basis $\{|v_i\rangle \}$.
The probability of Alice obtaining outcome $i$ is given by $p_i = |\langle v_i | \phi \rangle|^2$.
Next, she classically transmits the information about the measurement basis and the outcome $i$ to Bob, who prepares the respective quantum state $|v_i\rangle$.
This will serve as his guess as to what $|\phi\rangle$ was.
Hence, taking into account all possible measurement results $i$, Bob prepares the mixed state $\sigma_\phi = \sum_{i=0}^{d-1} p_i \, |v_i\rangle\langle v_i|= \sum_{i=0}^{d-1}  |\langle v_i | \phi \rangle|^2 \, |v_i\rangle\langle v_i|$.
The corresponding fidelity of transmission \eqref{eq:fid_of_transmission} is thus given by
\begin{align}
    f  = \int \diff\phi\, F(\sigma_\phi, |\phi\rangle ) 
    = \int \diff\phi\, \langle \phi | \sigma_\phi |\phi\rangle 
    = \sum_{i=0}^{d-1}\, \int \diff\phi\, |\langle v_i | \phi \rangle|^4 \;.
\end{align} 
To evaluate this integral, we write $|\phi\rangle = U_\phi |0\rangle$ and $|v_i\rangle = V_i |0\rangle$, where $U_\phi$ and $V_i$ are unitary.
This allows us to compute
\begin{align} \label{eq:appendix_min_fid_of_transmission}
    f = \sum_{i=0}^{d-1}\, \int \diff\phi\, |\langle v_i | \phi \rangle|^4 
     = \sum_{i=0}^{d-1}\, \int \diff\phi\, |\langle 0 | V_i^\dagger U_\phi |0 \rangle|^4 
     = \sum_{i=0}^{d-1}\, \int \diff\phi\, |\langle 0 | \Tilde{U}^\phi_i |0 \rangle|^4 
     = \sum_{i=0}^{d-1}\, \frac{2}{d(d+1)} 
     = \frac{2}{(d+1)} \,,
\end{align}
where we defined $\Tilde{U}^\phi_i:= V_i^\dagger U_\phi$ in the third equality.
The fourth equality follows directly from the explicit calculation in Ref.~\cite{spengler_composite_2012} using the composite parametrization of unitary matrices.

As the measure-and-prepare teleportation protocol is applicable independent of the resource state shared by Alice and Bob, the minimal fidelity of transmission attainable in any scenario is lower bound by \eqref{eq:appendix_min_fid_of_transmission}.

\subsection{Protocol for Activation of Bound Entanglement} \label{app:protocol_activation_BE}

Here, we reproduce the example of a free entangled qutrit-qutrit state from Ref.~\cite{horodecki_general_1999}, which is subsequently used in the activation protocol from Ref.~\cite{horodecki_bound_1999}.

Consider the free entangled qutrit-qutrit state
\begin{align} \label{eq:activation_FE_state}
    \rho^\mathrm{FE} = \mathcal{F}_0\, P_{0,0} + \frac{1-\mathcal{F}_0}{3} \, \sum_{i=0}^2 P_{i,1} \;,
\end{align}
with $0<\mathcal{F}_0<1$ being its singlet fraction \eqref{eq:singlet_fraction}, and $P_{k,l}$ are the Bell state projectors~\eqref{eq:bell_states}.
Note that this constitutes a mixture of the maximally entangled qutrit state~\eqref{eq:max_ent_state_omega_00} and a separable subgroup (line) state~\eqref{eq:subgroup_state}.
Therefore, it is a Magic Simplex state (cf. Def.~\ref{def:magic_simplex}).
That this state is free entangled can be shown by considering the local projection of $\rho^\mathrm{FE}$ onto the subspace of two degrees of freedom $(|0\rangle\langle 0| + |1\rangle\langle 1|)\otimes(|0\rangle\langle 0| + |1\rangle\langle 1|)$, resulting in the effective qubit state
\begin{align}
    \rho^\mathrm{FE}_{2\times 2} = \frac{\mathcal{F}_0}{1+\mathcal{F}_0}\left(|00\rangle\langle 00| + |00\rangle\langle 11| + |11\rangle\langle 00| + |11\rangle\langle 11| \right) +\frac{1-\mathcal{F}_0}{1+\mathcal{F}_0}\; |01\rangle\langle 01| \;.
\end{align}
This state is NPT for $0<\mathcal{F}_0<1$, and thus according to Thm.~\ref{thm:qubit_distillation} and Thm.~\ref{thm:PPT-crit_qubit_qubit} it is distillable.
Hence, also $\rho^\mathrm{FE}$ is distillable by the argument presented in Sec.~\ref{sec:ent_dist}.
However, as proven in Ref.~\cite{horodecki_general_1999}, the singlet fraction obtainable with nonzero probability by LOCC and a \textit{single} copy of \eqref{eq:activation_FE_state} is upper bounded by some constant $\mathcal{F}_{max}<1$.
This even holds when LOCC protocols with a success probability tending to zero are considered.
By Thm.~\ref{thm:teleportation_and_singlet_fraction}, faithful teleportation is thus impossible with only a single copy of $\rho^\mathrm{FE}$ because the maximally achievable fidelity of transmission is upper bounded by $f_{max} = (3 \,\mathcal{F}_{max}+1)/4 <1$.

For the specific activation protocol from Ref.~\cite{horodecki_bound_1999}, the authors defined the activation state
\begin{align} \label{eq:activation_BE_state}
    \rho^\mathrm{BE} = \frac{2}{7}\, P_{0,0} + \frac{\alpha}{21}\, \sum_{i=0}^2 P_{i,1} + \frac{5-\alpha}{21}\, \sum_{i=0}^2 P_{i,2} \;.
\end{align}
Again, this is a mixture of the maximally entangled qutrit state~\eqref{eq:max_ent_state_omega_00} and two distinct subgroup (line) states~\eqref{eq:subgroup_state}.
Thus, also \eqref{eq:activation_BE_state} constitutes a Magic Simplex state.
It is free entangled for $0\leq \alpha <1$ and $4<\alpha \leq 5$, bound entangled for $1\leq \alpha<2$ and $3<\alpha\leq 4$, and separable for $2\leq \alpha \leq 3$.
Separability follows from an explicit decomposition of \eqref{eq:activation_BE_state} into a mixture of separable states, the free entanglement can be shown as for \eqref{eq:activation_FE_state} above, and the bound entangled regions are detected by the realignment criterion (cf. Sec.~\ref{sec:realignment}).

Suppose now that Alice and Bob share one free entangled state \eqref{eq:activation_FE_state} and additionally multiple copies of the bound entangled state \eqref{eq:activation_BE_state} with $3<\alpha\leq4$.
They are furthermore both able to apply a generalized $d$-dimensional XOR-gate locally.
This quantum gate is represented by the unitary matrix
\begin{align} \label{eq:activation_XOR}
    U_{XOR} = \sum_{i,j=0}^{d-1} |i\rangle\langle i| \otimes |j + i\rangle\langle j| \;.
\end{align}
We call the unaltered first subsystem the control qudit and the altered second subsystem the target qudit.
Alice and Bob can increase the singlet fraction of the free entangled state $\rho^\mathrm{FE}$ by LOCC as follows:
\begin{enumerate}
    \item
    They each apply the XOR-gate to their parts of the free entangled state $\rho^\mathrm{FE}$ and the bound entangled state $\rho^\mathrm{BE}$. The free (bound) entangled qudit is the control (target).

    \item
    Alice and Bob each measure their respective target qudits in the computational basis and compare their results.
    If they differ, the procedure fails, and the free entangled resource is lost.
    However, if the results coincide, the protocol succeeds, and only the target qudits are discarded.
    The free entangled state now has a higher singlet fraction $\mathcal{F}_1 > \mathcal{F}_0$.
\end{enumerate}
It is straightforward to compute the success probability of step two to be
\begin{align} \label{eq:activation_probability}
    P_{\mathcal{F}_0\rightarrow\mathcal{F}_1} = \frac{2\mathcal{F}_0+(1-\mathcal{F}_0)(5-\alpha)}{7} \;.
\end{align}
After discarding the target qudits, the control pair is still of the form \eqref{eq:activation_FE_state}, but with an updated singlet fraction given by
\begin{align}
    \mathcal{F}_1 = \frac{2\mathcal{F}_0}{2\mathcal{F}_0+(1-\mathcal{F}_0)(5-\alpha)} \;.
\end{align}
A simple check confirms that for $3 < \alpha \leq 5$ one has $\mathcal{F}_1 > \mathcal{F}_0$ for any initial $0<\mathcal{F}_0<1$.
This includes the bound entangled range $3<\alpha\leq4$.

The above procedure can be repeated as often as Alice's and Bob's supply of bound entangled states allows.
The probability of $n$ consecutive successes $P_n$ is given by the product of probabilities \eqref{eq:activation_probability} with updated $\mathcal{F}_n$ for each round.
This probability may be low, but it is nonzero.
For each repetition, a bound entangled state is consumed in exchange for an increase in the singlet fraction of the free entangled state.
This can be interpreted as transferring the entanglement from a bound entangled state, which cannot be accessed directly, to a free entangled state, where it can be used.
By this procedure, the singlet fraction $\mathcal{F}_n$ can be boosted arbitrarily close to 1 for large $n$, provided Alice and Bob share enough bound entangled states.
Consequently, $\mathcal{F}_n > \mathcal{F}_{max}$ for some $n$, thus exceeding the maximal singlet fraction attainable by \eqref{eq:activation_FE_state} alone.
Interpreted as a distillation scheme, this suggests that the distillable entanglement (cf. Def.~\ref{def:distillable_entanglement}) may not be additive \cite{bennett_mixed_1996}.

By virtue of Thm.~\ref{thm:teleportation_and_singlet_fraction} this can be viewed as follows:
Generally, bound entangled states are not useful for teleportation as they do not outperform any classical teleportation protocol.
This also holds true for the case of \eqref{eq:activation_BE_state} with $3<\alpha\leq4$.
Furthermore, the free entangled state \eqref{eq:activation_FE_state} cannot achieve a fidelity of transmission \eqref{eq:fid_of_transmission} arbitrarily close to 1.
However, if these two resources are combined, one can ``activate'' the bound entanglement contained in \eqref{eq:activation_BE_state} by increasing the maximal fidelity of transmission of \eqref{eq:activation_FE_state}.
In this way, bound entangled states can improve the performance in specific teleportation scenarios.

Lastly, for completeness, we add that the bound entanglement of \eqref{eq:activation_BE_state} in the range $1\leq \alpha<2$ can also be activated, albeit with slight modifications.
This requires the free entangled resource state
\begin{align}
    \Tilde{\rho}^\mathrm{FE} = \mathcal{F}_0\, P_{0,0} + \frac{1-\mathcal{F}_0}{3} \sum_{i=0}^2 P_{i,2} \;.
\end{align}
The probability of equal outcomes in step two is
\begin{align} 
    \Tilde{P}_{\mathcal{F}_0\rightarrow \mathcal{F}_1} = \frac{2 \mathcal{F}_0 + (1-\mathcal{F}_0) \alpha}{7} \;,
\end{align}
after which the singlet fraction changes to
\begin{align}
    \Tilde{\mathcal{F}}_1 = \frac{2\mathcal{F}_0}{2\mathcal{F}_0+ (1-\mathcal{F}_0)\alpha} \;.
\end{align}
It holds that $\Tilde{\mathcal{F}}_1>\mathcal{F}_0$ for $0\leq \alpha <2$ and therefore in particular in the bound entangled region $1\leq \alpha<2$.

\clearpage
\phantomsection
\addcontentsline{toc}{section}{References}
\printbibliography


\end{document}